\documentclass[conference,compsoc]{IEEEtran}
\IEEEoverridecommandlockouts
\usepackage{cite}
\usepackage{amsmath}
\usepackage{amssymb}
\usepackage{amsfonts}
\usepackage{amsthm}
\usepackage{algorithm}
\usepackage{algpseudocode}
\usepackage{graphicx}
\usepackage{textcomp}
\usepackage{xcolor}
\usepackage{soul}
\usepackage{enumitem}

\def\BibTeX{{\rm B\kern-.05em{\sc i\kern-.025em b}\kern-.08em
    T\kern-.1667em\lower.7ex\hbox{E}\kern-.125emX}}
\usepackage{enumitem}
\setlist[enumerate]{itemsep=0mm}
\usepackage{stackengine}
\usepackage{bm}
\usepackage{tikz}
\usepackage{filecontents}
\usepackage{multirow}
\usepackage{booktabs}
\usepackage{color}
\usepackage{algorithm}  
\usepackage{algorithmicx}  
\usepackage{subcaption}
\usepackage{caption}
\usepackage[normalem]{ulem}
\usepackage{threeparttable}
\usepackage{pifont}
\usepackage{mathtools}
\usepackage{hyperref}
\usepackage{bbding}
\usepackage{textcomp}
\makeatletter

\newcommand{\Rmnum}[1]{\expandafter\@slowromancap\romannumeral #1@}
\makeatother

\newcommand{\CD}{\ensuremath{\mathcal{D}}}

\newcommand{\CN}{\ensuremath{\mathcal{N}}}

\newcommand{\CR}{\ensuremath{\mathcal{R}}}
\newcommand{\CS}{\ensuremath{\mathcal{S}}}
\newcommand{\CT}{\ensuremath{\mathcal{T}}}
\newcommand{\CU}{\ensuremath{\mathcal{U}}}

\newcommand{\CW}{\ensuremath{\mathcal{W}}}
\newcommand{\CX}{\ensuremath{\mathcal{X}}}
\newcommand{\CY}{\ensuremath{\mathcal{Y}}}

\newcommand{\FB}{\ensuremath{\mathbf{B}}}
\newcommand{\FU}{\ensuremath{\mathbf{U}}}

\newcommand{\BR}{\ensuremath{\mathbb{R}}}
\newcommand{\BP}{\ensuremath{\mathbb{P}}}

\newcommand{\BZ}{\ensuremath{\mathbb{Z}}}

\newcommand{\rd}{\ensuremath{\mathrm{d}}}

\newcommand{\rad}{\ensuremath{\mathtt{RAD}}}

\definecolor{darkblue}{RGB}{48, 84, 151}
\definecolor{lightblue}{RGB}{222, 235, 246}
\definecolor{ao(english)}{rgb}{0.0, 0.5, 0.0}
\definecolor{applegreen}{rgb}{0.55, 0.71, 0.0}

\useunder{\uline}{\ul}{}

\DeclareMathOperator*{\argmax}{arg\,max}

\usepackage[english]{babel}
\newtheorem{theorem}{Theorem}
\newtheorem{lemma}{Lemma}

\newtheorem{proposition}{Proposition}

\newtheorem*{claim}{Claim}
\theoremstyle{definition}
\newtheorem{definition}{Definition}  %[section]

\begin{document}

\title{Text-CRS: A Generalized Certified Robustness Framework against Textual Adversarial Attacks\thanks{\textsuperscript{\S}Corresponding Authors: Yuan Hong and Zhongjie Ba}
%\BW{CETRS or TextRS?} \xy{I have changed to Text-CRS} \\
}

% \author{\IEEEauthorblockN{Xinyu Zhang}
% \IEEEauthorblockA{\textit{Zhejiang University, University of Connecticut} \\
% % \textit{name of organization (of Aff.)}\\
% xinyuzhang53@zju.edu.cn}
% \and
% \IEEEauthorblockN{Hanbin Hong}
% \IEEEauthorblockA{\textit{University of Connecticut} \\
% hanbin.hong@uconn.edu}
% \and
% \IEEEauthorblockN{Yuan Hong}
% \IEEEauthorblockA{\textit{University of Connecticut} \\
% yuan.hong@uconn.edu}
% }

% \ddagger\#\Letter
\author{
\IEEEauthorblockA {Xinyu Zhang\textsuperscript{$*\dagger$}, Hanbin Hong\textsuperscript{$\dagger$}, Yuan Hong\textsuperscript{$\dagger$}\textsuperscript{\S}, Peng Huang\textsuperscript{$*$}, Binghui Wang\textsuperscript{$\ddagger$}, Zhongjie Ba\textsuperscript{$*$}\textsuperscript{\S}, Kui Ren\textsuperscript{$*$}}
\IEEEauthorblockA{
\textsuperscript{$*$}Zhejiang University, \textsuperscript{$\dagger$}University of Connecticut, \textsuperscript{$\ddagger$}Illinois Institute of Technology \\
\{xinyuzhang53, penghuang, zhongjieba, kuiren\}@zju.edu.cn, \{hanbin.hong, yuan.hong\}@uconn.edu, bwang70@iit.edu
}
}

\maketitle

\begin{abstract}

The language models, especially the basic text classification models, have been shown to be susceptible to textual adversarial attacks such as synonym substitution and word insertion attacks. To defend against such attacks, a growing body of research has been devoted to improving the model robustness. However, providing provable robustness guarantees instead of empirical robustness is still widely unexplored. In this paper, we propose Text-CRS, a generalized certified robustness framework for natural language processing (NLP) based on randomized smoothing. To our best knowledge, existing certified schemes for NLP can only certify the robustness against $\ell_0$ perturbations in synonym substitution attacks. Representing each word-level adversarial operation (i.e., synonym substitution, word reordering, insertion, and deletion) as a combination of permutation and embedding transformation, we propose novel smoothing theorems to derive robustness bounds in both permutation and embedding space against such adversarial operations. To further improve certified accuracy and radius, we consider the numerical relationships between discrete words and select proper noise distributions for the randomized smoothing. Finally, we conduct substantial experiments on multiple language models and datasets. Text-CRS can address all four different word-level adversarial operations and achieve a significant accuracy improvement. We also provide the first benchmark on certified accuracy and radius of four word-level operations, besides outperforming the state-of-the-art certification against synonym substitution attacks. 
%
% \footnote{Code is available at \url{https://github.com/xxx/xxx}}

% Furthermore, we evaluate and analyze the universality of our theorem. 
% Considering the numerical relationship among discrete words for the first time, our framework is more practical and effective than previous synonym substitution certification.
% We are the first to consider the numerical relationship among discrete words by incorporating the embedding space into the certification. 

% \begin{IEEEkeywords}
% component, formatting, style, styling, insert
% \end{IEEEkeywords}

\end{abstract}

%In addition to less researched synonym substitution operations, we supplement all other potential operations in word-level textual adversarial attacks, which can be consolidated and categorized into three fundamental operations (i.e., rearrangements, insertions, and deletions). 
%In order to address the significant absolute differences caused by the operations, we partition the input space into permutation and embedding spaces and model each word-level operation as a combination of permutation and embedding transformation. 

% (Moreover, we provide theoretical robustness bounds for the models trained with CERT and prove that our robustness bounds are tight.)
% Moreover, we propose a set of empirically enhanced training methods to mitigate noise-induced accuracy degradation.

%Note that prior works only provide certified robustness against synonym substitution operations. 

\vspace{-0.1in}
\section{Introduction}
\vspace{-0.1in}
% With recent advances in natural language processing (NLP), large language models (e.g., ChatGPT and chatbot) are increasingly common use. As an essential and fundamental part of language models, the text classification task has a wide range of applications, such as content moderation, sentiment analysis, fraud detection, and spam filtering. However, text classification models are vulnerable to adversarial attacks, deliberately manipulating the model output by modifying the input text. \cite{RealWorl22:online}

With recent advances in natural language processing (NLP), large language models (e.g., ChatGPT \cite{chatgpt:online} and Chatbots \cite{zhang2020dialogpt,shuster2020dialogue,roller2021recipes}) have become increasingly popular and widely deployed in practice. Wherein, text classification plays an important role in language models, and it has a wide range of applications, including content moderation, sentiment analysis, fraud detection, and spam filtering~\cite{9TextCla71:online,AIDocume25:online}. Nevertheless, text classification models are vulnerable to word-level adversarial attacks, which imperceptibly manipulate the words in input text to alter the output~\cite{ren2019generating,maheshwary2021generating,jin2020bert,garg2020bae,lee2022query,li2021contextualized,feng2018pathologies}. These attacks can be exploited maliciously to spread misinformation, promote hate speech, and circumvent content moderation~\cite{wu2019misinformation}.

% \BW{To mitigate such attacks on text classification models, a series of empirical defenses have been proposed~\cite{pruthi2019combating,zhufreelb,jones2020robust,zhou2021defense,dongtowards,bao2021defending,huang2022word,yang2022robust,yoo2022detection,moon2022gradmask,mosca2022suspicious,azizi2021t,liu2022piccolo,zhang2022improving,qi2021onion,minh2022textual,DBLP:conf/iclr/MiyatoDG17}, where the state-of-the-art is based on adversarial training~\cite{DBLP:conf/iclr/MiyatoDG17} or data augmentation~\cite{jones2020robust,zhou2021defense}. 
% % For instance, 
% % Miyato et al.~\cite{DBLP:conf/iclr/MiyatoDG17} proposed to apply  adversarial training to the text domain. 
% % Zhou et al.~\cite{zhou2021defense} augments training data by creating virtual texts that mix the embedding of the original word with its synonyms' embedding via Dirichlet sampling. 
% However, these methods are shown to be broken by adaptive adversaries~\cite{jin2020bert}. 
% % For instance, \cite{} show that the defense 
% To end the cat-and-mouse game, in the past three years, researchers start focusing on certified defenses~\cite{jia2019certified,huang2019achieving,ye2020safer,zeng2021certified,wang2021certified}, i.e., defenses with provable robustness guarantees against the worst-case attacks. Among different certified defense...}

To defend against such attacks, numerous techniques have been proposed to improve the robustness of language models, especially for text classification models. For instance, adversarial training~\cite{goodfellow2014explaining,madrytowards,dongtowards} retains the model using both clean and adversarial examples to enhance the model performance; feature detection~\cite{yoo2022detection,mozes2021frequency} checks and discards detected adversarial inputs to mitigate the attack; and input transformation~\cite{wang2021natural, yang2022robust} processes the input text to eliminate possible perturbations. However, these empirical defenses are only effective against specific adversarial attacks and can be broken by adaptive attacks~\cite{jin2020bert}.
% (without formal guarantees). New adversarial examples in unseen attacks or adaptive attacks may explicitly break the defense.
% Thus, the defender becomes caught in a cycle with the adversary.
% robustness and generalization of ML models can be improved by crafting high-quality adversaries and including them in the training data.

One promising way to win the arms race against unseen or adaptive attacks %escape the attack-defense cycle 
is to provide provable robustness guarantees for the model. This line of work aims to develop certifiably robust algorithms that ensure the model's predictions are stable over a certain range of adversarial perturbations. Among different certified defense methods, randomized smoothing~\cite{cohen2019certified,zhang2020black,li2022sok} does not impose any restrictions on the model architecture and achieves acceptable accuracy for large-scale datasets. This method injects random noise, sampled from a smoothing distribution, into the input data during training to smoothen the classifier. The smoothed classifier will make a consistent prediction for a perturbed test instance (with noise) as the original class label. %under the smoothing distribution and ultimately outputs the most probable results. 
% Although randomized smoothing methods have been successfully applied to image inputs, applying them to text inputs remains fairly challenging.
Despite the successful application of randomized smoothing in protecting vision models, applying these methods to safeguard language models remains fairly challenging.
% We summarize four challenges for certifying the text classification model:
% While existing works on adversarial examples have succeeded in the image domain, it is still challenging to deal with text data due to its discrete nature. 

\begin{table*}[]
\centering
\setlength\tabcolsep{3pt}
% \scriptsize
\caption{Comparison of certified defense methods for NLP robustness against textual adversarial attacks.}
\vspace{-0.07in}
\resizebox{\linewidth}{!}{
\begin{threeparttable}
\begin{tabular}{@{}lcccccccc@{}}
\toprule
\multirow{2}{*}{Method} & \multirow{2}{*}{Model architecture} & \multicolumn{4}{c}{Adversarial operations (smoothing distribution / $\ell_p$ perturbation)} & Certified & Uni- & Accuracy \\
%\cline{3-6}
 &  & Substitution & Reordering & Insertion & Deletion & radius / $\rad$ & versality & (large-scale data) \\
\midrule
IBP-trained~\cite{jia2019certified} & LSTM/Att. layer & \ding{51}& \ding{51}& \ding{51}& \ding{51} & \ding{55} & \ding{55} & Low \\
POPQORN~\cite{ko2019popqorn} & RNN/LSTM/GRU & \ding{51}& \ding{51}& \ding{51}& \ding{51} & \ding{55} & \ding{55} & Low \\
Cert-RNN~\cite{du2021cert} & RNN/LSTM & \ding{51}& \ding{51}& \ding{51}& \ding{51} & \ding{55} & \ding{55} & Low \\
DeepT~\cite{bonaert2021fast} & Transformer & \ding{51}& \ding{51}& \ding{51}& \ding{51} & \ding{55} & \ding{55} & Low \\
SAFER~\cite{ye2020safer} & Unrestricted & \ding{51} (Uniform / $\ell_0$) & \ding{55} & \ding{55} & \ding{55} & \ding{51} Practical & \ding{55} & High \\
% WordDP~\cite{wangw2021certified} & Unrestricted & \ding{51} (Uniform) & \ding{55} & \ding{55} & \ding{55} & \ding{51} (1) & \ding{55} & High \\
RanMASK~\cite{zeng2021certified} & Unrestricted & \ding{51} (Uniform / $\ell_0$) & \ding{55} & \ding{55} & \ding{55} & \ding{51} & \ding{55} & High\\
CISS~\cite{zhao2022certified} & Unrestricted & \ding{51} (Gaussian / $\ell_2$) & \ding{55}& \ding{55}& \ding{55} & \ding{51} & \ding{55} & High \\
\midrule
Text-CRS (Ours) & Unrestricted & \ding{51} (Staircase / $\ell_1$) & \ding{51} (Uniform / $\ell_1$) & \ding{51} (Gaussian / $\ell_2$) & \ding{51} (Bernoulli / $\ell_0$) & \ding{51} Practical & \ding{51} & $>$  SOTA \\
\bottomrule
\end{tabular}
\begin{tablenotes}
    \item {1. The model architectures applicable to the first four methods have size restrictions, i.e., the number and size of layers cannot be too large.}
    \item {2. "Practical" means that the certified radius can correspond to the $\rad$-word level of perturbation. (We propose four practical certified radii.)}
    % \item {3. We introduce four certified radii under different smoothing distributions.}
\end{tablenotes}
\end{threeparttable}
}
\label{tab:intro}
\vspace{-0.23in}
\end{table*}

% And because they do not consider the numerical relationships between words
First, due to the discrete nature of the text data, numerical $\ell_1$ or $\ell_2$-norms cannot be directly used to measure the distance between texts. Without considering word embeddings,
% (e.g., may need encoding of words before distance calculation). 
previously certified defenses in the NLP domain, such as SAFER~\cite{ye2020safer} and WordDP~\cite{wang2021certified2}, are limited to certifying robustness against $\ell_0$ perturbations generated by synonym substitution attacks. Also, their assumption of uniformly distributed synonyms is impractical, leading to relatively low certified accuracy. Second, text classification models are vulnerable to a range of word-level operations that result in various perturbations. For instance, the word insertion operation introduces random words in the lexicon, while the word reordering operation causes positional permutation. These diverse perturbations can deceive text classification models successfully~\cite{garg2020bae,jin2020bert,lee2022query,li2021contextualized,feng2018pathologies}. To our best knowledge, there exist 
no certified defense methods 
% have been investigated 
against these word-level perturbations. 
% can provide certifiable robustness bounds for text classification models against these word-level perturbations. 
Third, significant absolute differences between adversarial and clean texts may exist due to word-level operations, while conventional  
% Conversely, 
randomized smoothing can only ensure the model's robustness against perturbations within a small radius. 
% Thus, the substantial difference may exceed the bounds of classical smoothing methods. 
Prior works~\cite{fischer2020certified,li2021tss,alfarra2022deformrs,hao2022gsmooth, liu2021pointguard} for images address this challenge by providing robustness guarantees against semantic transformations (e.g., rotation, scaling, shearing). However, they cannot be directly applied to the NLP domain because texts and words have a more heterogeneous discrete domain, and word insertion and deletion are new semantic transformations not involved in the image domain. %Firstly, unlike image pixels, numerous words have a greater heterogeneous discrete domain. Secondly, word insertion and deletion are new semantic transformations not involved in the image domain. 

\begin{figure}[]
  \centering
  \includegraphics[width=0.95\linewidth]{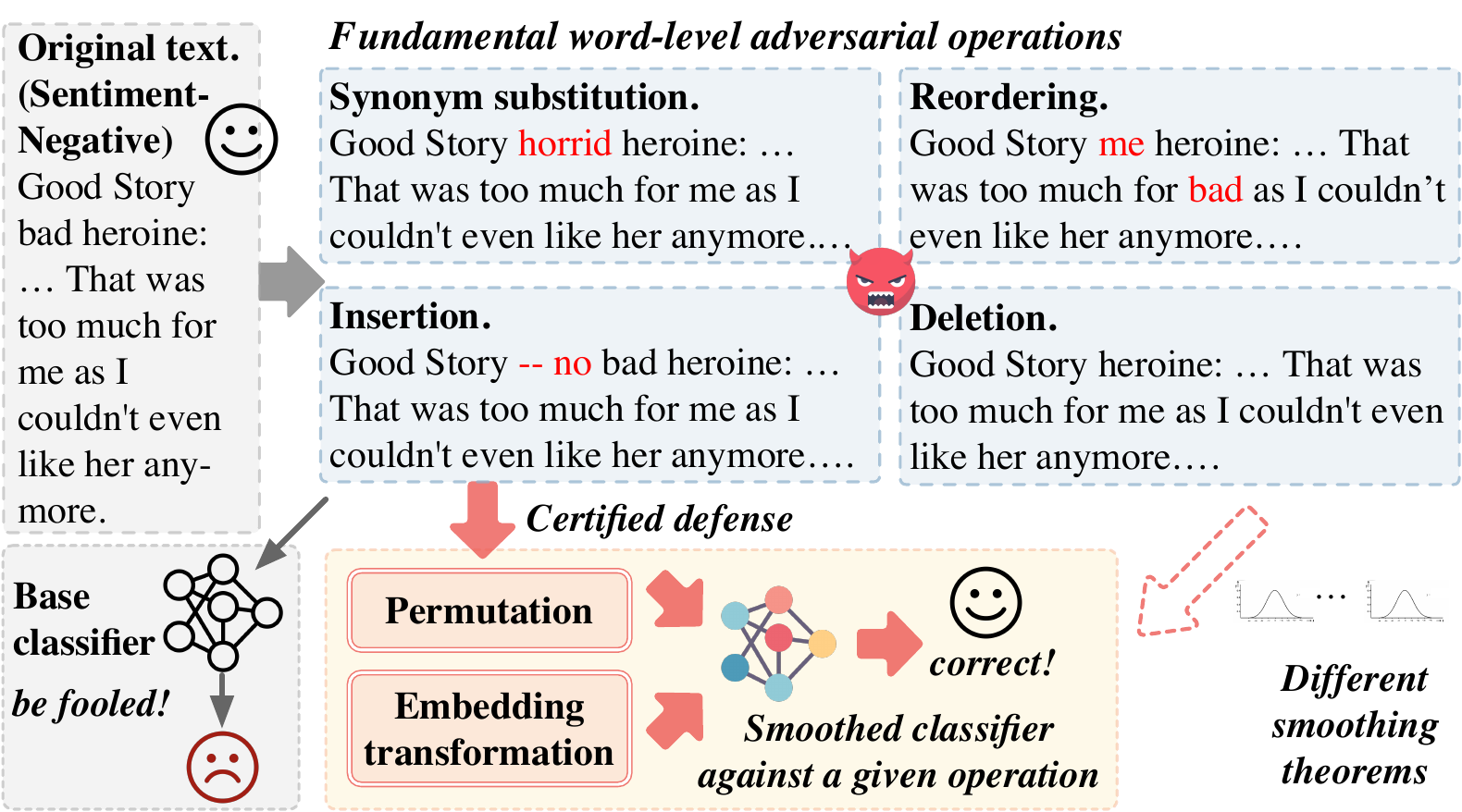}
  \vspace{-0.01in}
  \caption{Text-CRS is a robustness certification framework based on randomized smoothing of permutation and embedding transformations against word-level adversarial attacks.} %We develop different smoothing theorems for fundamental operations in order to provide better certified accuracy in synonym substitution, as well as the first benchmark in the other three operations.
%  \BW{Rearrangement $=>$ Reordering? Same in Figure 2. }\xy{done}
  \label{fig:Text-CRS}
  \vspace{-0.25in}
\end{figure}

%\BW{We need to mention somewhere that our Text-CRS can address the three challenges in the existing works.}
In this paper, we present Text-CRS, the first generalized certified robustness framework against common word-level textual adversarial attacks (i.e., synonym substitution, reordering, insertion, and deletion) via randomized smoothing (see Figure\,\ref{fig:Text-CRS}). Text-CRS 
%encompasses most of the fundamental operations that occur in textual adversarial attacks (i.e., synonym substitutions, reorderings, insertions, and deletions), and 
\emph{certifies the robustness in the word embedding space} without imposing restrictions on model architectures, and demonstrates \emph{high universality} and \emph{superior accuracy} compared to state-of-the-art (SOTA) methods (see Table\,\ref{tab:intro}). 
% instead of the discrete text space. 
%Particularly, we consider incorporating word embeddings instead of d
%\BW{To our best knowledge, we take the first cut to?} 
%we consider numerical relationships between words rather than discrete relationships by incorporating the word embedding vectors into the certification. 
% Furthermore, by partitioning the operations' input space into permutation and embedding spaces, 
Specifically, we first model word-level adversarial attacks as combinations of permutations and embedding transformations. 
%\BW{Here we briefly describe the permutation and embedding transformations for two attacks, e.g., Synonym substitution and word insertion?}
%\BW
{For instance, synonym substitution attacks transform the original words' embeddings with the synonyms' embeddings, while word reordering attacks transform the word orderings with certain word permutations.} 
Then, the {word-level adversarial attacks} can be guaranteed to be certified robust (``\emph{practical}'') if their corresponding permutation and embedding transformation are certified. 

%\BW
To this end, we develop customized theorems (see Section~\ref{sec:theorem}) to derive the certified robustness against each attack.
% Furthermore, to improve the certified accuracy/radius,  % theoretically,
In each theorem, we use an 
appropriate 
% novel theorems that apply well-designed 
noise distribution for our smoothing distributions in order to fit different word-level attacks. 
% transformations in word-level adversarial operations. 
% Specifically, 
% in Text-CRS, 
For instance, unlike existing works~\cite{ye2020safer,huang2022word} that use the uniform distribution, we propose to use the Staircase-shape distribution~\cite{geng2015staircase} to simulate the synonym substitutions, which ensures that semantically similar synonyms are more likely to be substituted. %fixing the relative importance of each position in a sentence to the prediction is impossible due to differences in model architectures and datasets; thus, 
% \BW{address the uniform sampling issue..}
Moreover, we use a uniform distribution to simulate the word reordering. %as it is difficult to determine the relative importance of the words.
For the word insertion with a wide range of inserted words, we inject Gaussian noise into the embeddings. For the word deletion, the embedding vector of each word is either kept or deleted (i.e., set to be zero). %(simulate deleted words). 
%the deleted words are set to be all-zero vectors. 
Hence, we use the Bernoulli distribution to simulate the status of each word. 
%randomize the word deletions with a certain probability. 
Our certified radii for the four attacks are then derived based on the corresponding noise distributions. 
% For each of the four adversarial operations, we theoretically derive the certified radius based on the corresponding randomization. 

%In summary, Text-CRS provides certified robustness for the four fundamental operations without imposing restrictions on model architectures and demonstrates superior accuracy compared to SOTA methods, as outlined in Table~\ref{tab:intro}.  %\hl{briefly discuss Table 1} \xy{(done). The IBP-based method and the RS-based method have no restrictions on the operation and model architecture, respectively. So should I define universality as having no restrictions on both operations and models? I only compared Text-CRS with SAFER and CISS in the evaluation.} \YH{we need to be careful here about over-claiming. It seems good if considering both of these as universal....and then we can add universal to the title, and revise "general" to "universal"}
% Our framework provides provable robustness for text classification models against four operations. 

To further improve the certified accuracy/radius, 
% Finally, 
we also propose a training toolkit incorporating three optimization techniques designed for training. For instance, instead of using isotropic Gaussian noises that can lead to distortion in the word embedding space, we propose to use an anisotropic Gaussian noise and optimize it to enlarge the certified radius.  
%for text classification models, including three simple yet effective training methods to improve the certified accuracy and robustness bound. 

%Notably, these word-level adversarial operations can essentially be treated as word insertions. As a result, our word insertion theorem also applies to the other three word-level adversarial operations, while our refined theorems are more effective for the other three operations. Furthermore, we provide theoretical robustness bounds of models trained with Text-CRS to allow the model owners (defenders) to assess the model's robustness against different operations. 

%We conduct extensive experiments to evaluate our framework Text-CRS on two types of models (i.e., LSTM~\cite{LSTM} and BERT~\cite{devlin-etal-2019-bert}) and three datasets with different tasks. We offer certified accuracy under different certified radii for word-level adversarial operations and certified accuracy against five SOTA adversarial attacks (i.e., TextFooler~\cite{jin2020bert}, WordReorder\cite{moradi2021evaluating}, SynonymInsert\cite{morris2020textattack}, BAE-Insert\cite{garg2020bae}, InputReduction\cite{feng2018pathologies}). Besides, we evaluate and analyze the universality of the word insertion theorem. Our framework outperforms SOTA methods in synonym substitution and provides the first certified robustness results for other word-level adversarial operations. 
% where the adversarial perturbation may exceed the radius threshold.

Thus, our major contributions are summarized below: 

\begin{itemize}[leftmargin=*]
\setlength\itemsep{0.3em}
    \item To our best knowledge, we propose the first generalized framework Text-CRS to certify the robustness for text classification models against four fundamental word-level adversarial operations, covering most word-level textual adversarial attacks. Also, the certification against word insertion can be \emph{universally} applied to other operations. %\hl{mention universality? and enhance?}

    \item We provide novel robustness theorems based on Staircase, Uniform, Gaussian, and Bernoulli smoothing distributions against different operations. We also theoretically derive certified robustness bounds for each operation. 

    % \item We propose an enhanced training toolkit, including three simple yet effective methods to improve certified accuracy and robustness bounds.

    \item To study the deceptive potential of adversarial texts, we apply ChatGPT to assess whether adversarial texts can be crafted to be semantically similar to clean texts.

    \item We conduct extensive experiments to evaluate Text-CRS, including our enhanced training toolkit, on three real datasets (AG’s News, Amazon, and IMDB) with two NLP models (LSTM and BERT). The results show that Text-CRS effectively handles five representative adversarial attacks and achieves an average certified accuracy of $81.7\%$, which is a $64\%$ improvement over SOTA methods. Text-CRS also significantly outperforms SOTA on the substitution operation. Besides, it provides new benchmarks for certified robustness against the other three operations. 
    % (AG’s News~\cite{zhang2015character}, Amazon~\cite{mcauley2013hidden}, and IMDB~\cite{maas2011learning}) with two NLP models (LSTM~\cite{LSTM} and BERT~\cite{devlin-etal-2019-bert})
    % On the certification against synonym substitution, Text-CRS significantly outperforms the SOTA methods. 
\end{itemize}

\vspace{-0.1in}
\section{Preliminaries} %and \BW{Problem Formulation}}
\label{sec:prelim}
\vspace{-0.1in}

%\BW{Should we separate Preliminaries and Related Work? This part looks like a mixture of Related work and Preliminaries.} \xy{I have merged Preliminaries and Related Work.} \YH{I think we need to separate them. This section is too long. We can introduce preliminaries here, and related work in the second last section (before the conclusion)}

%\BW{Fix $x$ as word sequence, text, or words?} \xy{done, text}

\subsection{Text Classification} \label{sec:textclass}
\vspace{-0.1in}
% intent detection, document classification
% https://levity.ai/blog/9-text-classification-examples
% Text classification 
% % is a machine learning technique that 
% aims to assign a predefined label to a given text
% % . It is one of the fundamental tasks of NLP 
% and has a wide range of NLP applications, such as sentiment analysis, topic classification, and content filtering~\cite{9TextCla71:online}. Text classification tasks can handle short sentences (such as tweets, chatbot queries, and news headlines) and large documents (such as product reviews, news articles, and government reports). The text classification involves three key components: data processing, embedding layer, and classification model.
% Using sentiment analysis on movie reviews as an example, a text classification model takes a movie review as input, analyzes its content, and automatically categorizes its sentiment as positive or negative.
% from the input space $\CX$ , in the label space $\CY$
% Specifically, 
The objective of text classification is to map texts to labels. A text consisting of $m$ words is denoted as $x=\{x_1,...,x_m\}$, where $x_i$ is the $i$th word, aka., a token. 
Text classification involves three key components: data processing, embedding layer, and classification model. Given a text $x$ labeled $y$, data processing first pads it to a fixed length $n$. When $m<n$, the processing inserts $n-m$ \texttt{<pad>} tokens to the end of the text, and when $m>n$, the processing drops the $m-n$ tokens. 
For brevity, we will denote a text $x$ with a fixed length of $n$ after data processing, i.e., $x=\{x_1,...,x_n\}$. 
Then, the embedding layer converts each token $x_i$ 
% in the text 
into a high-dimensional embedding vector $w_i$. After deriving the embedding matrix $w=\{w_1, \cdots, w_n\}$, the text classification model $h$ learns the relationship between the embedding matrix $w$ and the label $y$. For a tuple $(w, y)$, the model uses a loss function $\ell$ to derive the loss $\ell(h(w), y)$ and updates the model $h$, e.g., using stochastic gradient descent. 

In text classification tasks, the embedding layer is essential for representing words in a text and often has a large number of parameters. In practice, the embedding layer  is usually a pre-trained word embedding, such as Glove~\cite{pennington2014glove} or 
% an embedding layer in 
a pre-trained language model, such as BERT~\cite{devlin-etal-2019-bert}. The parameters in the embedding layer are frozen and 
% during the learning process and are not updated. On the contrary, 
only parameters in the classification model are updated during training. The classification model typically uses a deep neural network of different architectures, such as recurrent neural network (RNN)~\cite{schuster1997bidirectional}, convolutional neural network (CNN)~\cite{kim-2014-convolutional}, and Transformer model~\cite{vaswani2017attention}. 

% Then the classification model $h$ learns to map the embedding space $\CW$ to label space $\CY$.

\vspace{-0.1in}
\subsection{Adversarial Examples for Text Classification} \label{sec:adv_attacks}
\vspace{-0.1in}

% what is and how to generate adversarial examples
Adversarial examples are well-crafted inputs that lead to the misclassification of machine learning models via slightly perturbing the clean data. In text classification, an adversarial text fools the text classification model, and is also semantically similar to the clean text for human imperceptibility. 
To generate adversarial texts, there are three main approaches based on different perturbation strategies. (1) \textit{Character-level adversarial attacks}~\cite{li2018textbugger, ebrahimi2017hotflip, gao2018black} substitute, swap, insert, or delete certain characters in a word and generate carefully crafted typos, such as changing \textit{``terrible''} to \textit{``terrib1e''}. However, these typos can be detected and corrected by spell checker easily~\cite{pruthi2019combating}. 
(2) \textit{Word-level adversarial attacks}~\cite{jin2020bert, moradi2021evaluating, morris2020textattack, garg2020bae, feng2018pathologies} alter the words within a text through four primary operations: substituting words with their synonyms, reordering words, inserting descriptive words, and deleting unimportant words. This attack is widely exploited because only a few words of perturbations can lead to a high attack success rate~\cite{wang2adversarial}.
% Word-level perturbations do not alter the texts' semantic meanings since only a few words are modified, but they can completely change the output of the classification model. 
(3) \textit{Sentence-level adversarial attacks}~\cite{iyyer2018adversarial, ribeiro2020beyond} adopt another perturbation strategy that paraphrases the whole sentence or inserts irrelevant sentences into the text. This approach affects predictions by disrupting the syntactic and logical structure of sentences rather than specific words. However, this attack makes it difficult to preserve the original semantics due to rephrasing or inserting irrelevant sentences~\cite{yang2022robust}.

%\subsubsection{Word-level Adversarial Attacks} \label{sec:word_ae}
%To tackle the challenges of adversarial attacks,
As discussed above, the word-level adversarial attacks have widespread and severe impacts. Thus, in this paper, we focus on the certified defenses against \emph{word-level adversarial attacks}. We distill and consolidate such attacks into four fundamental adversarial operations: substitution, reordering, insertion, and deletion (see Figure \ref{fig:Text-CRS}). We offer provable robustness guarantees against these operations. The proposed theorems and defense methods can be extended to the other two types of adversarial attacks on text classification with the similar operations. %(as discussed in Section \ref{sec:disscuss}).

\vspace{0.03in}

\noindent\textbf{Synonym Substitution:} %~\cite{alzantot2018generating,ren2019generating,jin2020bert, zang2020word,tan2020s,dongtowards}. 
this operation generates adversarial texts by replacing certain words in the text with their synonyms, thereby preserving the text's semantic meanings. 
% Ren et al.~\cite{ren2019generating} propose a greedy PWWS algorithm to determine the replacement order of words in a sentence and the selection of synonyms. 
For instance, to minimize the word substitution rate, \textit{TextFooler}~\cite{jin2020bert} picks the word crucial to the prediction, i.e., when this word is removed, the prediction undergoes a significant deviation; then, it selects synonyms with high cosine similarity to the original embedding vector for substitution. 
% Tan et al.~\cite{tan2020s} proposed Morpheus, which generates plausible and semantically similar adversarial texts by replacing the words with their inflected form.

\vspace{0.03in}

\noindent\textbf{Word Reordering:} %~\cite{moradi2021evaluating,nie2019analyzing,yan2021consert,lee2020slm}. 
this operation selects and randomly reorders several consecutive words in the text while keeping the words themselves unchanged. For instance, 
\textit{WordReorder}~\cite{moradi2021evaluating} investigates the sensitivity of NLP systems to such adversarial operation and shows that reordering leads to an average decrease in the accuracy of 18.3\% and 15.4\% for the LSTM-based model and BERT on five datasets. 
% In Transformer models, the \texttt{"position\_ids"} parameter is utilized to determine the position of a token within the text\cite{devlin-etal-2019-bert}. However, this parameter is ineffective in defending against word reordering because the words' positions are altered directly in the input text. 

\vspace{0.03in}

\noindent\textbf{Word Insertion:} %~\cite{morris2020textattack,lietal2021contextualized,garg2020bae, behjati2019universal}. 
this operation generates adversarial text by inserting new words into the clean text. For instance, the NLP adversarial example generation library TextAttack~\cite{morris2020textattack} includes a basic insertion strategy \textit{SynonymInsert} that inserts synonyms of words already in the text to maintain semantic similarity between adversarial and clean texts. 
% Table~\ref{tab:certi_acc} displays that this operation can result in an average 26.48\% reduction in accuracy. 
\textit{BAE-Insert}~\cite{garg2020bae} uses masked language models (e.g., BERT) to predict newly inserted \texttt{<mask>} tokens in the text. The predicted words are then used to replace the \texttt{<mask>} tokens for the adversarial text. Compared to SynonymInsert, BAE-Insert improves syntactic and semantic consistency. %~\cite{garg2020bae}

\vspace{0.03in}

\noindent\textbf{Word Deletion:} %~\cite{feng2018pathologies,moradi2021evaluating,xie2022word}. 
this operation generates adversarial texts by removing several words from clean text. 
\textit{InputReduction}~\cite{feng2018pathologies} iteratively removes the least significant words from the clean text, and demonstrates that specific keywords play a critical role in the prediction of language models. Table~\ref{tab:attack_acc} shows that it can lead to an average of $51.78\%$ accuracy reduction. 
% As a result, modifying only a limited number of keywords can significantly impact the final prediction results while preserving semantic similarity. 

\vspace{-0.1in}
\subsection{Threat Model} \label{sec:threatmodel}
\vspace{-0.1in}
%\BW{A general comment on describing the attacks: should we add references for each attack only while leaving the attack details in the prelim section?} \xy{I have removed the details to the preliminary}\YH{let us merge 3.1 and 2.2 (adversarial examples and threat model) and name section 2 as preliminaries. Also, move 3.3 to Sec 4. Sections 2, 3, 4 are somewhat repetitive.} \xy{done}

We consider a threat model similar to that of other randomized smoothing and certified methods~\cite{ye2020safer, zhao2022certified, li2021tss}, which guarantees robustness as long as perturbations remain within the certified radius. They provide effective defense against both white-box and black-box attacks, irrespective of the specific attack types and adopted methods.
Specifically, we assume that the adversary can launch the evasion attack on a given text classification model by intercepting the input and perturbing the input with a wide variety of \emph{word-level adversarial attacks}~\cite{jin2020bert,moradi2021evaluating,morris2020textattack,garg2020bae}. 
% , assuming that the adversary has full access to the model (e.g., model architecture, model parameters). 
Given a text classification model $h$ and a text $x$ with label $y$, the goal of the adversary is to craft an adversarial text $x'$ from $x$ to alter its prediction, i.e., $h(x') \neq h(x)=y$. 
% To obtain $x'$, the adversary can employ various operations on $x$. 
While generating the adversarial texts, the adversary can choose any single aforementioned word-level adversarial operation or a combination thereof, as all textual adversarial attacks can be unified as the transformation of the embedding vector. These four types of operations encompass almost all possible modifications to texts in adversarial attacks~\cite{moradi2021evaluating} and their 
% (to be performed by the adversary) 
formal definitions are provided as below: 
% replacing words in $x$ with their synonyms, randomly reordering the words in $x$, inserting words, and removing random from $x$. These four operations encompass common word-level adversarial attacks.
% Next, we comprehensively explain how to produce an adversarial sequence $x'$ from a specified sequence $x$ using the four operations.
% and show that embedding-level transformations are equivalent to word-level operations. 

\vspace{0.03in}

\noindent\textbf{Synonym Substitution:} we replace certain words in the text $x$ with synonyms of the original word. Specifically, we convert each word $x_i$ to $x_i'$, where $x_i'$ may be a synonym of $x_i$ or $x_i$ itself. The operation can be represented as:
$$x=\{{x_1}, \cdots, {x_n}\}\to x'=\{x_1', \cdots, x_n'\},$$
where $x$ and $x'$ are of the same length.

\vspace{0.03in}

\noindent\textbf{Word Reordering:} to reorder certain words in the text $x$, we move the word at position $i$ to position $r_i$, where $r_i$ may be equal to $i$. The operation can be expressed as:
$$x=\{{x_1}, \cdots, {x_n}\} \to x'=\{{x_{r_1}}, \cdots, {x_{r_n}}\},$$
where $x$ and $x'$ are of the same length. 

\vspace{0.03in}

\noindent\textbf{Word Insertion:} 
we insert a word $x_\mathrm{In}^1$ into the $j_1$th position, a word $x_\mathrm{In}^2$ into the $j_2$th position, $\cdots$, and a word $x_\mathrm{In}^{m'}$ into the $j_{m'}$th position to $x$, where $m'$ is the total number of inserted words. The operation can be expressed as:
\begin{equation}
\begin{aligned}
x= &\{x_1,\cdots, x_{j_1-1}, x_{j_1}, \cdots,  x_{j_2-1}, x_{j_2}, \cdots, x_n\} \to \\
x'= &\{x_1,\cdots, x_{j_1-1}, x_\mathrm{In}^1, x_{j_1}, \cdots, x_{j_2-1}, x_\mathrm{In}^2, x_{j_2}, \cdots, x_n\},
\nonumber
\end{aligned}
\end{equation}
where $x'$ includes $m'$ more words than $x$.

\vspace{0.03in}

\noindent\textbf{Word Deletion:} 
we delete $m'$ words at position $j_1, \cdots, j_{m'}$ from $x$. The operation can be represented as:
\begin{equation}
\begin{aligned}
x=&\{x_1,\cdots, x_{j_1-1}, x_{j_1}, x_{j_1+1}, \cdots, x_{j_2-1}, x_{j_2}, \cdots, x_n\} \to \\
x'=&\{x_1,\cdots, x_{j_1-1}, x_{j_1+1}, \cdots, x_{j_2-1}, x_{j_2+1}, \cdots, x_n\},
\nonumber
\end{aligned}
\end{equation}
where $x'$ includes $m'$ words less than $x$.

%\BW{The above two attacks are described in a general way, but here described using a single word? It is for ease of description, but seems not consistent. Same for word deletion attacks}\xy{done}

\vspace{-0.1in}
\subsection{Randomized Smoothing for Certified Defense}
\vspace{-0.1in}

Randomized smoothing~\cite{lecuyer2019certified,cohen2019certified} is a widely adopted certified defense method that offers state-of-the-art provable robustness guarantees for classifiers against adversarial examples. % This approach involves training a smoothed version of the classifier and certifying its robustness during prediction. Specifically, during  training, a random noise sampled from certain distribution is injected to the input data. The classifier is then trained on the noisy input, resulting in a more stable decision boundary that is less susceptible to perturbations. During certification, the smoothed classifier is theoretically guaranteed to maintain stable prediction results within a certain range of perturbations. The perturbation radius for a given input is determined by considering the random noise level and the probability of the smoothed classifier assigning the input to the most probable class ($y_A$) as well as the second most probable class ($y_B$). 
% adding random noise sampled from a suitable probability distribution
It has two key advantages: applicable to any classifier and scalable to large models. 
Given a testing example ${x}$ with label $y$ from a label set $\mathcal{Y}$, 
randomized smoothing has three steps: 1) 
% given an arbitrary 
define a (\emph{base}) classifier $h$; 2) build a \emph{smoothed classifier} $g$ based on $h$, ${x}$, and a noise distribution;  
% via adding random noise to the testing example, 
and 3) derive certified robustness for the smoothed classifier $g$. 
Under our context, the base classifier $h$ can be any trained text classifier and $x$ is a testing text.
Let  $\epsilon$ be a random noise drawn from an \emph{application-dependent} {noise distribution}. Then the smoothed classifier $g$ is defined as $g({x}) = \argmax_{l\in \mathcal{Y}}\text{Pr}(h({x}+\epsilon)=l)$.
Let $p_A, p_B\in[0,1]$ be the probability of the most ($y_A$) and the second most probable class ($y_B$) outputted by $\text{Pr}(h({x}+\epsilon))$, respectively, i.e., $p_A= \max_{l} \text{Pr}(h({x}+\epsilon)=l)$ and $p_B=\max_{l \neq y_A} \text{Pr}(h({x}+\epsilon)=l)$. 
Then $g$ provably predicts the same label $y_A$ for ${x}$ once the adversarial perturbation $\delta$ is bounded, i.e., $ g({x}+\delta) = y_A, \forall ||\delta||_p \leq R$, where  $||\cdot ||_p$ is an $\ell_p$ norm and $R$ is called \emph{certified radius} that depends on $p_A, p_B$. 
For example, when the noise distribution is an isotropic Gaussian distribution with mean 0 and standard deviation $\sigma$, \cite{cohen2019certified} adopted the Neyman-Person Lemma~\cite{neyman1933ix} and derived a \emph{tight} robustness guarantee against $l_2$ perturbation, i.e., 
$g({x}+\delta) = y_A, \forall ||\delta||_2 \leq R=\frac{\sigma}{2}(\Phi^{-1}({p_A})-\Phi^{-1}({p_B}))$, 
where $\Phi^{-1}$ is the inverse of the standard Gaussian {CDF}. 
This property implies that the smoothed classifier $g$ maintains constant predictions if the norm of the perturbation $\delta$ is smaller than the certified radius $R$.

\vspace{-0.05in}
\section{The Text-CRS Framework}
\vspace{-0.1in}

This section introduces Text-CRS, a novel certification framework that offers provable robustness against adversarial word-level operations. We first outline the new challenges in designing such certified robustness. Next, we 
%provide an overview of the frequently used notations and 
formally define the permutation and embedding transformations that correspond to each adversarial operation. We then define the perturbation that Text-CRS can certify for each permutation and embedding transformation. Finally, we conclude with a summary of our framework and its defense goals.

\vspace{-0.1in}
\subsection{New Challenges in Certified Defense Design} 
\vspace{-0.1in}

Previous studies on certified defenses in text classification, including SAFER~\cite{ye2020safer} and CISS~\cite{zhao2022certified}, only provide robustness guarantees \emph{against substitution operations}. 
% under an unrealistic assumption
Certified defenses against other word-level operations, such as word reordering, insertion, and deletion, are 
% remain largely 
unexplored. 
We list below the weaknesses of existing defenses as well as several technical challenges: 
% due to the following main challenges. %should be addressed.

% \begin{enumerate}[(C1)]
\begin{itemize}[leftmargin=*]

\setlength\itemsep{0.3em}

     % \item 
     % % The assumption of a uniform substitution distribution among synonyms is impractical, as different synonyms exhibit varying substitution probabilities. However, prior studies on d a uniform distribution among synonyms, resulting in low certified accuracy. 
      
     % Existing certified defenses against synonym substitutions assume that all synonyms have the same probability to substitute a target word. However, this assumption is unrealistic, as certain synonyms are more frequent/similar than others, and hence should have larger probabilities to do so. Such an assumption makes the existing works yield relatively low certified radius/accuracy. 
     % % A follow-up research question is how can we model the different p 

    \item \textbf{Measuring the perturbation and certified radius}. Words are unstructured strings and there is no numerical relationship among discrete words. This makes it challenging to measure the $\ell_1$ and $\ell_2$ distance between words, as well as the perturbation distance between the original and adversarial text (while deriving the certified radius).
    
    \item \textbf{Customized noise distribution for randomized smoothing against different word-level attacks}.
           The assumption of a uniform substitution distribution among synonyms is unrealistic, as different synonyms exhibit varying substitution probabilities. However, previous works on the certified robustness against word substitution almost assume a uniform distribution within the set of synonyms. Such an assumption makes these works yield relatively low certified accuracy (see Section~\ref{sec:results}). Hence, we need to 
           construct customized noise distributions best suited for the certification against the synonyms substitution attack as well as the other three word-level attacks.

         \item \textbf{Inaccurate representation of distance}. The absolute distance between operation sequences is typically high for word reordering, insertion, and deletion operations. Although studies, such as TSS~\cite{li2021tss} and DeformRS~\cite{alfarra2022deformrs}, have investigated the pixel coordinate transformation in the image domain, the word reordering transformation has not been studied in the NLP domain. Moreover, word insertion and deletion are unique transformations specific to NLP which are not applicable in the image domain. 

\end{itemize}

To address these challenges, we employ the 
% numerical data structure, i.e., 
numerical word embedding matrix\cite{pennington2014glove, devlin-etal-2019-bert} as inputs to our model instead of the word sequence. We can then use embedding matrices to measure the $\ell_1$ and $\ell_2$ distance between word sequences. We also introduce a permutation space to solve the problem of high absolute distance. 
Moreover, we design a customized noise distribution for randomized smoothing w.r.t. each attack and derive the corresponding certified radius, shown in Theorems in Section~\ref{sec:theorem}. 
% also 
% caused by modifying word permutations. 
% Furthermore, in Section~\ref{sec:theorem}, we propose four randomized smoothing theorems that provide certified robustness against four word-level operations.

%\BW{Is it possible the universally certify the robustness against all the four attacks?} \xy{Yes. Theorem 3 can certify the robustness against all attacks, and we have evaluated the performance in Table~5. The performance of Theorem 3 is well but worse than the other theorems' performances.}

\vspace{-0.1in}
\subsection{Permutation and Embedding Transformation} \label{sec:notation}
\vspace{-0.1in}

\begin{table}[]
\centering
\caption{Frequently Used Notations}
\vspace{-0.05in}
% \resizebox{\linewidth}{!}{
% \begin{threeparttable}  
\begin{tabular}{@{}ll@{}}
\toprule
Term      & Description               \\
\midrule
$\CX$ & Input text space \\
$\CW\subseteq \BR^{n\times d}$ & Embedding space \\
$\CU\subseteq \BR^{n\times n}$ & Permutation space \\
$\CU\cdot \CW \subseteq \BR^{n\times d} $ & Embedding space after applying permutation \\ 
$\CY$ & Label space \\ 
$\theta(u, r): \CU \times \CR$ & Permutation with parameter $r$ on $u$ \\ %: \CU\times \CR \to \CU
$\phi(w, t): \CW \times \CT$ & Embedding transformation with parameter $t$ on $w$ \\ %: \CW\times \CT \to \CW
$L_{emb}:\CX\to \CW$ & The pre-trained embedding layer \\
$h:\CU\cdot \CW\to \CY$ & Classification model (i.e., base classifier) \\
%$m$ & Length of the $x$, varying among samples  \\
$n$ & Constant maximum length of input sequence \\
$d$ & Dimension of each embedding vector \\
$\delta$ & Perturbations of permutation or embedding space \\
% $S$ & Set of parameters for permutation or embedding \\
\bottomrule
\end{tabular}
% \begin{tablenotes}
%     \item{(E)-Effectiveness; (S)-Stability; (R)-Robustness}
% \end{tablenotes}
% \end{threeparttable}
% }
\label{tab:notation}
\vspace{-4mm}
\end{table}

\subsubsection{Notations}
We denote the space of input text as $\CX $, the space of corresponding embedding as $\CW\subseteq \BR^{n\times d}$ (where $n$ is the constant maximum length of each input sequence and $d$ is the dimension of each embedding vector), and the space of output as $\CY=\{1, \cdots C\}$ where $C$ is the number of classes. We denote the space of embedding permutations as $\CU \subseteq \BR^{n\times n}$. For instance, given a word sequence $x$, its embedding matrix is $w=\{w_1, \cdots , w_n\}$, its permutation matrix is $u=\{u_1, \cdots, u_n\}$, and the input to the classification model will be $u\cdot w$. The position of $w_i$ is denoted by $u_i=[0, \cdots, 0, 1, 0, \cdots, 0]$, a standard basis vector represented as a row vector of length $n$ with a value of $1$ in the $i$th position and a value of $0$ in all other positions.
% \[
%     x_{z_0} = 0, \cdots, 0 \underbrace{1}_{n-z_0}, 0, \cdots, 0
% \]
%
% \begin{equation}
% u\cdot w=
% \begin{bmatrix}
%  u_1\\  u_2\\ \vdots \\ u_n
% \end{bmatrix}
% \cdot 
% \begin{bmatrix}
%  w_1\\  w_2\\ \vdots \\ w_n
% \end{bmatrix}
% =
% \begin{bmatrix}
%   1&  0&  \cdots& 0\\
%   0&  1&  \cdots& 0\\
%   \vdots & \vdots & \ddots& \vdots\\
%   0&  0& \cdots& 1
% \end{bmatrix} 
% \cdot 
% \begin{bmatrix}
%  w_1\\  w_2\\ \vdots \\ w_n
% \end{bmatrix}
% = 
% \begin{bmatrix}
%  w_1\\  w_2\\ \vdots \\ w_n
% \end{bmatrix}
% \nonumber
% \end{equation}
%

We model the \emph{permutation transformation} as a deterministic function $\theta: \CU \times \CR \to \CU$, where the permutation matrix $u\in \CU$ is permuted by a $\CR$-valued parameter $r$. Vector $u_i$ is replaced with $u_{r}$ by applying $r$, and then the word embedding ($w_i$) is moved from position $i$ to $r$. 
% The only transformation for the permutation matrix $u$ is to swap the positions of $u_i$, which is equivalent to reordering the words in the sequence. 
Moreover, we model the \emph{embedding transformation} as a deterministic function $\phi: \CW\times \CT \to \CW$, where the original embedding $w\in \CW$ is transformed by a $\CT$-valued  parameter $t$. Based on such transformation, we can define all the operations in Section~\ref{sec:threatmodel}. For example, $\theta_I(u, r)\cdot \phi_I(w, t)$ represents the word insertion on the original input $u\cdot w$ with a permutation parameter $r$ and an embedding parameter $t$. Here, we denote $\cdot$ as applying the permutation $u$ to the embedding $w$, and $\times$ as the operations of the parameters applied to the permutation or the embedding matrices.
% We define different types of permutation and embedding transformations corresponding to the four word-level operations (i.e., synonym substitution, reordering, insertion, and deletion) and explain them in Section ~\ref{sec:emb_transform}. 

For simplicity of notations, we denote the classification model as $h: \CU\cdot \CW \to \CY$. Then, we adopt the pre-trained embedding layer ($L_{emb}$) for the text classification task, freeze its parameters, and only update parameters in the classification model. Essentially, the input space of the model ($\CU\cdot \CW \subseteq \BR^{n\times d}$) is the same as $\CW$, and the training process of $h$ is identical to that of a classical text classification model.
Table~\ref{tab:notation} shows our frequently used notations.
% \BW{Here we use $\theta_I(u, r)$, but later $\theta_I(u)$. Make them consistent?} \xy{done}

%\subsubsection{Word-level Operations} \label{sec:emb_transform}

%\YH{first introducing notations for permutation and embedding transformation, and then map them to the word-level operations. I moved the framework to the last part of this section}
%\xy{If we move Figure 2 and Section 3.3.3 to the end, is this less helpful for illustration? Maybe Figure 2 and Section 3.3.3 can just be deleted?}

\vspace{-0.1in}
\subsubsection{Permutation and Embedding Transformation} \label{sec:transform_detail}
%The four word-level operations mentioned in Section ~\ref{sec:threatmodel} can all be converted to the corresponding 
Given the above permutation and embedding transformations, synonym substitution and word reordering are single transformations while word insertion and deletion are composite transformations. Our transformations of the input tuple $(u=\{u_1, \cdots, u_n\}, w=\{w_1, \cdots, w_n\})$ are further represented as follows (no change to the sizes of $u$ and $w$).

%\noindent \BW{In Threat model, text length is $m$, but here it is $n$?} \xy{Because during training, the embedding layer pads each text to a fixed length $n$. The original text length $m$ is not fixed.}

\vspace{0.03in}

% \BW{I also have a general question here. Should we indeed need two parameters to represent each transformation $\theta(u,z)$ and $\phi(w,t)$? I checked the proof and looks like it may be unnecessary, as using $\theta(u)$ and $\phi(w)$ is fine? OR should we instantiate each $z$ and $t$ for each operation, as otherwise the two parameters are isolated. For instance, for Synonym Substitution, $z=\textrm{null}$ and $t=\{a^1, a^2, \cdots\}$. Also I suggest $z$ and $t$ correspond to the notations $a_i$ and $r_i$. How about changing $t$ to $a$ and $z$ to $r$, but not knowing whether are used in other places?}. 
% \xy{

% \noindent\ding{172} I think we need to preserve $r, t$ to control the noise level. Because in Definition 1 and Theorem 1-4 in Section  4, we need the noise $\rho, \varepsilon$ added on $u,w$ to follow some distribution. I mentioned $r,t$ here so I can replace $r,t$ with $\rho, \varepsilon$. 
% If I use $\theta(u)$ and $\phi(w)$, can I add $\rho$ and $\varepsilon$ directly to Definition 1 and Theorem?

% \noindent\ding{173} I have instantiate each $r$ and $t$.

% \noindent\ding{174} I have changed $z$ to $r$, and preserves $t$, because $t$ represents three transformations ($a, w, j$) in (substitution, insertion and deletion).
% }

% \BW{We can have a discussion if needed.}

\noindent\textbf{Synonym Substitution}
replaces the original word's embedding vector with the synonym's embedding vector. 

\vspace{-0.05in}

\begin{equation}
\small
\begin{aligned}
    (\theta_S(u,\textrm{null}), \phi_S(w,& \{a_1, \cdots, a_n\})) \\
    \qquad = (u, {w'})=& (u, \{w_1^{a_1}, \cdots, w_n^{a_n}\})
\nonumber
\end{aligned}
\end{equation}

\vspace{-0.05in}

where $w_j^{a_j}$ ($a_j$ are nonnegative integers) is the embedding vector of the $a_j$th synonym of the original embedding $w_j$, and $a_j=0$ indicates % the embedding vector 
$w_j$ itself: $w_i^0=w_i$. The permutation $\theta_S(u,\textrm{null})$ does not modify any entries of the permutation matrix $u$.

\vspace{0.03in}

\noindent\textbf{Word Reordering} does not modify the embedding vector but modifies the permutation matrix.

\vspace{-0.05in}

\begin{equation}
\small
\begin{aligned}
    (\theta_R(u, \{r_1, \cdots, r_n\}),&\ \phi_R(w,\textrm{null})) \\
     \qquad = (u', {w}) =& (\{{u_{r_1}}, \cdots, {u_{r_n}}\}, w)
\nonumber
\end{aligned}
\end{equation}

\vspace{-0.05in}

where $\{r_1, \cdots, r_n\}$ is the reordered list of $\{1, \cdots, n\}$. The transformation $\phi_R(w,\textrm{null})$ does not modify the elements of the embedding matrix $w$.

\vspace{0.03in}

\noindent\textbf{Word Insertion} first inserts $m'$ embedding vectors of the specified words at the specified positions ($j_1, \cdots, j_{m'}$). Then, it removes the last $m'$ embedding vectors to maintain the constant length $n$ of the text. (see Figure~\ref{fig:insert_delete}) %to obtain $(u'_0,w'_0)$. 

\vspace{-0.15in}

\begin{equation}
\small
\begin{aligned}
    &(\theta_I(u, \{r_1, \cdots, r_n\})), \phi_I(w,\{w_\mathrm{In}^{1}, \cdots, w_\mathrm{In}^{m'}\})) \\
    =&(u'_0, w'_0)
    =\!(\{u_1, \cdots, u_{j_1-1}, u_{j_1}, u_{j_1+1}, \cdots, u_{n}\}, \\
    &\qquad \quad \{w_1, \cdots, w_{j_1-1}, w_\mathrm{In}^{1}, w_{j_1}, \cdots, w_{n-m'}\}) \\
    =&(u', w')
    =\!(\{u_1, \cdots, u_{j_1-1}, u_{j_1+1}, \cdots, u_{n}, u_{j_1}, \cdots u_{j_{m'}}\}, \\
    &\qquad \quad \{w_1, \cdots, w_{j_1-1}, w_{j_1}, \cdots, w_{n-m'}, w_\mathrm{In}^{1}, \cdots, w_\mathrm{In}^{m'}\})
\nonumber
\end{aligned}
\end{equation}

%\vspace{-0.05in}

where $w_\mathrm{In}^{i}$ is the embedding vector of the inserted word at position $j_i$, $i\in [1,m']$. To minimize the distance between $w$ and $w'_0$, $w_\mathrm{In}^{i}$ and its corresponding position vector $u_{j_i}$ are shifted to the end of the sequence to obtain $(u', w')$.
% The $\ell_0$ distance between $u$ and $u'$ is $2(n-j)$. The $\ell_2$ distance between $w$ and $w'$ is $\delta_\mathrm{In}= \|w_\mathrm{In}-w_n\|_2$.

\vspace{-2mm}
\begin{figure}[!h] %htbp
  \centering
  \includegraphics[width=1.0\linewidth]{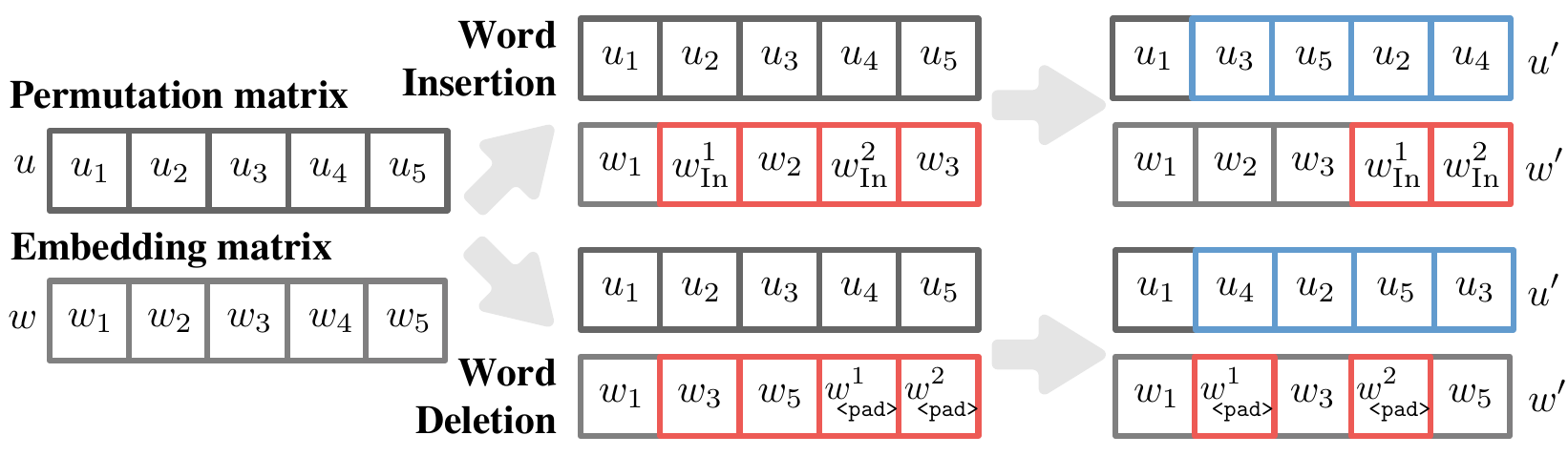}
% \vspace{-1mm}
  \caption{Word insertion and deletion. Blue and red indicate the changes to the permutation and embedding matrices.}
  \label{fig:insert_delete}
   \vspace{-1mm}
\end{figure}

\noindent\textbf{Word Deletion}
first replaces $m'$ embedding vectors with all-zero % embedding 
vectors (of \texttt{<pad>}) at position $j_1, \cdots, j_{m'}$. Then it moves the positions (i.e., permutation vector) of these all-zero vectors to the end of the sequence. (see Figure~\ref{fig:insert_delete})  

\vspace{-0.1in}

\begin{equation}
\small
\begin{aligned}
    &(\theta_D(u,\{r_1, \cdots, r_n\})), \phi_D(w,\{j_1, \cdots, j_{m'}\})) \\
    % =&(u'_0=\{u_1, \cdots, u_{j_1-1}, u_{j_1}, \cdots, u_{n-m'}, {u_{n-m'+1}}, \cdots, {u_{n}}\}\\
    % &{w'_0}=\{w_1, \cdots, w_{j_1-1}, w_{j_1+1}, \cdots, w_n, {w_{\texttt{<pad>}}^{1}}, \cdots, {w_{\texttt{<pad>}}^{m'}}\}) \\
    =&(u', w')=(\{u_1, \cdots, u_{j_1-1}, {u_{n-m'+1}}, u_{j_1}, \cdots, u_{n-m'}\}, \\
    &\qquad \qquad \quad \{w_1, \cdots, w_{j_1-1}, {w_{\texttt{<pad>}}^{1}}, w_{j_1+1}, \cdots, w_n\})
\nonumber
\end{aligned}
\end{equation}

\vspace{-0.05in}

where $w_\texttt{<pad>}^i$ is the embedding vector of \texttt{<pad>} (i.e., the all-zero vector), and it replaces the original embedding vector $w_{j_i}$, $i\in [1, m']$. The position vector of $w_\texttt{<pad>}^i$ is $u_{n-m'+i}$, such that $u'\cdot w'$ corresponds to the embedding matrix generated after deleting $m'$ words at position $j_1, \cdots, j_{m'}$ in the text $x$.
% The $\ell_0$ distance between $u$ and $u'$ is $2(n-j)$. The $\ell_0$ distance between $w$ and $w'$ is the number of deletion words.
% the delete can only defend the delete attacks, which replace the word with a \texttt{<pad>}

% \BW{Is it appropriate to use a single word insertion to illustrate the transformation? Same for word deletion.} \xy{done, use $m'$ words}

\vspace{-0.1in}
\subsection{Framework Overview}
\vspace{-0.1in}

% \BW{By checking the proofs and notations, I strongly suggest clearly defining the perturbation for each operation and what we want to certify here, as otherwise it is easily to get lost. For instance, to describe the perturbation for Synonym Substitution, we can say $w \oplus \delta_S = [w_1 \oplus a_1, w_2 \oplus a_2, \cdots]$, where  $w_j \oplus a_j$ means we replace the word embedding $w_j$ with any synonym $w_j^{a_j}$. 
% Then in the next Section, we do not need to redefine them.} \xy{I have defined the four perturbations here and deleted the definition in Section  4.}

\subsubsection{Perturbations of Adversarial Operations} \label{sec:perturbation}
Since different operations involve different permutations and embedding transformations, we first define their perturbations and will certify against these perturbations in Section ~\ref{sec:theorem}.

\vspace{0.03in}

\noindent\textbf{Synonym Substitution} involves only the embedding substitution, which can be represented as $w \oplus \delta_S = \{w_1 \oplus a_1, \cdots, w_n \oplus a_n\}$. Here, $w_j \oplus a_j$ denotes the replacement of the word embedding $w_j$ with any of its synonyms $w_j^{a_j}$, while the original embedding can be $w_j=w_j \oplus 0$. Thus, the perturbation is defined as $\delta_S = \{a_1, \cdots, a_n\}$.  
% $a_j\in [1,s]$ where $s$ is the size of the thesaurus. 

\vspace{0.03in}

\noindent\textbf{Word Reordering} involves only the permutation of embeddings, which can be represented as $u \oplus \delta_R = \{u_1 \oplus r_1, \cdots, u_n \oplus r_n\}$, where $u_j \oplus r_j = u_{r_j}$ indicates that the embedding originally at position $j$ is reordered to position $r_j$. The reordering perturbation is $\delta_R=\{r_1-1, \cdots, r_n-n\}$.

\vspace{0.03in}

\noindent\textbf{Word Insertion} includes permutation and embedding insertion, with the permutation perturbation, $\delta_R$, being identical to word reordering. Embedding insertion preserves the first $n-m'$ embeddings while replacing only the last $m'$ embeddings with new ones. The perturbation of insertion is defined as $\delta_I=\{w_\mathrm{In}^1-w_{n-m'+1}, \cdots, w_\mathrm{In}^{m'}-w_n\}$.

\vspace{0.03in}

\noindent\textbf{Word Deletion} involves permutation and embedding deletion, where permutation perturbation is equivalent to $\delta_R$. We model the embedding deletion that converts any selected embedding to $w_{\texttt{<pad>}}$ as an embedding state transition from $b=1$ to $b=0$. The deletion perturbation is therefore defined as $\delta_D=\{1-b_1, \cdots, 1-b_n\}$, where all $b_j, j\in[1, n]$ are equal to 1, except for $b_{j_1}=0, \cdots, b_{j_{m'}}=0$, which represents the deleted embeddings at positions $j_1, \cdots, j_{m'}$.

%\YH{defense goal can be merged with Framework}

\vspace{-0.1in}
\subsubsection{Framework and Defense Goals}
Figure~\ref{fig:system} summarizes our Text-CRS framework. The input space is partitioned into a permutation space $\CU$ and an embedding space $\CW$, by representing each operation as a combination of permutation and embedding transformation (see Section~\ref{sec:transform_detail}). We analyze the characteristics of each operation and select an appropriate smoothing distribution to ensure certified robustness for each of them (see Section~\ref{sec:theorem}).

% \vspace{-3mm}

\begin{figure}[!h] %htbp
  \centering
  \includegraphics[width=1.0\linewidth]{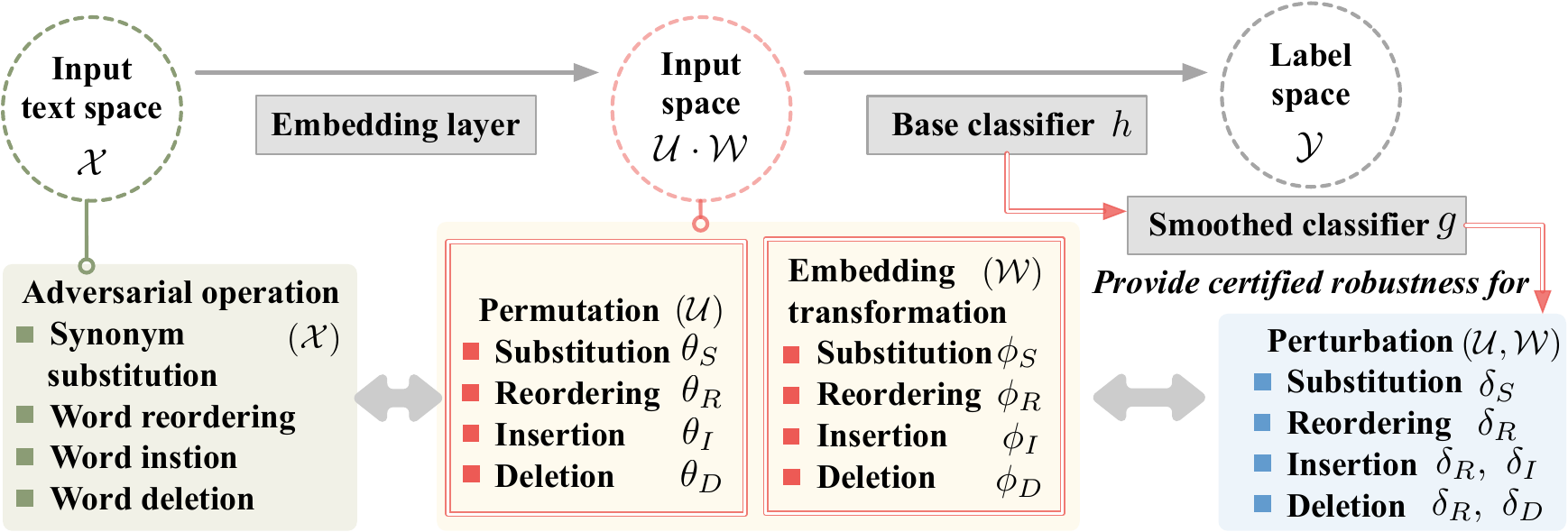}
  % \vspace{-0.02in}
  \caption{An overview of Text-CRS. %\xy{should I delete Sec.2.3 here? Because 2.3 is the threat model and the other two are system model}\YH{yes}
  }
  \label{fig:system}
  \vspace{-4mm}
\end{figure}

Since each word-level operation is equivalent to a combination of the permutation and embedding transformation, any adversary that perturbs the text input is indeed adjusting the parameter tuple $(r, t)$ of permutation and embedding transformation. Our goal is to guarantee the robustness of the model under the attack of a particular set of parameter tuple $(r, t)$. 
Specifically, we aim to find a set of permutation parameters $S_{adv}^r \subseteq \CR$ and a set of embedding parameters $S_{adv}^t \subseteq \CT$, such that the prediction results of the model $h$ remain consistent under any $(r, t) \in S_{adv}^r \times S_{adv}^t$, i.e., 
\begin{equation}
\small
    h(u\cdot w)=h(\theta(u, r)\cdot \phi(w, t)), \forall{r} \in S_{adv}^r, \forall{t} \in S_{adv}^t
\end{equation}

% $h: \BR^{n\times d} \to \CY$ be any deterministic or random function. 
% From $h$, we aim to construct a provably robust classifier $g: \BR^{n\times d} \to \CY$. 

\vspace{-0.1in}
\section{Permutation and Embedding Transformation based Randomized Smoothing} \label{sec:theorem}
\vspace{-0.1in}

%\BW{Should we first show a universal robustness guarantee against the 4 attacks (as this is the most practical and most powerful). Then we relax to show a customized robustness guarantee for each attack?} 
%\xy{
%\ding{172} Synonym substitution is a most common attack which has other baseline methods. 
%\ding{173} Word reordering, insertion, and deletion all use uniform-based permutations. The last two are combination operations. My original logic was to illustrate operations from simple to complex. Then explain why word insertion can provide a universal guarantee. 
%\ding{174} If we describe the word insertion first, there is little context in word reordering.  

%But I also agree with Prof. Wang's suggestion. }
 
% \BW{I also suggest to clearly define "what perturbations we are going to certify" at the beginning.} \xy{I define the perturbation at the beginning of each subsection}

In this section, we design the randomized smoothing for Text-CRS. %by first defining the base transformation-smoothing classifier and then 
% And subsequently 
%optimizing it against various word-level attacks. 
We construct a new transformation-smoothing classifier $g$ from an arbitrary base classifier $h$ by performing random permutation and random embedding transformation. Specifically, the transformation-smoothing classifier $g$ predicts the top-1 class returned by $h$ when the permutation and embedding transformation perturb the input embedding $u\cdot w$. Such a smoothed classifier can be defined as below.

\begin{definition}[$(\rho, \varepsilon)$-Smoothed Classifier]
\label{def:smoothg}
Let $\theta: \CU\times \CR \to \CU$ be a permutation, $\phi: \CW\times \CT \to \CW$ be an embedding transformation, and let $h:\CU\cdot\CW\to \CY$ be an arbitrary base classifier. Taking random variables $\rho\sim \BP_\rho$ from $\CR$ and $\varepsilon\sim \BP_\varepsilon$ from $\CT$, respectively, we define the $(\rho, \varepsilon)$-smoothed classifier $g:\CU\cdot\CW\to \CY$ as
\begin{equation}
\small
    g(u\cdot w)=\argmax_{y\in \CY}\BP(h(\theta(u, \rho)\cdot \phi(w, \varepsilon)))
    % q(y|u\cdot w; \rho, \varepsilon)\coloneqq \BE(p(|\theta(u, \rho)\cdot \phi(w, \varepsilon)))
\label{eq:smoothg}
\end{equation}
Given a constant permutation matrix, only the embedding transformation is performed. Thus, we have

\begin{equation}
\small
    g(u\cdot w)=\argmax_{y\in \CY}\BP(h(u\cdot \phi(w, \varepsilon)))
    % q(y|u\cdot w; \rho, \varepsilon)\coloneqq \BE(p(|\theta(u, \rho)\cdot \phi(w, \varepsilon)))
\label{eq:smoothg1}
\end{equation}
Similarly, given a constant embedding matrix, only the permutation is performed. Thus, we have
\begin{equation}
\small
    g(u\cdot w)=\argmax_{y\in \CY}\BP(h(\theta(u, \rho)\cdot w))
    % q(y|u\cdot w; \rho, \varepsilon)\coloneqq \BE(p(|\theta(u, \rho)\cdot \phi(w, \varepsilon)))
\label{eq:smoothg2}
\end{equation}
\end{definition}

To certify the classifiers against various word-level attacks, adopting an appropriate permutation $\theta$ and embedding transformation $\phi$ is necessary. For instance, to certify robustness against synonym substitution, using the same (substitution) transformation in the smoothed classifier is reasonable. Nevertheless, this strategy may not yield the desired certification for other types of operations. %necessitating an alternative approach. 
Next, we illustrate the transformations and certification theorems corresponding to the four word-level operations.

\vspace{-0.1in}
\subsection{Certified Robustness to Synonym Substitution} 
\vspace{-0.1in}

Synonym substitution only transforms the embedding matrix without changing the permutation matrix. 
% $u$ 
%unchanged. 
Previous works~\cite{ye2020safer, zeng2021certified} assume a uniform distribution over a set of synonymous substitutions, i.e., the probability of replacing a word with any synonym is the same. \emph{However, this assumption is unrealistic since the similarity between each synonym and the word to be substituted would be different.} For instance, when substituting the word \emph{good}, \emph{excellent} and \emph{admirable} are both synonyms, but the cosine similarity between the embedding vector of (\emph{good, excellent}) is higher than that of (\emph{good, admirable})~\cite{pennington2014glove}. Hence, the likelihood of selecting \emph{excellent} as a substitution should be higher than choosing \emph{admirable}. 
To this end, 
% introduce 
% to model the relationship between a word and its synonyms 
we design a smoothing method based on the Staircase randomization~\cite{geng2015staircase} . 
\vspace{-0.1in}
\subsubsection{Staircase Randomization-based Synonym Substitution}
% \subsubsection{Staircase  Mechanism}
%\BW{Motivation of Staircase Mechanism is unclear to me, though I know it is a better mechanism than Laplace.} \xy{I've added some reasons in the underlined area. Do you think they are convincing?}
The Staircase randomization mechanism originally uses a staircase-shape distribution to 
replace the standard Laplace distribution for improving the accuracy of differential privacy~\cite{geng2015staircase}. It consists of partitioning an additive noise into intervals (or steps) and adding the noise to the original input, with a probability proportional to the width of the step where the noise falls. The one-dimensional staircase-shaped probability density function (PDF) is defined as follows:

\begin{definition}[Staircase PDF~\cite{geng2015staircase}]
\label{def:staircase}
% a geometric mixture of uniform random variables, symmetric 
Given constants $\gamma, \epsilon \in[0,1]$, we define the PDF $f_\gamma^\epsilon(\cdot)$ at location $\mu$ with $\Delta>0$ as
\small
\begin{align}
    \label{eq:f_gamma} f_\gamma^\epsilon(x\ |\ \mu, \Delta)&=\exp({-l_\Delta(x\ |\ \mu)\epsilon})a(\gamma) \\
    \label{eq:l_delta} l_\Delta(x\ |\ \mu)&=\lfloor \frac{\|x-\mu\|_1}{\Delta}+(1-\gamma) \rfloor
\end{align}
\normalsize
where the normalization factor $a(\gamma)$ ensures $\int_{\BR} f_\gamma^\epsilon(x)\rd x=1$ and $\lfloor \cdot \rfloor$ is the floor function. 
%\BW{$\epislon$ is i} 
Note that we unify and simplify the segmentation function in the original definition. 
% We denote the Staircase noise $\CS_\gamma^\epsilon(x)$ as the distribution which PDF is $f_\gamma^\epsilon(\cdot)$. 
% The original definition of the Staircase mechanism and detailed expression for $a(\gamma)$ can be found in Appendix~\ref{appendix:staircase}.
\end{definition} 

In Text-CRS, we use the Staircase PDF to model the relationship between a word and its synonyms.  
% fit the distribution between a word and its various synonyms. 
Specifically, given a target word, we first compute the cosine similarity between the embedding of itself and its synonyms. 
%and rank the synonyms based on their similarities. 
Then, we define the noise intervals and the substitution probability, where the number of intervals equals the number of synonyms, and the synonyms are symmetrically positioned on the intervals based on their similarities. 
%cosine similarity to the original embedding.
Table~\ref{tab:synonym} shows an example target word ``\emph{good}'' and it has two synonyms \emph{excellent} and \emph{admirable} (the total number of synonyms is $s = 5/\epsilon=5$). %and we ensure that $\epsilon\in [0,1]$ in the evaluation). 
%As shown in Table~\ref{tab:synonym} (when $s=3$), 
Thus, the synonym with the highest cosine similarity, i.e., \textit{good} itself, is placed on the $[-\Delta, \Delta)$ interval, while the synonym with the lowest cosine similarity, i.e., \textit{admirable}, is placed on the $[-5\Delta, -4\Delta)$ and $[4\Delta, 5\Delta)$ intervals. 
%,  noise interval and  
% as shown in Figure~\ref{fig:synonymcorpus}, 
Figure~\ref{fig:synonymcorpus} shows that \textit{good} is replaced by \textit{excellent} with  probability $\exp(-\epsilon)a(\gamma)$, while by \textit{admirable} with probability $\exp(-4\epsilon)a(\gamma)$-- closer relationship between \textit{good} and \textit{excellent} is captured.

\vspace{-0.05in}
\begin{table}[!h]
  \centering
  \setlength\tabcolsep{3pt}
  \caption{Staircase-based synonym substitutions for \textit{good}}
  \vspace{-0.05in}
    \begin{tabular}{llll}
    \toprule
    \begin{tabular}[c]{@{}l@{}}Synonym\\word\end{tabular} & \begin{tabular}[c]{@{}l@{}}Cosine\\similarity\end{tabular}  & \begin{tabular}[c]{@{}l@{}}Noise interval\\(or step)\end{tabular} & \begin{tabular}[c]{@{}l@{}}Substitution\\probability\end{tabular} \\
    \midrule
    \textit{admirable}  &  $0.223$ & $[-5\Delta,-4\Delta)$ & $\exp({-4\epsilon})a(\gamma), \exp({-5\epsilon})a(\gamma)$  \\
    \textit{excellent}  &  $0.788$ & $[-2\Delta,-\Delta)$ & $\exp({-\epsilon})a(\gamma), \exp({-2\epsilon})a(\gamma)$ \\
    \textit{good} &  $1.000$ & $[-\Delta,\Delta)$ & $a(\gamma), \exp({-\epsilon})a(\gamma)$ \\
    \textit{excellent}  &  $0.788$ & $[\Delta,2\Delta)$ & $\exp({-\epsilon})a(\gamma), \exp({-2\epsilon})a(\gamma)$ \\
    \textit{admirable}  &  $0.223$  & $[4\Delta, 5\Delta)$ & $\exp({-4\epsilon})a(\gamma), \exp({-5\epsilon})a(\gamma)$\\
    \bottomrule
    \end{tabular}
    \vspace{-0.1in}
  \label{tab:synonym}%
\end{table}%

\begin{figure}[!h] %H为当前位置，!htb为忽略美学标准，htbp为浮动图形
\centering
\includegraphics[width=0.95\linewidth]{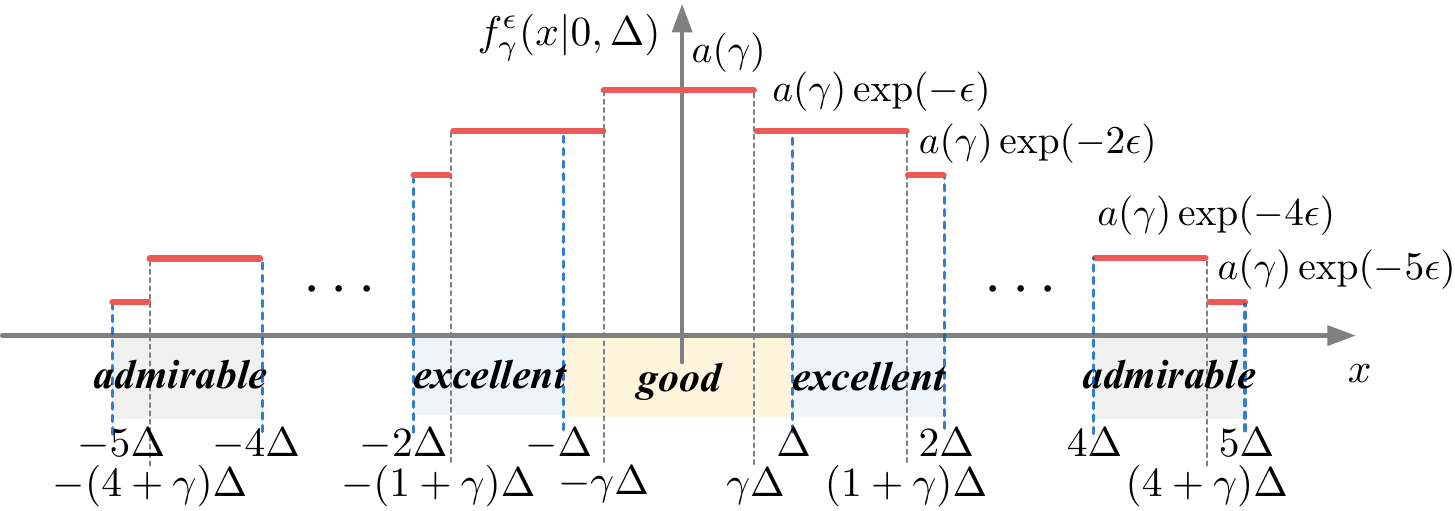} 
\caption{PDF of synonym substitution for the word ``\textit{good}''. The horizontal axis represents the embedding vector of each synonym. The vertical axis shows their probability. 
%\BW{The total number of synonyms for  \textit{good} is $s=3$, and sensitivity $\epsilon=1/s$. Am I correct?} \xy{Yes, I set $\epsilon=5/s$, $\Delta=1$, and $\gamma=1$ then we can obtain $x\in[-s, s]$.}
%\BW{Figure does not correspond to the word ``good"?}\xy{done}
}
\vspace{-0.15in}
\label{fig:synonymcorpus} 
\vspace{-0.05in}
\end{figure}

\vspace{-0.1in}
%\subsubsection{Staircase-based Embedding Substitution} 
\subsubsection{Certification for Synonym Substitution}
The embedding transformation $\phi_S$ is defined by substituting each word $x_i$'s embedding $w_i$ with its $a_i$th synonym's $w_i^{a_i}$. 
% $w_i$ with the $a_i$th synonym. 
$a_i$ decides the substitution of $x_i$, e.g., a closer synonym $w_i^{a_i}$ to $x_i$ has a smaller $a_i$. 
% the magnitude of the transformation $\phi_S$. 
We assume $a_i$ follows a Staircase PDF, in which the probability of each synonym being selected is defined as Table~\ref{tab:synonym}. 
%We set the size of the thesaurus as $s\propto 1/\epsilon$. \BW{$s$ is not used here?}
Then, we provide the below robustness certification for the synonym substitution perturbation $\delta_{S}$.
% , denoted as $\delta_{S}$, which indicates the distance between all the words in a text $x$ and their synonym indexes,  %(including the words themselves),  
% by modeling the distance of synonym indexes, 
% i.e., $ \delta_{S}= \{a_1, \cdots, a_n\}$.
% \begin{equation}
%     \delta_{S}= \{a_1, \cdots, a_n\}. %,\quad \|\delta_{S}\|_1=\sum_{i=1}^{n} a_i
% \end{equation}

% The synonyms are positioned symmetrically on the steps based on their cosine similarity to the original embedding. As shown in Table~\ref{tab:synonym} (when $s=3$), the synonym with the highest cosine similarity, i.e., \textit{good} itself, is placed on the $[-\Delta, \Delta)$ interval, while the synonym with the lowest cosine similarity, i.e., \textit{admirable}, is placed on the $[-3\Delta, -2\Delta)$ and $[2\Delta, 3\Delta)$ intervals. \BW{Define the meaning of $\delta_{S}$ in advance?} \xy{have defined in advance.}

% $h: \BR^d\to \CY$ be any deterministic or random function and let 
\begin{theorem} \label{thm:wss_1}
Let $\phi_S: \CW \times \BR^n \to \CW$ be the embedding substituting transformation based on a Staircase distribution $\varepsilon \sim \CS_\gamma^\epsilon(w, \Delta)$ with PDF $f_\gamma^\epsilon(\cdot)$, and let $g_S$ be the smoothed classifier from any deterministic or random function $h$, as in (\ref{eq:smoothg1}). Suppose $y_A, y_B\in \CY$ and $\underline{p_A}, \overline{p_B} \in [0,1]$ satisfy:
\begin{equation}
\small
\begin{aligned}
    \BP(h(u\cdot \phi_S(w,\varepsilon))=y_A) \geq \underline{p_A} &\geq \overline{p_B} \geq\\
    \max_{y_B\neq y_A}& \BP(h(u\cdot \phi_S(w,\varepsilon))=y_B) 
\nonumber
\end{aligned}
\end{equation}
then $g_S(u\cdot \phi_S(w, \delta_{S}\cdot \Delta))=y_A$ for all $\|\delta_{S}\|_1 \leq \rad_S$, where 
% $\rad_S$ is called the certified radius that satisfies:
%
\begin{equation} \label{eq:rad_s}
\small
\rad_S = \max \Big \{\frac{1}{2\epsilon}\log({\underline{p_A}}/{\overline{p_B}}), -\frac{1}{\epsilon}\log(1-\underline{p_A}+\overline{p_B}) \Big \}. %
\end{equation}
\end{theorem}

\begin{proof}
Proven in Appendix~\ref{proof:wss_1}.
\end{proof}
% \BW{From the proof, certified radius is irrelevant to $\Delta$?} \xy{Yes, as shown in Figure 3, the perturbation added on $w$ is $\delta_{S} = \{a_1, \cdots, a_n\}\cdot \Delta$. I set the interval size ($\Delta=1$) in the evaluation.}

% \BW{Do you study the impact of the interval size $\Delta$? If not we can set $\Delta=1$ by default?}

% \BW{Here the perturbation $\delta$ is imposed on the ``synonym" indexes, but the proof is imposed on $w$? The perturbation definition is unclear to me.} \xy{I will revise the proof}

% The perturbation of the embedding substituting transformation, characterized by $\delta_{S}$, denotes the distance of synonym substitution between the original $w=\{w_1, \cdots, w_n\}$ and the transformed $w'=\{w_1^{a_1}, \cdots, w_n^{a_n}\}$, i.e., 
% \begin{equation}
%     \delta_{S}= \{a_1, \cdots, a_n\} %,\quad \|\delta_{S}\|_1=\sum_{i=1}^{n} a_i
% \end{equation}
Theorem~\ref{thm:wss_1} states that 
%if the $\ell_1$ norm of $\delta_{S}$ is sufficiently small, then $g_S$ exhibits a constant prediction. This means that 
\emph{any} synonym substitution would not succeed as long as the $\ell_1$ norm of $\delta_{S}$ is smaller than $\rad_S$. 
% in (\ref{eq:rad_s}). 
We can observe that the certified radius $\rad_S$ is large under the following conditions: 1) the noise level $\epsilon$ is low, indicating a larger synonym size; 2) the probability of the top class $y_A$ is high, and those of other classes are low. 

% The higher probability of the top class $y_A$ and the lower probability of the other classes are also general conditions for improving the certified radius of Theorem~\ref{thm:wcr_1}, Theorem~\ref{thm:wi_1} and Theorem~\ref{thm:wd_1}.

% \vspace{5pt}
% % \begin{theorem}[binary case]
% % \label{thm:wss_2}
% \noindent\textbf{Theorem 1 (binary case).} 
% \textit{Let $\phi_S: \CW\times\BR\to \CW$ be the embedding transformation based on Staircase noise $\varepsilon \sim f_\gamma(w)$ and let $g_S$ be defined as the smoothed classifier as in (\ref{eq:smoothg1}). Suppose $y_A\in \CY$ and $\underline{p_A} \in (\frac{1}{2},1]$ satisfy: }
% %
% $$
%     \BP(h(u\cdot \phi_S(w, \varepsilon))=y_A) \geq \underline{p_A}
% $$
% \textit{Then $g_S(u\cdot \phi_S(w, \delta_{S}))=y_A$ for all $\|\delta_{S}\|_1 \leq \rad_S\cdot\Delta$ where }
% $$\rad_S=-\frac{1}{\epsilon}\log (2(1-\underline{p_A}))$$
% % \end{theorem}
% \begin{proof}
% We provide a different proof from the one in Theorem~\ref{thm:wss_1}, see Section~\ref{proof:wss_2} in Appendix.
% \end{proof}

% \BW{Binary case is very straightforward from the multi-class case, i.e., by setting $\overline{p_B} = 1 - \underline{p_A}$? The proof looks like very complicated.}
% \xy{Because I first proved the binary case in a complicated way, but only afterward I realized that there is a simple way to prove Theorem 1. Maybe I can delete the proof of the binary case.} \BW{I agree. Setting $\overline{p_B} = 1 - \underline{p_A}$ directly reaches the binary case.} \xy{thanks, I delete the binary case directly}
% \textbf{Certified Word Corpus}

\vspace{-0.1in}
\subsection{Certified Robustness to Word Reordering}
\vspace{-0.1in}
% Word reordering only changes the position of permutation vector $u_i$ while keeping the embedding matrix $w$ unchanged, which is denoted as
% $$\theta_R(\{u_1, \cdots, u_n\}, z)=\{u_{r_1}, \cdots, u_{r_n}\}$$

\subsubsection{Uniform-based Permutation}
%Fixing the relative importance of each position to the prediction can be challenging due to heterogeneous model structure and datasets. 
%\BW{This claim is inaccurate, as a defender definitely can know the model details?} \xy{I have revised the reason.} 
%To address this, 
We assume that each position of the word is equally important to the prediction, and add a uniform distribution to the permutation matrix ($u$) to model permutation. Specifically, we simulate the uniform distribution by grouping the row vectors of $u$, then randomly reordering vector positions within the groups. For example, given a permutation matrix with $n$ row vectors $\{u_i\}$, we divide all $u_i$ uniformly and randomly into $n/\lambda$ groups with length $\lambda$ each. The row vector $u_i$ can be reordered randomly within the group. In this way, the noise added to each position is $1/\lambda$ (uniform). Also, $\lambda=n$ means randomly shuffling all row vectors of the entire permutation matrix. 
%\BW{I am confused about this. If different positions have different levels of importance, why using uniform distribution?} 
%\xy{Maybe there is something wrong with my expression. The different level here is for evaluation that can use multiple levels of noise (Figure 6).}
The proposed uniform smoothing method provides certification for the permutation perturbation $\delta_{R}$.
% , which is the distance between the positions of $1$ in $u_i$ before and after the permutation and can be expressed as 
% \begin{equation}
%     \delta_{R}=\{r_1-1, \cdots, r_n-n\}.
% \label{eq:perturb_r}
% \end{equation}

\begin{theorem}
\label{thm:wcr_1}
Let $\theta_R: \CU\times \BZ^n \to \CU$ be a permutation based on a uniform distribution $\rho\sim \FU[-\lambda, \lambda]$ 
%\BW{This means uniformly sampling from a set of integers -n to n? $\BR$ should be $\BZ$, the integer space?}\xy{Thank you. I have corrected this error.} 
and $g_R$ be the smoothed classifier from a base classifier $h$, as in (\ref{eq:smoothg2}). Suppose $g_R$ assigns a class $y_A$ to the input $u\!\cdot\! w$, and $\underline{p_A}, \overline{p_B}\in(0,1)$. If 
\small
\begin{align}
\BP(h(\theta_R(u, \rho)\cdot w)=y_A) \geq \underline{p_A} &\geq \overline{p_B} \geq \nonumber\\
\max_{y_B\neq y_A}& \BP(h(\theta_R(u, \rho)\cdot w)=y_B))\nonumber
\end{align}
\normalsize
then $g_R(\theta_R(u, \delta_{R})\cdot w)=y_A$ for all permutation perturbations satisfies $\|\delta_{R}\|_1 \leq \rad_R$, where
%\BW{No perturbation $\delta_{R}$ on $g_R$.} \xy{added}
\begin{equation} 
\small
    \rad_R= \lambda(\underline{p_A}-\overline{p_B})
\label{eq:rad_r}
\end{equation}
\end{theorem}

\begin{proof}
Proven in Appendix \ref{proof:wcr_1}.
\end{proof}

Theorem~\ref{thm:wcr_1} states that \emph{any} permutation would not succeed as long as $\delta_{R}< \rad_R$ in Eq.(\ref{eq:rad_r}). We can observe that the certified radius $\rad_R$ is larger when $\lambda$ is higher, which requires more shuffling, or/and $\underline{p_A}$ is larger. The maximum certified radius is $\lambda$, which is the size of a reordering group. Note that both $\underline{p_A}$ and $\lambda$ depend on the noise magnitude. 
%\BW{any parameter relevant to the noise level?}

%\BW{Since $\lambda=n$ shows the largest certified radius, why not just setting $\lambda=n$ directly, meaning shuffling all rows? This also does not involve a new hyperparameters?} \xy{If we directly set $\lambda=n$, can't we change the noise level? I set $\lambda=n, n/2, n/4$ in the evaluation.} \BW{I see. So $\underline{p_A}$ depends on the noise level?} \xy{Yes, $\underline{p_A}$ and $\lambda$ depend on the noise level.}

\vspace{-0.1in}
\subsection{Certified Robustness to  Word Insertion}
\vspace{-0.1in}
% Word insertion can be considered as the combination of permutation $\theta_I$ and embedding transformation $\phi_I$. Inserting a word at $j_i$th ($i\in[m']$) position can be expressed as 
% \begin{equation}
% \begin{aligned}
% \theta_I(\{&u_1, \cdots, u_n\}, z)= \\
% &\{u_1, \cdots, u_{j_1-1}, u_{j_1+1}, \cdots, u_{n}, u_{j_1}, \cdots , u_{j_m'}\} \\
% \phi_I(\{&w_1, \cdots, w_n\}, t)= \\
% &\{w_1, \cdots, w_{j_1-1}, w_{j_1}, \cdots, w_{n-m'}, w_{in}^{1*}, \cdots, w_{in}^{m'*}\} 
% \nonumber
% \end{aligned}
% \end{equation}

Word insertion employs a combination of permutation $\theta_I$ and embedding transformation $\phi_I$. Recall that the only transformation performed on the permutation matrix is to shuffle the position of $u_i$. Hence, we utilize the uniform-based permutation with the noise level set as the number of words ($n$), i.e., $\theta_I(u, \rho)= \theta_R(u, \rho)$, where $\rho \sim \FU[-n, n]$. 
% and the permutation perturbation follows the definition in Eq.(\ref{eq:perturb_r}) ($\delta_{R}\coloneqq \delta_{R}$). 
On the other hand, embedding insertion involves replacing $w_{n-m'+i}$ with an unrestricted inserted embedding $w_{in}^{j}$, which sets it apart from synonym substitution. To address the challenge of unrestricted embedding insertion, we propose a Gaussian-based smoothing method for the certification.
% , i.e., the embedding distance between the original text and the inserted text as follows 
% \begin{equation}
%     \delta_{I}=\{ w_{in}^{1}-w_{n-m'+1}, \cdots, w_{in}^{m'}-w_{n} \}.
% \end{equation}

% The definition of $\delta_{R}$ is the same as $\delta_{R}$ in (\ref{eq:perturb_r}). And we can obtain $\|\delta_{R}\|_1 = 2\times(\min(m, n)-j)$, $\|\delta_{R}\|_1 = \sum_{i=1}^{m'-1} j_{m'}-j_i$ according to the following analysis. 
% In reality, when $m\leq n$, there is no $w_{\mathtt{<pad>}}$ at the end of $w$, then $u_j$ moves to the $n$th position; when $n>m$, there are $n-m$ $w_{\mathtt{<pad>}}$ at the end of $w$, $u_j$ just need to move to the $(m+1)$th position.

% This method can provide provable robustness as long as both $\delta_{R}$ and $\delta_{I}$ satisfy the certified condition, and we provide 

\vspace{-0.1in}
\subsubsection{Gaussian-based Embedding Insertion} 
We consider the embedding matrix as a whole, the length of which is $n\times d$, where $d$ is the dimension of each embedding vector. We add Gaussian noise to the embedding matrix directly, similar to adding independent identically distributed Gaussian noise to each pixel of an image. We invoke Theorem 1 in \cite{cohen2019certified} as

\begin{theorem} %[Cohen'19\cite{cohen2019certified}]
\label{thm:wi_1}
Let $\phi_I:\CW\times \BR^{n\times d} \to \CW$ be the embedding insertion transformation based on Gaussian noise $\varepsilon \sim \CN(0,\sigma^2 I)$ and $\theta_I$ be the perturbation as same as $\theta_R$ based on a uniform distribution $\rho\sim \FU[-n, n]$. Let $g_I$ be the smoothed classifier from a base classifier $h$ as in (\ref{eq:smoothg}), and suppose $y_A,y_B\in \CY$ and $\underline{p_A}, \overline{p_B} \in [0,1]$ satisfy:
\begin{equation}
\small
\begin{aligned}
\BP(h(\theta_I(u,\rho)\cdot\phi_I(w,\varepsilon)))=&\ y_A) \geq \underline{p_A} \geq \overline{p_B} \geq \\
\max_{y_B\neq y_A}& \BP(h(\theta_I(u,\rho)\cdot\phi_I(w,\varepsilon))=y_B))
\nonumber
\end{aligned}
\end{equation}
then $g_I(\theta_I(u,\delta_{R})\cdot \phi_I(w,\delta_{I}))=y_A$ for all $\|\delta_{R}\|_1 < \rad_R$ as in Eq.(\ref{eq:rad_r}) and $\|\delta_{I}\|_2 < \rad_I$ where 
\begin{equation}\small
    \rad_I =\frac{\sigma}{2}(\Phi^{-1}(\underline{p_A})-\Phi^{-1}(\overline{p_B}))
\label{eq:rad_i}
\end{equation}
\end{theorem}

% \begin{proof}
% The proof is the same as that in~\cite{cohen2019certified}.
% \end{proof}

Theorem~\ref{thm:wi_1} states that $g_I$ can defend against word insertion transformation as long as the conditions about $\delta_{R}$ and $\delta_{I}$ are satisfied simultaneously. We observe that the certified radius $\rad_I$ is large when the noise level $\sigma$ is high.

\vspace{-0.1in}
\subsubsection{Combination of Two Perturbations} 
The certified radii of Eq.(\ref{eq:rad_r}) and Eq.(\ref{eq:rad_i}) ensure the robustness of the smoothed classifier against permutation $\theta_I$ or embedding insertion $\phi_I$, respectively. However, the word insertion is a combination of $\theta_I$ and $\phi_I$. Next, we propose Theorem~\ref{thm:combination} to provide certified robustness for the combination of them. 

%\BW{I do not quite understand the proof for Theorem 4...}
%\xy{I have made some revisions to Theorem 4 and its proof. Do you think my writing is clearer now?}

\begin{theorem} \label{thm:combination}
If a smoothed classifier $g$ is certified robust to permutation perturbation $\delta_u$ as defined in Eq.(\ref{eq:u_delta}), and to embedding permutation $\delta_w$ as defined in Eq.(\ref{eq:w_delta}), then 
%we can conclude that 
$g$ can provide certified robustness to the combination of perturbations $\delta_u$ and $\delta_w$ as defined in Eq.(\ref{eq:uw_delta}) assuming $\theta(u, \rho)$ is uniformly distributed in the permutation space.
%Formally, 

\vspace{-0.15in}

\small
\begin{align}
    \label{eq:u_delta} 
    & \forall \delta_u, \delta_w, 
    g(\theta(u\!+\!\delta_u,\rho)\cdot\! \phi(w, \varepsilon))\! =\! g(\theta(u,\rho)\!\cdot\! \phi(w, \varepsilon)) \\
    \label{eq:w_delta} 
    & \qquad \, \, \, \& \, 
    g(\theta(u,\rho)\!\cdot\! \phi(w\!+\!\delta_w, \varepsilon))\! =\! g(\theta(u,\rho)\!\cdot\! \phi(w, \varepsilon)) \\
    \label{eq:uw_delta} 
    & \implies 
    g(\theta(u\!+\!\delta_u,\rho)\!\cdot\! \phi(w\!+\!\delta_w, \varepsilon))\! =\! g(\theta(u,\rho)\!\cdot\! \phi(w, \varepsilon))
\end{align}
\normalsize

\end{theorem}

\begin{proof}
    Proven in Appendix~\ref{proof:combination}.
\end{proof}

% \noindent \textbf{Remark:} Permutation noise $\rho$ and embedding noise $\varepsilon$ are applied to $u$ and $w$ during training, and then certified radii $\rad_R$ and $\rad_I$ are calculated in two spaces using Eq.(\ref{eq:rad_r}) and Eq.(\ref{eq:rad_i}). According to 

Thus, when the certified radii $\rad_R$ and $\rad_I$ are derived w.r.t. $\delta_u$ and $\delta_w$ in two spaces using Eq.(\ref{eq:rad_r}) and Eq.(\ref{eq:rad_i}), respectively, then $\theta(u, \delta_u)\cdot \phi(w, \delta_w)$ is within the certified region of $g$ as long as both $\delta_u<\rad_R$ and $\delta_w<\rad_I$ hold. 

\vspace{-0.1in}
\subsection{Certified Robustness to Word Deletion}
\vspace{-0.1in}
% Word insertion also combines permutation $\theta_D$ and embedding transformation $\phi_D$. Deleting a word at the $j_i$th ($i\in[m']$) position can be expressed as
% \begin{equation}
% \begin{aligned}
% \theta_D(u,z)=&\{u_1, \cdots, u_{j_1-1}, {u_{n-m'+i}}, u_{j_1}, \cdots, u_{n-m'}\} \\ %\{u_1, \cdots, u_n\}
% \phi_D(w,t)=&\{w_1, \cdots, w_{j_1-1}, {w_{\texttt{pad}}^1}, w_{j_1+1}, \cdots, w_n\} %\{w_1, \cdots, w_n\}
% \nonumber
% \end{aligned}
% \end{equation}

Similarly, a completely uniform shuffling of the orders is performed on permutation, denoted as $\theta_D(u, \rho) = \theta_R(u, \rho)$, where $\rho$ is drawn from $\FU[-n, n]$. Recall that embedding deletion is modeled as a change in the embedding state from $b_i=1$ to $b_i=0$. To certify against the $\ell_0$-norm of embedding deletion perturbation (as the number of deleted words), we propose a Bernoulli-based smoothing method.

% \subsubsection{Bernoulli Distribution} 
\vspace{-0.1in}
\subsubsection{Bernoulli-based Embedding Deletion} 
The transition of the embedding state ($a_i$) can be considered to follow the Bernoulli distribution $\FB(n,p)$, where $n$ is the total number of words in the text. Each word embedding can be transformed to $w_{\texttt{<pad>}}$ with probability $p$ and maintained with $1-p$, i.e, $\BP(b_i=1\to b_i=0)=p$ and $\BP(b_i=1\to b_i=1)=1-p$. 
\emph{The $w_{\texttt{<pad>}}$ cannot be transformed into a word embedding since we cannot recover the deleted words when they are removed from the text.} It can be denoted as $\BP(b_i=0\to b_i=1)=0$.

\begin{theorem} %[Cohen'19\cite{cohen2019certified}]
\label{thm:wd_1}
Let $\phi_D:\CW\times \{0,1\}^n \to \CW$ be the embedding deletion transformation based on Bernoulli distribution $\varepsilon \sim \FB(n,p)$ and $\theta_D$ be the perturbation same as $\theta_R$ based on a uniform distribution $\rho\sim \FU[-n,n]$. Let $g_D$ be the smoothed classifier from a base classifier $h$, as in (\ref{eq:smoothg}) and suppose $y_A,y_B\in \CY$ and $\underline{p_A}, \overline{p_B} \in [0,1]$ satisfy:
\begin{equation}\small
\begin{aligned}
\BP(h(\theta_D(u,\rho)\cdot\phi_D(w,\varepsilon)))& =y_A) \geq \underline{p_A} \geq \overline{p_B} \geq \\
\max_{y_B\neq y_A}& \BP(h(\theta_D(u,\rho)\cdot\phi_D(w,\varepsilon))=y_B))
\nonumber
\end{aligned}
\end{equation}
then $g_D(\theta_D(u,\delta_{R})\cdot \phi_D(w,\delta_{D}))=y_A$ for all $\|\delta_{R}\|_1 < \rad_R$ as in Eq.(\ref{eq:rad_r}) and $\|\delta_{D}\|_0 < \rad_D$, where %$\rad_D$ is the solution to the following optimization problem:
\begin{equation}
\small
\begin{aligned}
    & \rad_D= \argmax \delta, \\
    & \textrm{\ s.t.\ } \tbinom{z_{\max}}{\delta} \leq \underline{p_A}/\overline{p_B}, \\
    & \qquad z_{\max} \!=\! \argmax z, \textrm{\ s.t.\ } \tbinom{n}{z} p^z (1-p)^{(n-z)} \!\leq\! \overline{p_B}.
\end{aligned}
\end{equation}
\end{theorem}

\begin{proof}
Proven in Appendix~\ref{proof:wd_1}.
\end{proof}

Theorem~\ref{thm:wd_1} states that $g_D$ can defend against any word deletion transformation as long as the conditions about $\delta_{R}$ and $\delta_{D}$ are met. We observe that the certified radius $\rad_D$ is large when the total number of words $n$ is small.  % and $p_A$ are large 

% The same as Word Substitution Perturbation, Each word has the probability of $\exp({−k\epsilon})a(\gamma)$ to be deleted and the probability of $\exp({−(k+1)\epsilon}) a(\gamma)$ to remain unchanged. 

\vspace{-0.1in}
\subsection{Universality of Certification for Insertion}
\vspace{-0.1in}

All four adversarial operations are essentially transformations of the embedding vector. Hence, our certification for the word insertion,  a combination of uniform-based permutation ($\theta_I$) and Gaussian-based embedding transformation ($\phi_I$), is applicable to all the four operations. The synonym substitution operation only employs the embedding transformation ($\theta_I$), with the embedding perturbation being the sum of the embedding distances of the replaced synonyms. Word reordering is a simple version of word insertion, using only permutation ($\theta_I$). Word deletion, on the other hand, uses both permutation ($\theta_I$) and embedding transformation ($\theta_I$), with its embedding perturbation being the sum of the embedding distances of the deleted words.

\vspace{-0.05in}
\section{Practical Algorithms}
\vspace{-0.1in}

\subsection{Training}
\vspace{-0.1in}
Given the word operation $T\in\{S, R, I, D\}$, we aim to generate a certified model against the corresponding word-level attack. As described in Algorithm~\ref{alg:train}, we first generate a set of embedding matrices $w$ with the pre-trained embedding layer $L_{emb}$. The permutation matrix $u$ is an identity matrix with the same length as $w$ (line~1). Then, we perform permutation and embedding transformations on $u$ and $w$ to generate the dataset $\CD_T$. Finally, we update the model with the training dataset $\CD_T$ and obtain the model $h$.

%\BW{Inject noises in each epoch not just once?}\xy{I only inject noise once during training. I was wondering If there might be any errors in the sentences above?} 
%\BW{Okay, I see. No errors. I thought injecting noise per epoch, as this is what I did before.}

\begin{algorithm}[!h]\small
\caption{Training algorithm} \label{alg:train}
\begin{algorithmic}[1]
\Require Training dataset $\CD=\{(x,y)_i\}$, operation $T$, pre-trained embedding layer $L_{emb}$, permutation $\theta_T$ with noise $\rho$, embedding transformation $\phi_T$ with noise $\varepsilon$
% \Ensure $y = x^N$
\State $u\cdot w\gets L_{emb}(x)$ \Comment{$u$ is $w$'s permutation matrix}
\State $\CD_T=\{(\theta_T(u,\rho)\cdot \phi_T(w,\varepsilon), y)_i\}$
\State $h \gets$ Train the classification model with $\CD_T$
\State \Return Classification model $h$
\end{algorithmic}
\end{algorithm}

\vspace{-0.15in}

\subsubsection{Enhanced Training Toolkit for Word Insertions} \label{sec:enhance}

High-level Gaussian noise leads to distortion in the embedding space and results in barely model convergence. To address this issue, we develop a toolkit with three methods for the three steps in the training (see Figure~\ref{fig:toolkit}). The toolkit is mainly used to improve the certified accuracy against word insertions, and it is also applicable to other operations. 

\begin{figure}[]
\centering
\includegraphics[width=0.9\linewidth]{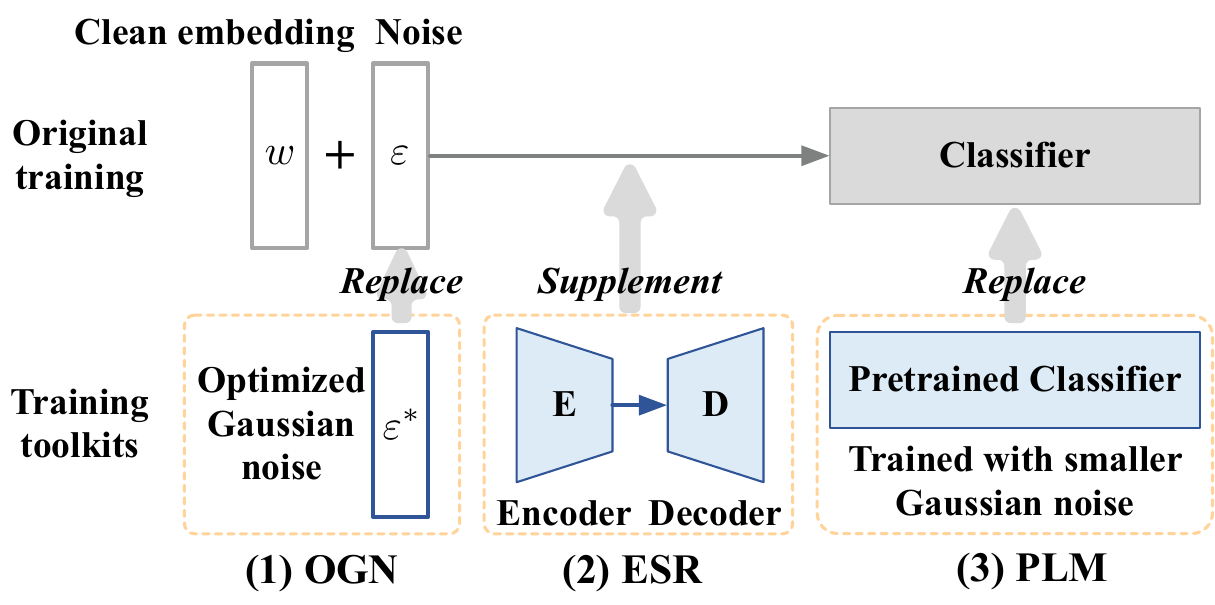} 
\vspace{-1mm}
\caption{The training toolkit can enhance model accuracy by replacing or supplementing part of the training process.} \vspace{-0.1in}
\label{fig:toolkit} 
\vspace{-2mm}
\end{figure}
% The training toolkit can enhance the accuracy of a model (i.e., classifier) by either replacing or supplementing a part of the original training process.

\vspace{0.03in}
\noindent\ding{172} \textbf{Optimized Gaussian Noise (OGN).} Inspired by the Anisotropic-RS~\cite{hong2022certified}, an appropriate mean value of Gaussian noise can improve the certified accuracy. We analyze the embedding vectors of all words and observe that each element in the embedding vectors approximates a Gaussian distribution with a nonzero mean, as illustrated in Figure~\ref{fig:emb_values}. Consequently, we can enhance the certified accuracy by modifying the Gaussian noise of each dimension $\CN(0, \sigma I^2) \to \CN(\mu_i, \sigma I^2)$, where $\mu_i, \ i\in[1,d]$ denotes the average of the original embedding space of each dimension. 
% \BW{Unclear about $\mu$. Is $\mu$ a vector of a scalar? For context, each dimension $i$ has its mean $\mu_i$} \xy{I counted the entire thesaurus and set up $300$ $\mu_i, i\in[300]$. I did not set specific $\mu_i$ for each context.}

% ~\cite{wang2021natural}

\vspace{0.03in}
\noindent\ding{173} \textbf{Embedding Space Reconstruction (ESR).} To mitigate the disturbance of the embedding space, we introduce an encoder-decoder architecture to reconstruct the clean embedding space. The encoder-decoder can also be viewed as sanitizing the additive noise. This method can effectively improve accuracy in small-dimension embedding space, such as with the $300$-dimension GloVe embedding and LSTM classifier, resulting in a $10\%$ increase in average accuracy. 
% For large embedding space, a guiding loss should be added to guide the embedding space reconstruction. We refer to the Information Bottleneck-based strategy~\cite{tian2021farewell} and introduce an additional loss of the clean embedding vector to guide the training, i.e., $\textrm{loss}= \ell(h(w+\varepsilon), y)+\alpha\cdot \ell(h(w),y)$, where $\alpha$ is the balanced weight. 

% However, for large pre-trained models, such as BERT, we do not have much data to pre-train the added encoder-decoder structure, so this method is not applicable.

\vspace{0.03in}
\noindent\ding{174} \textbf{Pre-trained Large Model (PLM).} Fine-tuning from a pre-trained model is a typical training approach for large models. When applying high-level Gaussian noise to a large model, we can fine-tune it on a pre-trained large model trained with small Gaussian noise (e.g., $\sigma=0.1$). For instance, when adding Gaussian noise of $\sigma = 1.5$ to the IMDB dataset and using BERT as the classifier, this approach can substantially enhance model accuracy from $50\%$ to $84\%$.

\vspace{-0.1in}
\subsection{Certified Inference}  % Prediction and Certification
\vspace{-0.1in}

The certified inference algorithm is identical to the classical randomized smoothing in \cite{cohen2019certified}. We first obtain the embedding $u\cdot w$ of the test sample $x$ by the pre-trained embedding layer $L_{emb}$. Then, we utilize $\theta_T(u, \rho)\cdot \phi_T(w, \varepsilon)$ to draw $N$ samples by the $T$ transformation. Finally, we certify robustness on $N$ samples and output the robust prediction.
Details are presented in Algorithm~\ref{alg:certify} in Appendix~\ref{appendix:alg}.
%  prediction and Algorithms~\ref{alg:predict}

\vspace{-0.05in}
\section{Experiments}
\label{sec:exps}
\vspace{-0.1in}

\subsection{Experimental Setup}
\label{sec:setup}
\vspace{-0.1in}

\noindent\textbf{Datasets}.
We evaluate Text-CRS on three textual datasets, AG’s News (AG)~\cite{zhang2015character}, Amazon~\cite{mcauley2013hidden}, and IMDB~\cite{maas2011learning}. The AG dataset collects news articles (sentence-level), covering four topic classes. The Amazon dataset consists of positive and negative product reviews (document level). %Reviews with a rating of 1 or 2 are labeled negative sentiments, and those with 4 or 5 are labeled positive. 
The IMDB dataset contains document-level movie reviews with positive and negative sentiments. %reviews scoring $\geq 7$ and negative reviews scoring $\leq 4$. 
The average sample lengths of them are $43$, $81$, and $215$, respectively. 
% Dataset statistics are summarized in Table~\ref{tab:dataset}. (see Table~\ref{tab:dataset})

\vspace{0.03in}

\noindent\textbf{Models and Embedding Layers}.
We conduct experiments on two common NLP models, LSTM \cite{LSTM} and BERT \cite{devlin-etal-2019-bert} with the pre-trained embedding layers. For LSTM, we use a 1-layer bidirectional LSTM with $150$ hidden units and the $300$-dimensional Glove word embeddings trained on $42$ billion tokens of web data from Common Crawl \cite{pennington2014glove}. For BERT, we use a pre-trained $12$-layer bert-base-uncased model with $12$ attention heads. The pre-trained embedding layer in BERT outputs $768$-dimensional hidden features for each token. We freeze the pre-trained embedding layer in both LSTM and BERT and only update the parameters of the classification models.

%\BW{Which embedding/pre-trained models are used?}\xy{added}
\vspace{0.03in}

\noindent\textbf{Evaluation Metric}. We report the model accuracy on the clean test set for vanilla training (\emph{Clean vanilla}) and certified robust training (\emph{Clean Acc.}) (see Table~\ref{tab:clean_acc}). Under robust training, we also evaluate the \emph{certified accuracy}, defined as the fraction of the test set classified correctly and certified robust. 
% We compare certified accuracy under different radii. 
{We uniformly select $500$ examples from the clean test set of each dataset. For each example, we use $N_0=100$ samples for selecting the most likely class $y_A$ and $N=10^5$ samples for estimating confidence lower bound $\underline{p_A}$.} We set $\alpha=0.001$ for certification with at least 99.9\% confidence. 
To test the model's robustness against unseen attacks,
we evaluated Text-CRS against five real-world attacks (TextFooler\cite{jin2020bert}, WordReorder\cite{moradi2021evaluating}, SynonymInsert\cite{morris2020textattack}, BAE-Insert\cite{garg2020bae}, and InputReduction\cite{feng2018pathologies} w.r.t. four word-level adversarial operations)\footnote{BAE-Insert attacks on BERT, and other attacks are model-agnostic. 
Similar to Jin et al.\cite{jin2020bert}, we use Universal Sentence Encoder~\cite{cer2018universal} to encode text as high-dimensional vectors and constrain the similarity of the adversarial vector to the original vector to ensure that all generated adversarial examples are semantically similar to the original texts. 
}, and generate $1,000$ successful adversarial examples for each dataset and model. We calculate the \emph{attack accuracy} of the vanilla model under attack as the percent of unsuccessful adversarial examples divided by the number of attempted examples (see Table~\ref{tab:attack_acc}). For Text-CRS against these attacks, we uniformly select 500 successful adversarial examples for each attack and evaluate the \emph{certified accuracy} with $N=2\times 10^4$. 
Since the baselines can only certify against substitution operations, we evaluate their \emph{certified accuracy} against TextFooler. For the other attacks, we evaluate their \emph{empirical accuracy}, the percent of adversarial examples correctly classified (no certification).
 % with $N=2\times 10^4$. 

% and evaluate the impact of different $N$ on certified accuracy
% Following the setup in~\cite{cohen2019certified}, 
% report the attack performance to the vanilla training model. 

\vspace{0.03in}

\noindent\textbf{Noise Parameters}. 
We use three levels of noise (i.e., Low, Medium (Med.), and High) for the four smoothing methods (see Table~\ref{tab:params}). Synonym substitution based on the Staircase PDF uses the size of the lexicon ($s$) to control the PDF's sensitivity, i.e., $\epsilon=5/s$. For other parameters in the staircase PDF, we fix the interval size of a word as $\Delta=1$ and set an equal probability within each interval, i.e., $\gamma=1$. Uniform-based permutation specifies the noise with the length of the reordering group, i.e., the noise PDF of $\CU[-\lambda, \lambda]$ is $1/2\lambda$. 
% In the low level, $2\lambda=n/4$, which means we randomly and uniformly divide the entire text into four groups. 
While using uniform permutation in word insertion and deletion, we set the noise level to $n$, i.e., reordering the entire text. Gaussian-based embedding insertion uses the standard deviation $\sigma$ to specify the noise. We set different $\sigma$ in LSTM and BERT models due to different embedding dimensions and magnitudes. %(which can tolerate different levels of Gaussian noise). 
Bernoulli-based smoothing uses the deletion probability $p$ of each word as the noise.

\begin{table}[!t]
%\vspace{-0.1in}
\centering
\setlength\tabcolsep{3pt}
\caption{Noise parameters for adversarial operations}
\vspace{-0.05in}
\begin{tabular}{|c|c|c|cc|cc|}
\hline
Operations & Substitution & Reordering & \multicolumn{2}{c|}{Insertion} & \multicolumn{2}{c|}{Deletion} \\ \hline
Noise & $s$ & $2\lambda$ & $2\lambda$ & $\sigma$(LSTM, BERT) & $2\lambda$ & $p$  \\ \hline
Low & $50$ & $n/4$ & $n$ & \ $0.1$,\quad $0.5$  & $n$ & $0.3$ \\
Med. & $100$ & $n/2$ & $n$ & \ $0.2$,\quad $1.0$ & $n$ & $0.5$  \\
High & $250$ & $n$ & $n$ & \ $0.3$,\quad $1.5$ & $n$ & $0.7$  \\ \hline
\end{tabular}
\label{tab:params}
\vspace{-4mm}
\end{table}

\vspace{0.03in}

\noindent\textbf{Training Toolkit}. 
\ding{172} In OGN, we use the average of the parameters of the pre-trained embedding layers (i.e., Glove word embeddings for LSTM and BERT's pre-trained embedding layer for BERT) as the mean value of Gaussian noise in each dimension ($\mu_i$). 
\ding{173} In ESR, we utilize two fully-connected layers as the encoder-decoder for LSTM. 
\ding{174} In PLM, we set the learning rate to $0.00003$, training epochs to $10$, and the Gaussian noise level to $\sigma=0.1$ for training the pre-trained BERT model with Gaussian noise.

%\noindent\textbf{Real-world Adversarial Attacks}. 
\vspace{0.03in}

\noindent\textbf{Baselines}. 
We compare our methods with two SOTA randomized smoothing-based certified defenses, SAFER~\cite{ye2020safer} and CISS~\cite{zhao2022certified}. 1) \emph{SAFER} certifies against synonym substitution. 
% We set the noise level (i.e., the size of the lexicon) in SAFER the same as our substitution method.
Following the same setting~\cite{ye2020safer,jia2019certified} for the synonym substitution, we construct synonym sets by the cosine similarity of Glove word embeddings~\cite{pennington2014glove} and sort all synonyms in the descending order of similarity.
2) \emph{CISS} is an IBP and randomized smoothing-based method against word substitution. CISS provides only the training and certification pipelines for the BERT model. Thus, we only compare our Text-CRS with CISS under the BERT model. 
% We set the noise parameter as 

% For fair comparisons, we used the same training pipeline for each method. We train LSTM from scratch (without a pre-trained classification model) for 50 epochs with a learning rate of $lr=3\times 10^{-5}$ and BERT from the bert-base-uncased model for 50 epochs with $lr=1\times 10^{-3}$. 

% \subsubsection{Code and implementation}
% We implement our framework based on PyTorch. 
% All experiments were run on 16-core Intel Xeon Gold 6226R CPUs and NVIDIA GeForce RTX 3090 GPUs with 24GB RAM. 
% TextAttack\cite{morris2020textattack}

% Substitution, reordering, and deletion operations are directly trained with noise, without employing any training toolkits. Note that toolkits can further improve their certification performance.

Note that the smoothed models against substitution, reordering, and deletion operations are trained without the use of training toolkits. The training toolkits can further improve their certification performance.

% Using toolkits can improve accuracy even more.

\vspace{-0.1in}
\subsection{Experimental Results}
\label{sec:results}
\vspace{-0.1in}

\subsubsection{Evaluating Adversarial Examples using ChatGPT}

We assess the practical significance of textual adversarial examples by evaluating whether ChatGPT (version ``gpt-3.5-turbo-0301'') can determine if a successful adversarial example is semantically similar to the original text. Specifically, we request ChatGPT to assess the semantic similarity (categorized as ``Yes'' or ``No'') and calculate the cosine similarity between the adversarial and original texts. Table~\ref{tab:chatgpt_results} shows the deception rate, which represents the percent of successful adversarial examples that ChatGPT considered semantically similar to the original texts, as well as the cosine similarity of total successful examples, and the cosine similarity of the examples that deceive ChatGPT. The results show an average of $73\%$ of all successful adversarial examples can deceive ChatGPT, highlighting the severe threat posed by textual adversarial examples. Furthermore, we observe that longer texts, such as those in Amazon and IMDB, are easier to fool ChatGPT, due to the challenges in detecting small perturbations in longer texts. Finally, the cosine similarities of all our successful adversarial examples are high, particularly in the case of SynonymInsert and BAE-Insert attacks, unveiling that word insertion can effectively fool the vanilla model with minimal perturbations. 
% and the higher the similarity, the better the effect of deceiving ChatGPT.

\begin{table}[!h]
\setlength\tabcolsep{3pt}
\centering
\scriptsize
\caption{Evaluation on adversarial examples by ChatGPT.}
\vspace{-0.05in}
% \resizebox{\linewidth}{!}{
\begin{tabular}{lcccccc}
\toprule
Metric & \multicolumn{3}{c}{Deception rate} & \multicolumn{3}{c}{Cosine similarity (success / deception)} \\
\cmidrule(r){2-4} \cmidrule(r){5-7}
Attack type & AG & Amazon & IMDB & AG & Amazon & IMDB \\
\midrule
TextFooler & 58\% & 78\% & 85\% & 0.92 / 0.97 & 0.81 / 0.94 & 0.97 / 0.99 \\
WordReorder & 63\% & 53\% & 60\% & 0.84 / 0.95 & 0.79 / 0.91 & 0.90 / 0.97 \\
SynonymInsert & 75\% & 89\% & 85\% & 0.97 / 0.98 & 0.95 / 0.98 & 0.98 / 0.99 \\
BAE-Insert & 66\% & 88\% & 71\% & 0.96 / 0.97 & 0.91 / 0.97 & 0.97 / 0.99 \\
InputReduction & 75\% & 78\% & 73\% & 0.89 / 0.96 & 0.85 / 0.93 & 0.94 / 0.98 \\
\midrule
Average & 67\% & 77\% & 75\% & 0.92 / 0.97 & 0.86 / 0.95 & 0.95 / 0.98 \\
\bottomrule
\end{tabular}
% }
\vspace{-1mm}
\label{tab:chatgpt_results}
\end{table}

\vspace{-0.1in}
\subsubsection{Certified Robustness of Text-CRS}

\begin{table}[]
\centering
\setlength\tabcolsep{2pt}
\scriptsize
\caption{Certified accuracy under adversarial operations. We compare the substitution with SAFER\cite{ye2020safer} and CISS\cite{zhao2022certified}, and provide a benchmark for other operations.}
\vspace{-0.05in}
\resizebox{\linewidth}{!}{
\begin{tabular}{ccccccccc}
\toprule
\multirow{2}{*}{\begin{tabular}[c]{@{}c@{}}Dataset\\ (Model) \end{tabular}} & \multirow{2}{*}{\begin{tabular}[c]{@{}c@{}} \emph{Clean}\\ \emph{vanilla}  \end{tabular}} & \multirow{2}{*}{\begin{tabular}[c]{@{}c@{}} Noise  \end{tabular}} & \multicolumn{3}{c}{\begin{tabular}[c]{@{}c@{}} Synonym substitution\end{tabular}} & \begin{tabular}[c]{@{}c@{}} Reordering\end{tabular} & \begin{tabular}[c]{@{}c@{}} Insertion\end{tabular} & \begin{tabular}[c]{@{}c@{}} Deletion\end{tabular} \\
\cmidrule(r){4-6} \cmidrule(r){7-7} \cmidrule(r){8-8} \cmidrule(r){9-9}
 & & & SAFER & CISS & Ours & Ours & Ours & Ours \\
\midrule
\multirow{3}{*}{\begin{tabular}[c]{@{}c@{}}AG\\ (LSTM)\end{tabular}} & & Low & 86.4\% &  & \textbf{88.8\%} & \textbf{92.4\%} & \textbf{88.6\%} & \textbf{91.2\%} \\
& 91.79\% & Med. & 85.2\% & - & 88.6\% & 91.6\% & 88.2\% & 90.2\% \\
& & High & 83.2\% &  & 87.0\% & \textbf{92.4\%} & 82.8\% & 88.4\% \\
\midrule
 \multirow{3}{*}{\begin{tabular}[c]{@{}c@{}}AG\\ (BERT)\end{tabular}} &  & Low & 92.0\% & 85.6\% & \textbf{92.8\%} & \textbf{95.6\%} & \textbf{93.6\%} & \textbf{94.6\%} \\
& 93.68\% & Med. & 89.2\% & 86.8\% & 92.0\% & 93.6\% & 93.0\% & 93.3\% \\
& & High & 87.2\% & 85.6\% & 91.2\% & 94.4\% & 91.2\% & 92.8\% \\
 \midrule
 \multirow{3}{*}{\begin{tabular}[c]{@{}c@{}}Amazon\\ (LSTM)\end{tabular}} &  & Low & 82.8\% &  & 82.6\% & 87.2\% & \textbf{85.2\%} & \textbf{88.8\%} \\
& 89.82\% & Med. & 80.0\% & - & 82.4\% & 85.8\% & 79.0\% & \textbf{88.8\%} \\
& & High & 80.2\% &  & \textbf{83.6\%} & \textbf{88.4\%} & 75.6\% & 87.2\% \\
\midrule
\multirow{3}{*}{\begin{tabular}[c]{@{}c@{}}Amazon\\ (BERT)\end{tabular}} &  & Low & 90.2\% & 84.4\% & \textbf{93.6\%} & \textbf{94.8\%} & \textbf{94.4\%} & \textbf{93.8\%} \\
& 94.35\% & Med. & 86.4\% & 83.4\% & 90.2\% & 93.6\% & 92.6\% & 91.8\% \\
& & High & 86.2\% & 83.2\% & 87.6\% & 93.8\% & 89.0\% & 88.6\% \\
 \midrule
\multirow{3}{*}{\begin{tabular}[c]{@{}c@{}}IMDB\\ (LSTM)\end{tabular}} &  & Low & 77.8\% &  & \textbf{84.0\%} & \textbf{88.8\%} & \textbf{82.2\%} & \textbf{86.0\%} \\
& 86.17\% & Med. & 77.6\% & - & 81.6\% & 84.6\% & 79.2\% & 85.4\% \\
& & High & 77.0\% &  & 78.8\% & 83.4\% & 71.4\% & 84.8\% \\
\midrule
\multirow{3}{*}{\begin{tabular}[c]{@{}c@{}}IMDB\\ (BERT)\end{tabular}} &  & Low & 86.4\% & 84.8\% & \textbf{89.6\%} & 92.2\% & \textbf{90.8\%} & \textbf{91.4\%} \\
& 91.52\% & Med. & 82.8\% & 82.4\% & 83.6\% & \textbf{92.4\%} & 88.4\% & 90.2\% \\
& & High & 75.8\% & 84.0\% & 78.8\% & 91.8\% & 84.0\% & 89.2\% \\
\bottomrule
\end{tabular}
}
\vspace{-4mm}
\label{tab:certi_acc_benchmark}
\end{table}

Table~\ref{tab:certi_acc_benchmark} summarizes the certified accuracy of Text-CRS against four word-level operations on different datasets, models, and noise. For synonym substitution, we use the same synonym set and noise parameters as SAFER. The results demonstrate that Text-CRS outperforms SAFER for all noise levels under three datasets and two models. Specifically, Text-CRS is more accurate than SAFER (under LSTM and BERT) and CISS (under BERT) in all the settings. To our best knowledge, Text-CRS is the first to provide certified robustness for word reordering, insertion, and deletion. Compared to \emph{Clean vanilla}, Text-CRS sacrifices only a small fraction of accuracy. Across all $24$ settings, including $4\ \textrm{operations} \times\! \ 3\ \textrm{datasets} \times\! \ 2\ \textrm{models}$, Text-CRS provides an average best certified accuracy of $90.2\%$ (bold), with a small drop compared to the average vanilla accuracy of $91.22\%$. 

Moreover, our methods for substitution, insertion, and deletion achieve the best certified accuracy with low noise, indicating that smaller noise has less impact on model performance. Conversely, the best certified accuracy for reordering can be achieved at low, medium, or high noise levels, as the results suggest that robust training with reordering operations has little effect on the model accuracy. Regarding model structures, Text-CRS outperforms SOTA methods under LSTM and achieves superior performance under BERT, with a certified accuracy that is closer to or exceeds that of \emph{Clean vanilla}. Regarding the dataset, Text-CRS demonstrates better performance on large-scale datasets, such as IMDB, where the substitution method outperforms SAFER by $1.6\%$ and $4.7\%$ on AG and IMDB, and CISS by $6.0\%$ and $4.8\%$ on AG and IMDB, respectively.

% We report the certified accuracy of the clean test set of each dataset under different smoothing parameters in Table~\ref{}. 

% 1. compare: safer, ciss
% 2. compare: clean vanilla
% 3. from  params, model structure, dataset

\vspace{-0.1in}
\subsubsection{Certified Accuracy under Different Radii} \label{sec:acc_radii}

We examine the effect of certified radii with different noise levels on the certified accuracy against four adversarial operations, as depicted in Figure~\ref{fig:noise5_certify_result} to \ref{fig:noise4_certify_result}. The results indicate that the certified accuracy declines as the radius increases, and it abruptly drops to zero at a certain radius threshold, consistent with the results in the image domain~\cite{cohen2019certified}. Moreover, the impact of noise level on certified accuracy increases with the magnitude of the noise, while a large noise level can improve the certified radius. In other words, selecting a larger noise is necessary while aiming for a wider certification range. Thus, it is crucial to choose an appropriate smoothing magnitude for each specific setting carefully. 
% We also observe that BERT performs better than LSTM overall, especially on the Amazon and IMDB datasets. This is due to the fact that the Bert model provides a higher \emph{clean vanilla} accuracy and 

\begin{figure}
    \centering
    \begin{subfigure}[b]{0.32\columnwidth}
        \centering
        \includegraphics[width=\textwidth]{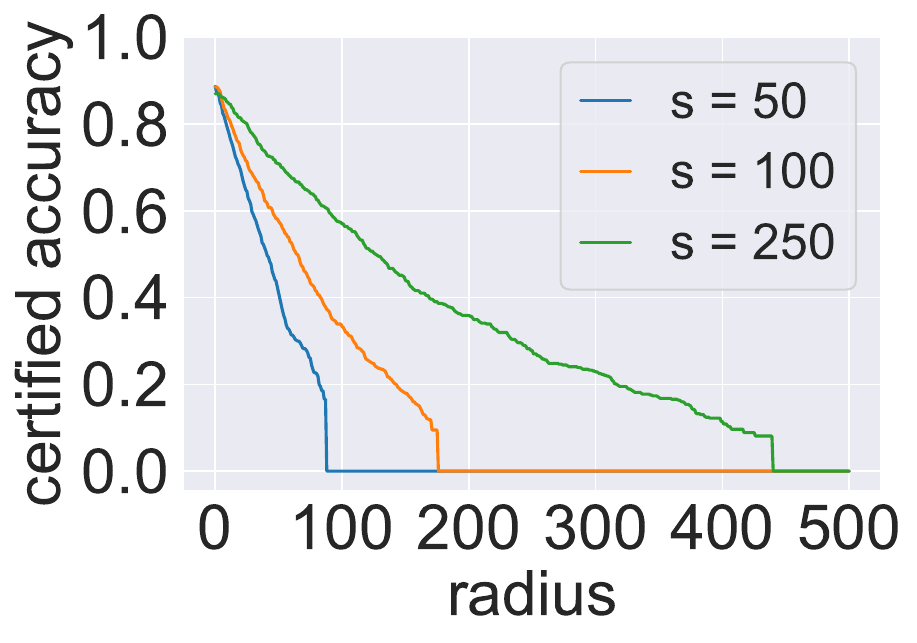}
    \end{subfigure}
    \begin{subfigure}[b]{0.32\columnwidth}
        \centering
        \includegraphics[width=\linewidth]{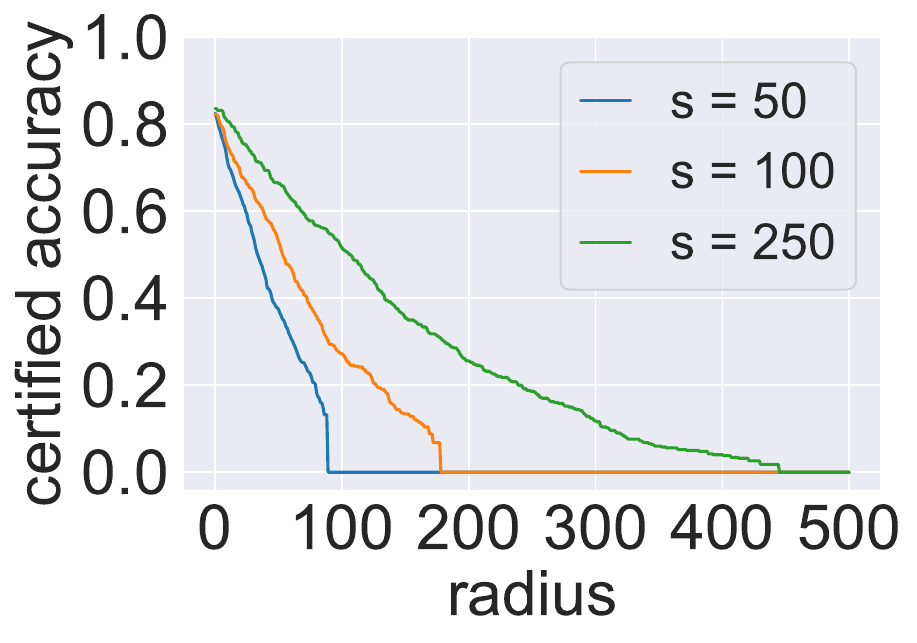}
    \end{subfigure}
    \begin{subfigure}[b]{0.32\columnwidth}
        \centering
        \includegraphics[width=\linewidth]{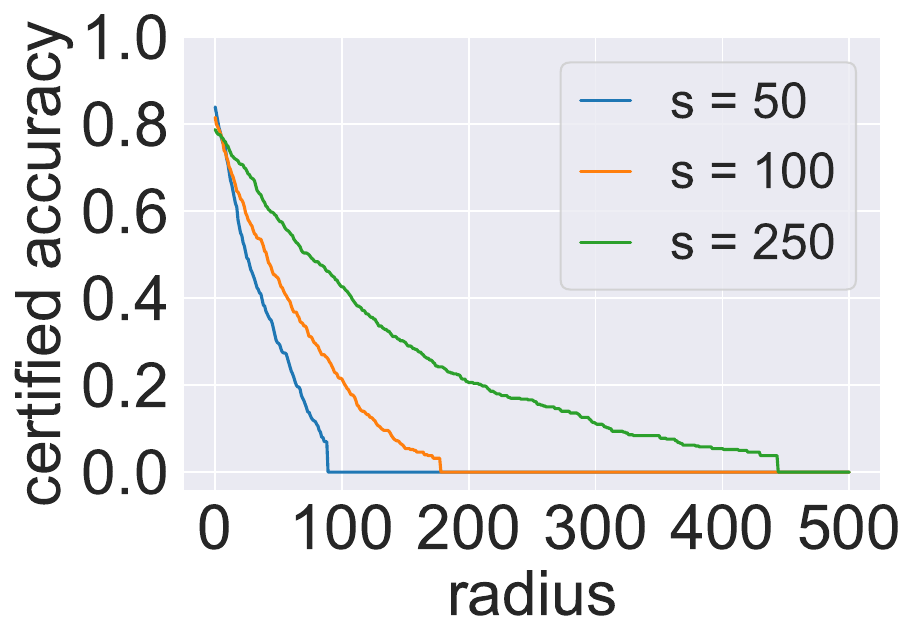}
    \end{subfigure}
    
    \begin{subfigure}[b]{0.32\columnwidth}
        \centering
        \includegraphics[width=\textwidth]{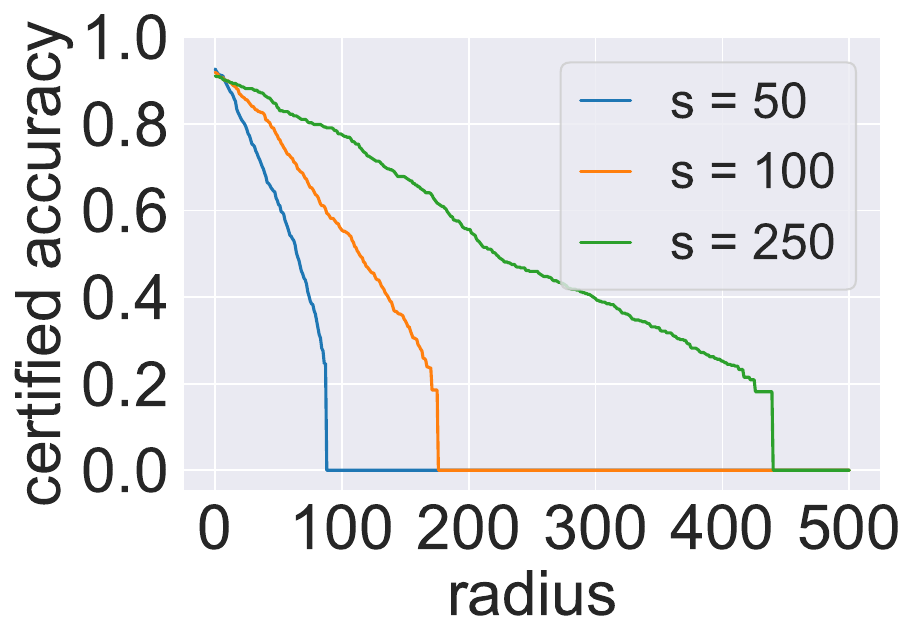}
    \end{subfigure}
    \begin{subfigure}[b]{0.32\columnwidth}
        \centering
        \includegraphics[width=\linewidth]{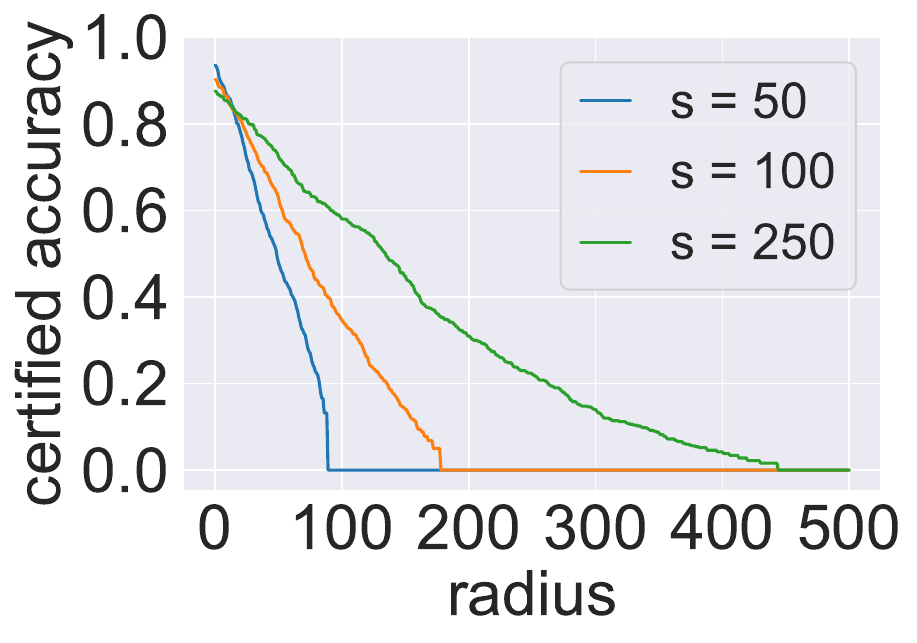}
    \end{subfigure}   
    \begin{subfigure}[b]{0.32\columnwidth}
        \centering
        \includegraphics[width=\linewidth]{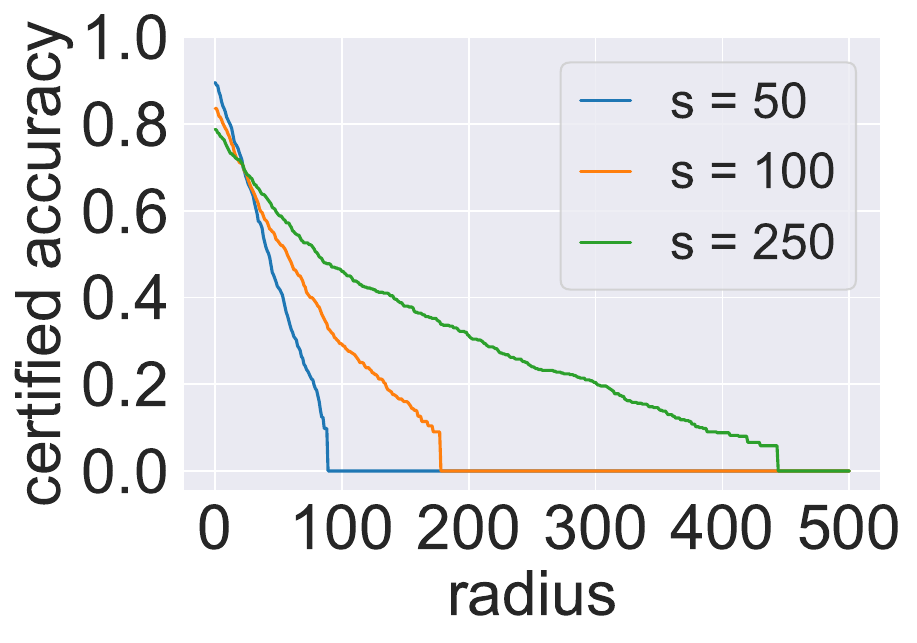}
    \end{subfigure}
\vspace{-1mm}
\caption{Certified accuracy at different radii against synonym substitution. Datasets from left to right: AG, Amazon, and IMDB; Models: top-LSTM and bottom-BERT.} %noise5
\label{fig:noise5_certify_result}
\vspace{-4mm}
\end{figure}

\begin{figure}
    \centering
     \begin{subfigure}[b]{0.32\columnwidth}
        \centering
        \includegraphics[width=\textwidth]{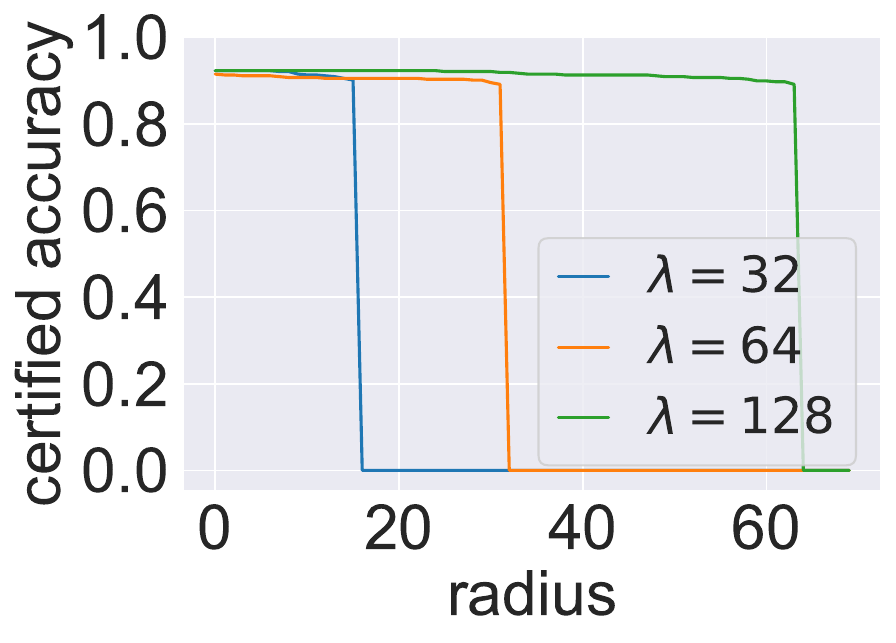}
    \end{subfigure}
    \begin{subfigure}[b]{0.32\columnwidth}
        \centering
        \includegraphics[width=\linewidth]{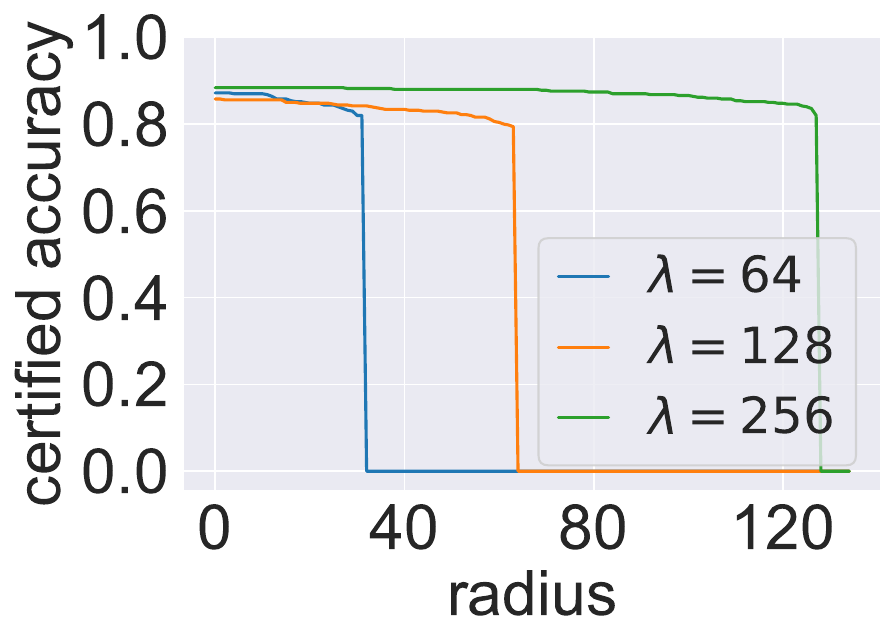}
    \end{subfigure}   
    \begin{subfigure}[b]{0.32\columnwidth}
        \centering
        \includegraphics[width=\linewidth]{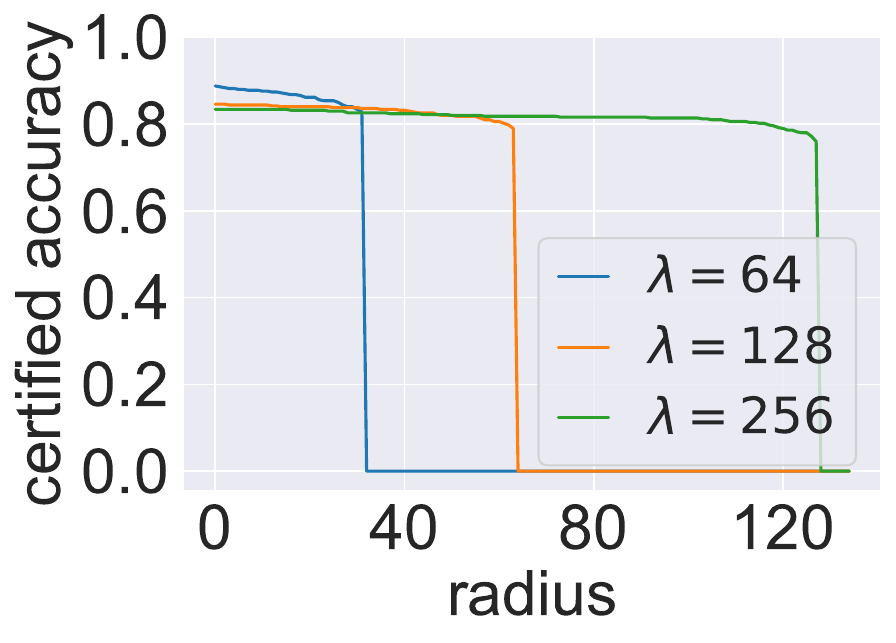}
    \end{subfigure}
    
    \begin{subfigure}[b]{0.32\columnwidth}
        \centering
        \includegraphics[width=\textwidth]{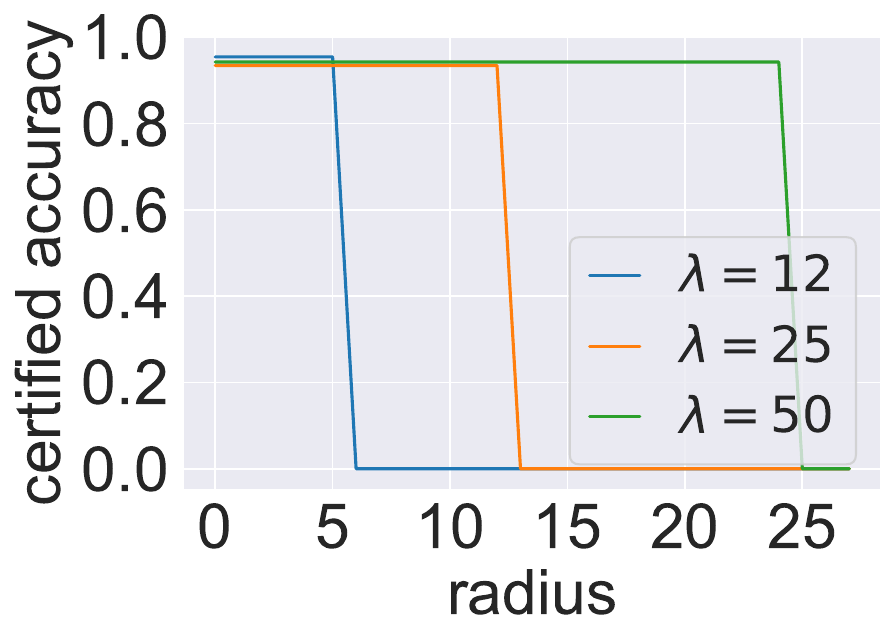}
    \end{subfigure}
    \begin{subfigure}[b]{0.32\columnwidth}
        \centering
        \includegraphics[width=\linewidth]{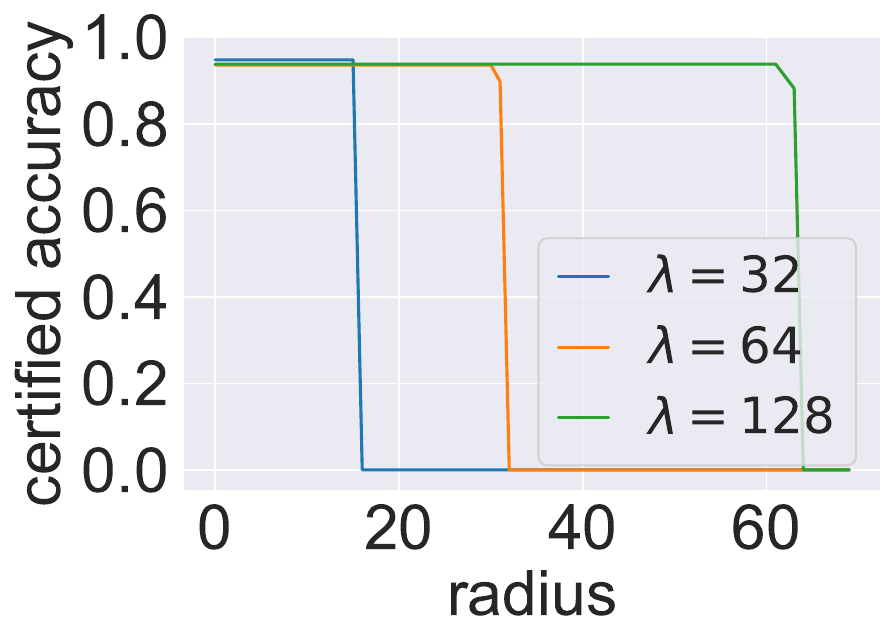}
    \end{subfigure} 
    \begin{subfigure}[b]{0.32\columnwidth}
        \centering
        \includegraphics[width=\linewidth]{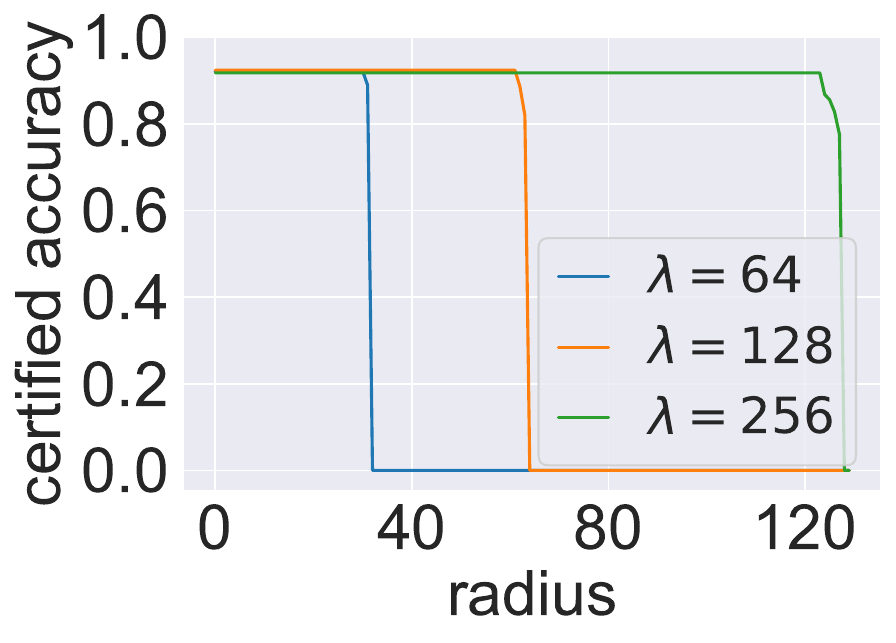}
    \end{subfigure}
\vspace{-1mm}
    \caption{Certified accuracy at different radii against word reordering. Datasets from left to right: AG, Amazon, and IMDB; Models: top-LSTM and bottom-BERT. } %noise8
    \label{fig:noise8_certify_result}
\vspace{-5mm}
\end{figure}

\begin{figure}
    \centering
    \vspace{-3mm}
    \begin{subfigure}[b]{0.32\columnwidth}
        \centering
        \includegraphics[width=\textwidth]{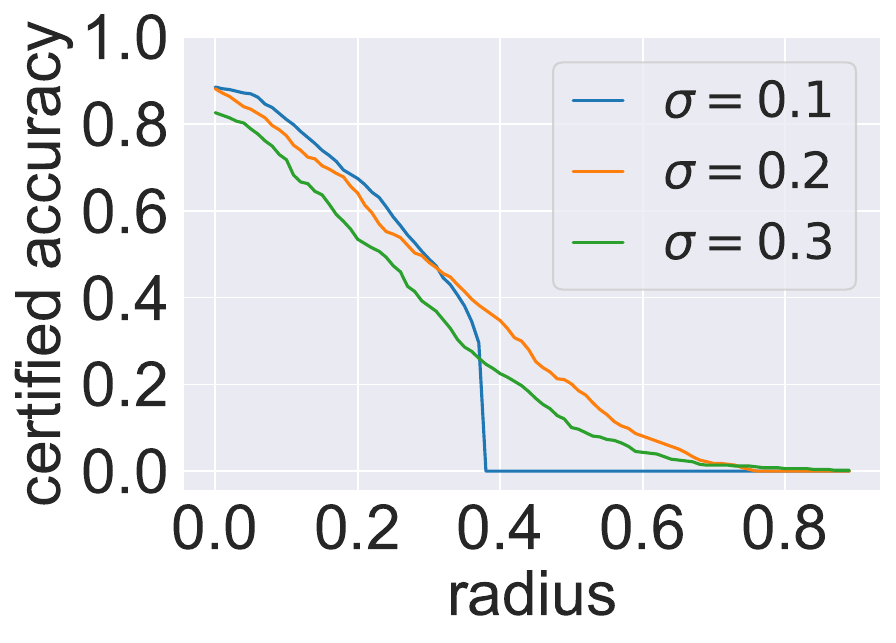}
    \end{subfigure}
    \begin{subfigure}[b]{0.32\columnwidth}
        \centering
        \includegraphics[width=\linewidth]{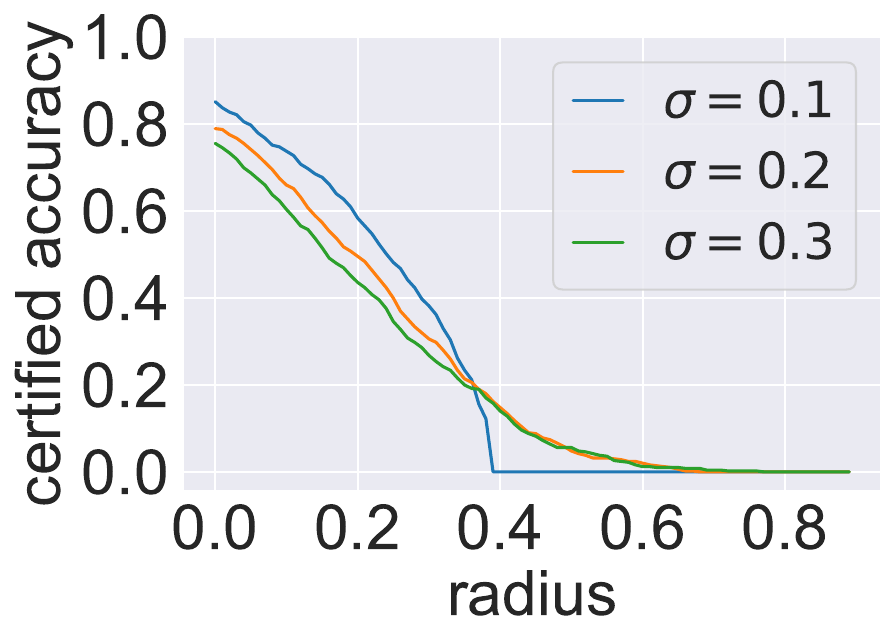}
    \end{subfigure}
    \begin{subfigure}[b]{0.32\columnwidth}
        \centering
        \includegraphics[width=\linewidth]{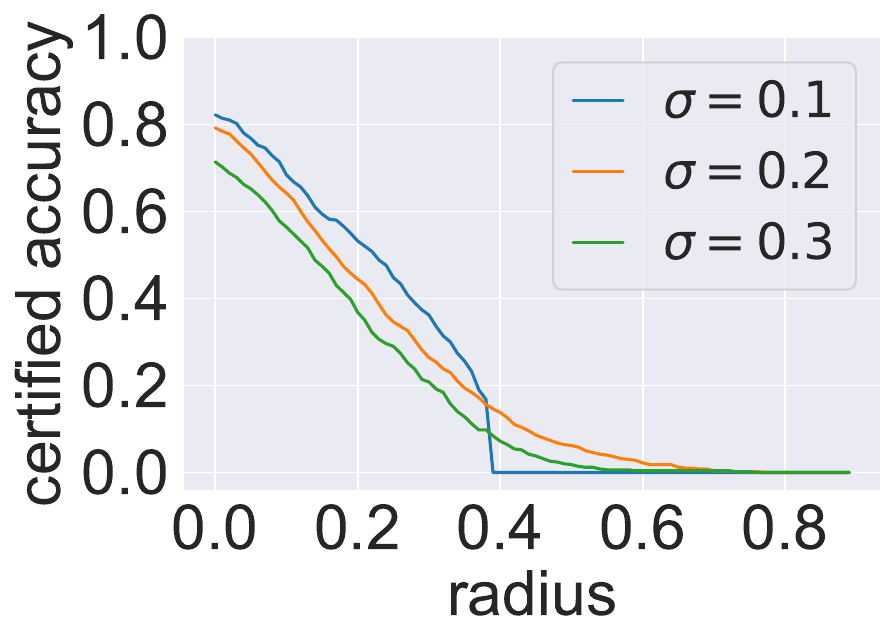}
    \end{subfigure}
    
    \begin{subfigure}[b]{0.32\columnwidth}
        \centering
        \includegraphics[width=\textwidth]{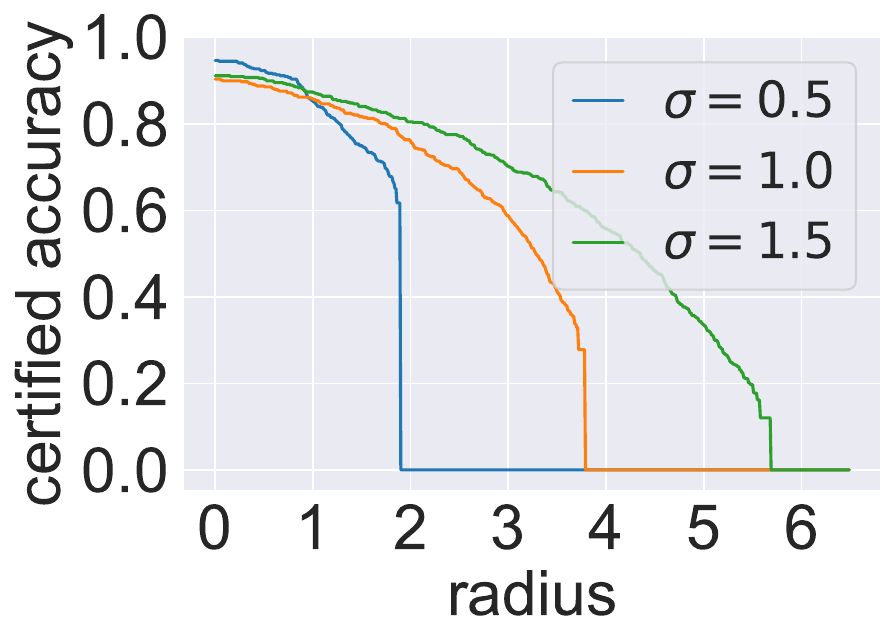}
    \end{subfigure}
    \begin{subfigure}[b]{0.32\columnwidth}
        \centering
        \includegraphics[width=\linewidth]{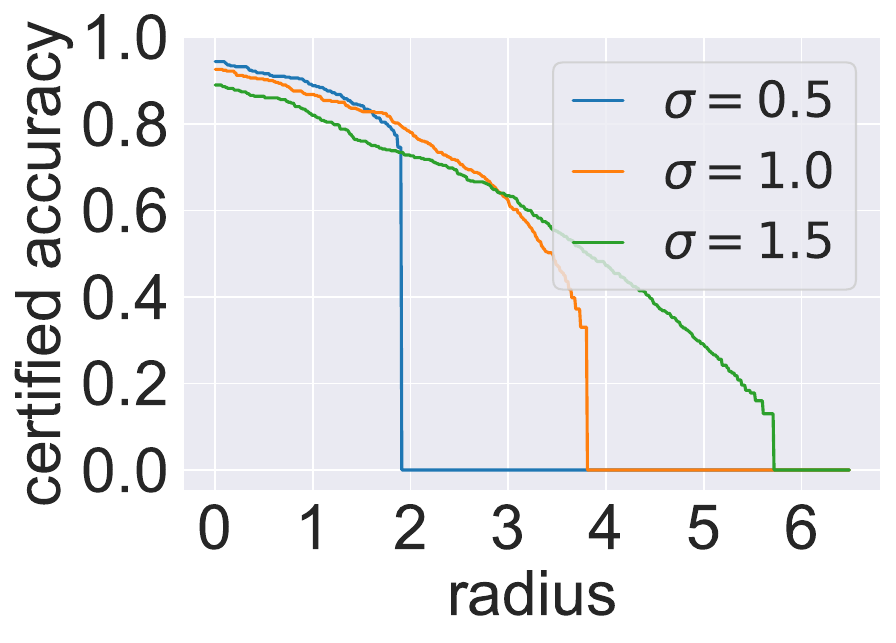}
    \end{subfigure}  
    \begin{subfigure}[b]{0.32\columnwidth}
        \centering
        \includegraphics[width=\linewidth]{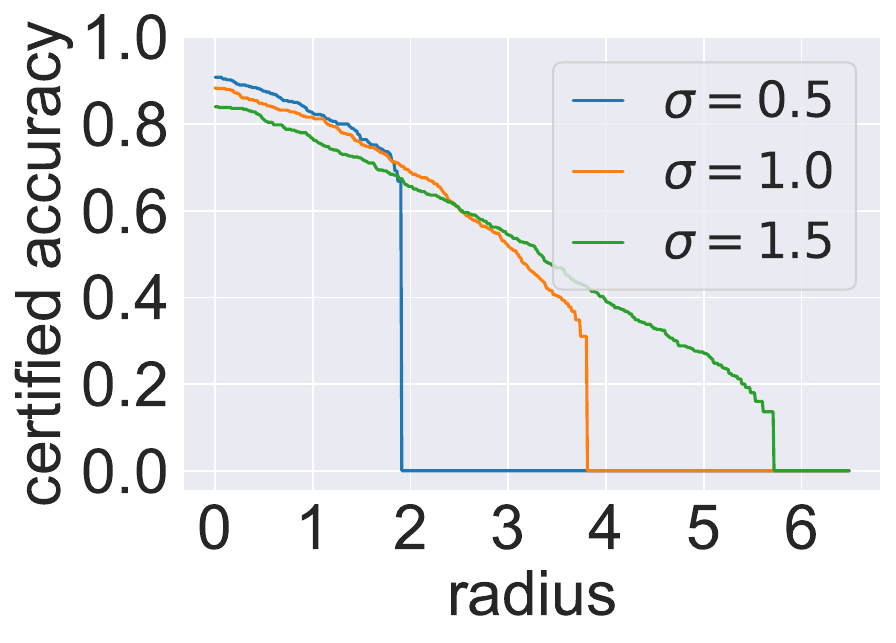}
    \end{subfigure}
\vspace{-1mm}    
\caption{Certified accuracy at different radii against word insertion. Datasets from left to right: AG, Amazon, and IMDB; Models: top-LSTM and bottom-BERT.
%\BW{Word deletion results are still running?}\xy{I have updated the word deletion radius.}
} %noise3
\label{fig:noise3_certify_result}
\vspace{-2mm}
\end{figure}

\begin{table*}[]
\setlength\tabcolsep{3pt}
\scriptsize %footnotesize 
\caption{Comparison of certified accuracy of Text-CRS, SAFER~\cite{ye2020safer} and CISS~\cite{zhao2022certified} under different attacks. “$*$” indicates that SAFER and CISS cannot certify operations other than substitution, resulting in a certified accuracy of $0\%$ for these attacks, so we report their empirical accuracy. “-” indicates that BAE-Insert and CISS cannot be performed on LSTM. } %  does not provide an LSTM training pipeline
\vspace{-0.05in}
\resizebox{\linewidth}{!}{
\begin{tabular}{lcccccccccccccccc}
\toprule
 & \multicolumn{1}{l}{} & \multicolumn{3}{c}{TextFooler\cite{jin2020bert}} & \multicolumn{3}{c}{WordReorder\cite{moradi2021evaluating}} & \multicolumn{3}{c}{SynonymInsert\cite{morris2020textattack}} & \multicolumn{3}{c}{BAE-Insert\cite{garg2020bae}} & \multicolumn{3}{c}{InputReduction\cite{feng2018pathologies}} \\
 \cmidrule(r){3-5} \cmidrule(r){6-8} \cmidrule(r){9-11} \cmidrule(r){12-14} \cmidrule(r){15-17} 
Dataset (Model) & \multicolumn{1}{l}{\emph{Vanilla}} & \multicolumn{1}{c}{SAFER} & \multicolumn{1}{c}{CISS} & \multicolumn{1}{c}{Ours} & \multicolumn{1}{c}{SAFER$^*$} & \multicolumn{1}{c}{CISS$^*$} & \multicolumn{1}{c}{Ours} & \multicolumn{1}{c}{SAFER$^*$} & \multicolumn{1}{c}{CISS$^*$} & \multicolumn{1}{c}{Ours} & \multicolumn{1}{c}{SAFER$^*$} & \multicolumn{1}{c}{CISS$^*$} & \multicolumn{1}{c}{Ours} & \multicolumn{1}{c}{SAFER$^*$} & \multicolumn{1}{c}{CISS$^*$} & \multicolumn{1}{c}{Ours} \\
\midrule
AG (LSTM) & 0\% & 90.4\% & -  & \textbf{91.2\%} & 75.4\% & - & \textbf{89.2\%} & 77.6\% & - & \textbf{84.2\%} & - & - & - & 65.8\% & - & \textbf{78.4\%} \\
AG (BERT) & 0\% & 93.2\% & 71.4\% & \textbf{93.6\%} & 86.8\% & 83.8\% & \textbf{87.6\%} & 77.8\% & 74.6\% & \textbf{83.6\%} & \textbf{79.8\%} & 74.0\% & 79.4\% & 55.8\% & 59.0\% & \textbf{68.4\%} \\
Amazon (LSTM) & 0\% & 82.4\% & - & \textbf{83.4\%} & 78.0\% & - & \textbf{91.2\%} & 64.0\% & - & \textbf{71.8\%} & - & - & - & 71.2\% & - & \textbf{74.2\%} \\
Amazon (BERT) & 0\% & 87.0\% & 75.4\% & \textbf{90.8\%} & 82.0\% & \textbf{88.0\%} & 84.6\% & 67.8\% & 67.8\% & \textbf{80.6\%} & 70.4\% & 70.4\% & \textbf{71.2\%} & 68.6\% & 74.0\% & \textbf{82.0\%} \\
IMDB   (LSTM) & 0\% & 82.0\% & - & \textbf{84.6\%} & 77.8\% & - & \textbf{86.0\%} & 66.8\% & - & \textbf{69.6\%} & - & - & - & 68.8\% & - & \textbf{77.0\%} \\
IMDB   (BERT) & 0\% & 83.8\% & 25.6\% & \textbf{84.4\%} & 80.2\% & 24.8\% & \textbf{83.0\%} & 76.8\% & 31.6\% & \textbf{86.2\%} & 77.0\% & 33.0\% & \textbf{80.6\%} & 70.6\% & 29.8\% & \textbf{82.8\%} \\
\midrule
Average & 0\% & 86.5\% & 57.5\% & \textbf{88.0\%} & 80.0\% & 65.5\% & \textbf{86.9\%} & 71.8\% & 58.0\% & \textbf{79.3\%} & 75.7\% & 59.1\% & \textbf{77.1\%} & 66.8\% & 54.3\% & \textbf{77.1\%} \\
\bottomrule
\end{tabular}
}
\vspace{-4mm}
\label{tab:certi_acc_attacks}
\end{table*}

\begin{figure}
    \centering
    \begin{subfigure}[b]{0.32\columnwidth}
        \centering
        \includegraphics[width=\textwidth]{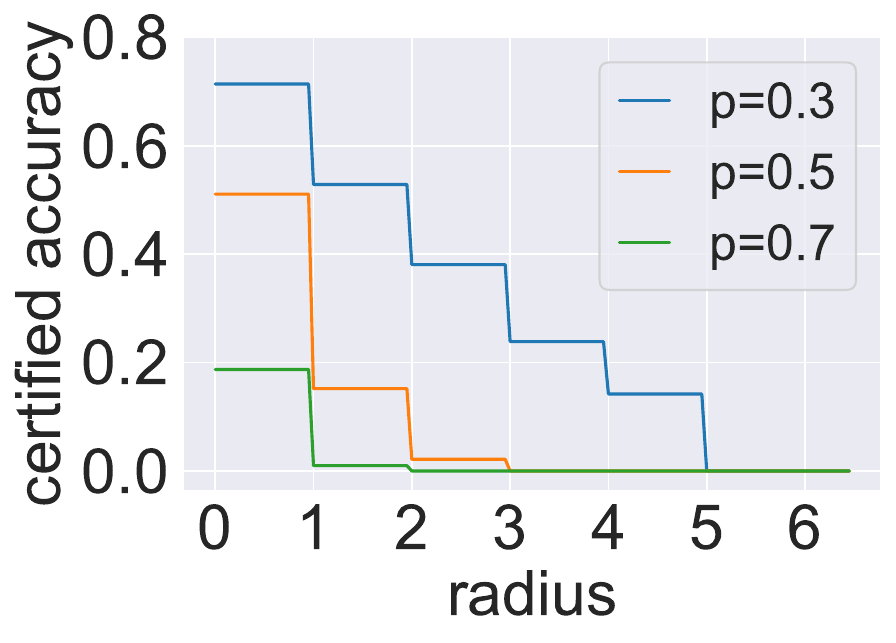}
    \end{subfigure}
    \begin{subfigure}[b]{0.32\columnwidth}
        \centering
        \includegraphics[width=\linewidth]{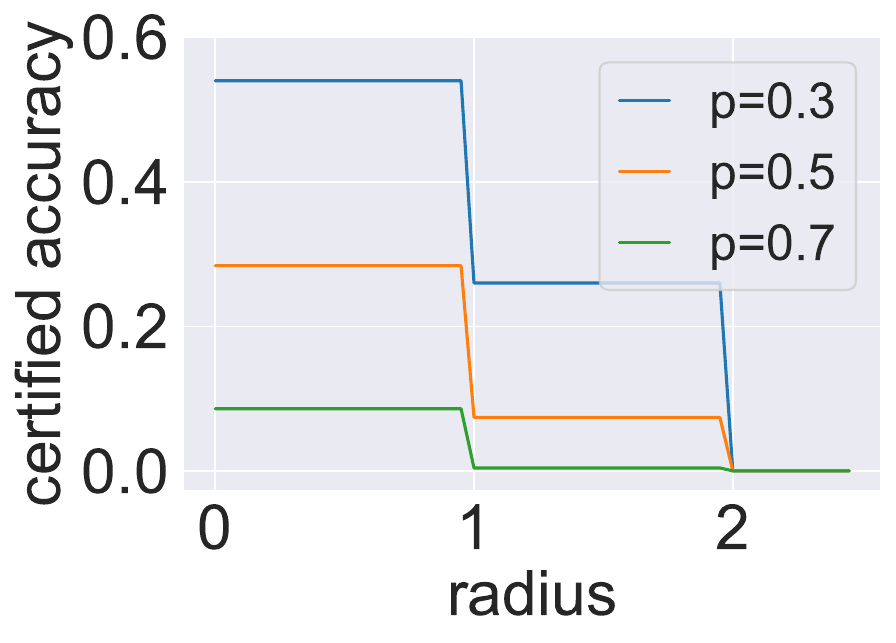}
    \end{subfigure}
    \begin{subfigure}[b]{0.32\columnwidth}
        \centering
        \includegraphics[width=\linewidth]{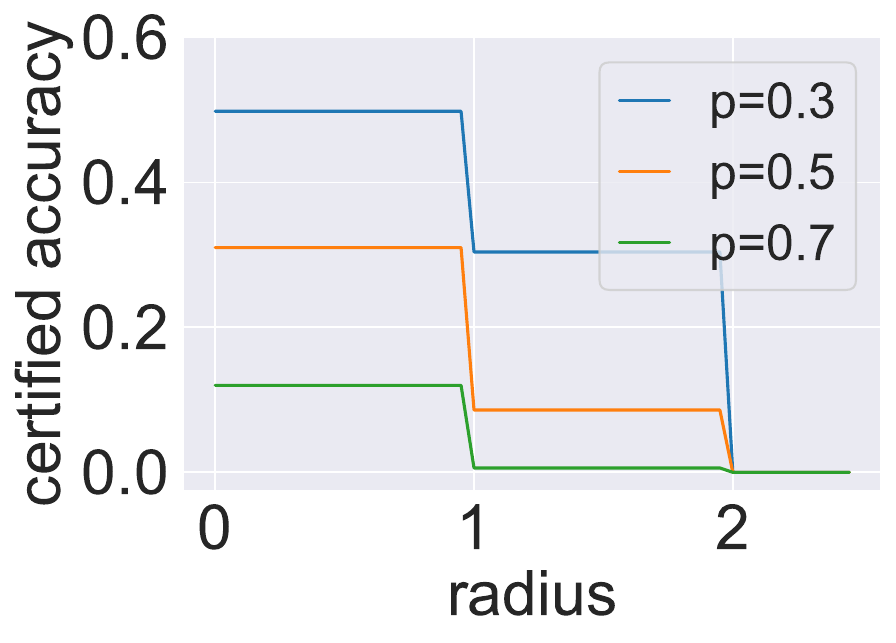}
    \end{subfigure}
    
    \begin{subfigure}[b]{0.32\columnwidth}
        \centering
        \includegraphics[width=\textwidth]{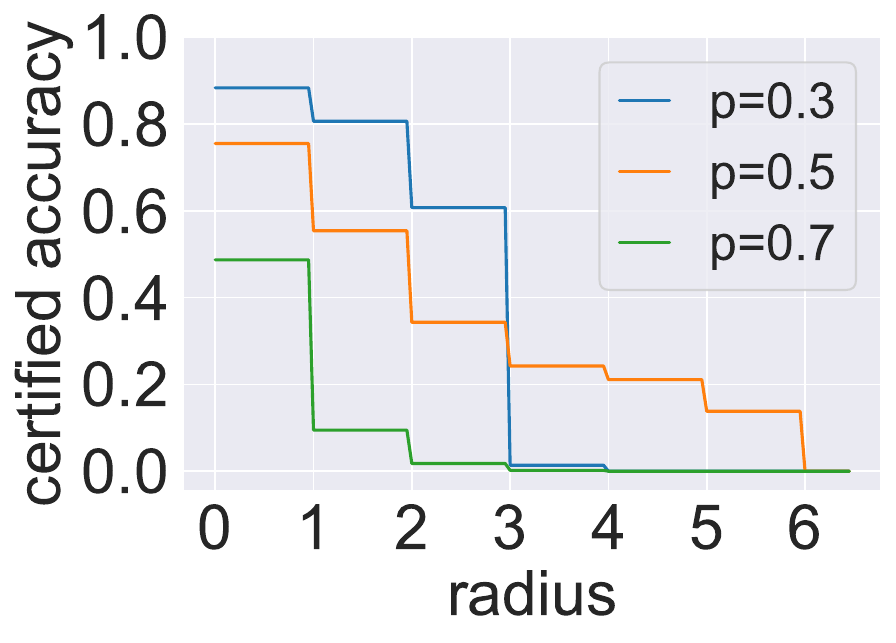}
    \end{subfigure}
    \begin{subfigure}[b]{0.32\columnwidth}
        \centering
        \includegraphics[width=\linewidth]{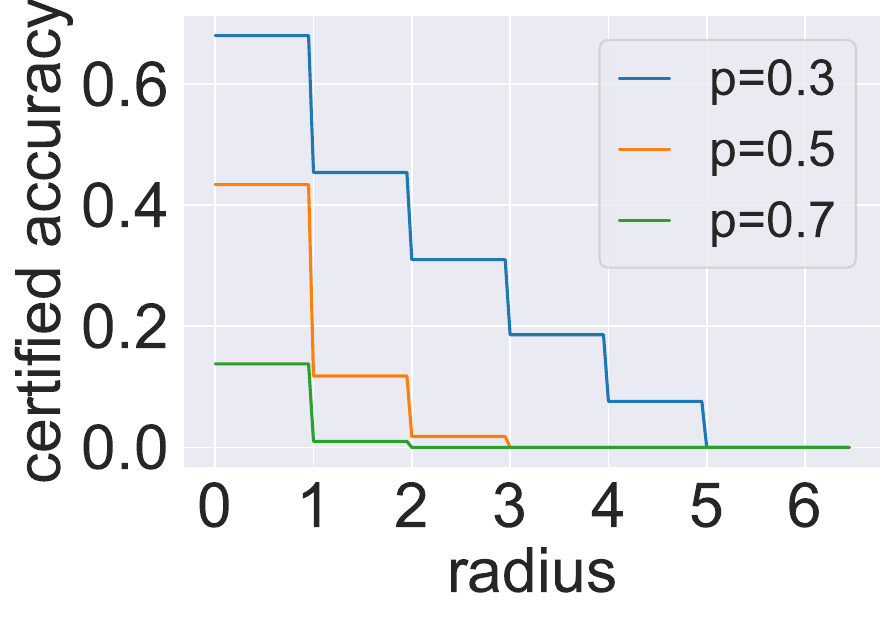}
    \end{subfigure}  
    \begin{subfigure}[b]{0.32\columnwidth}
        \centering
        \includegraphics[width=\linewidth]{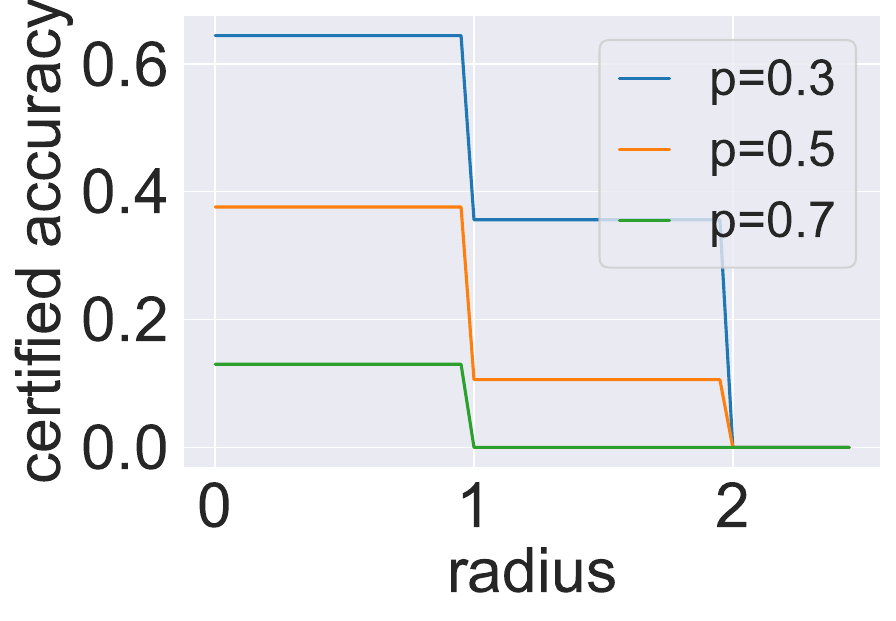}
    \end{subfigure}

\vspace{-1mm}
\caption{Certified accuracy at different radii against word deletion. Datasets from left to right: AG, Amazon, and IMDB; Models: top-LSTM and bottom-BERT.} %noise4
\label{fig:noise4_certify_result}
\vspace{-4mm}
\end{figure}

Figure~\ref{fig:noise5_certify_result} depicts the certified accuracy with different sizes of synonym sets. A radius of $\rad_S \!=\! 200$ for a sentence with a length of 50 implies that each word can be substituted with its four closest synonyms in the thesaurus. In such cases, the prediction results of the smoothed classifier remain the same as the original sentence. Figure~\ref{fig:noise8_certify_result} depicts the certified accuracy under different sizes of shuffling groups. The radius $\rad_R \!=\! 100$ indicates that Text-CRS certifies a text in which the sum of all word positions changes is less than $100$. Figure~\ref{fig:noise3_certify_result} presents the certified accuracy under different Gaussian noise. The radius $\rad_I$ denotes the cumulative embedding $\ell_2$ distances between the original and the inserted word. To illustrate the practical significance of our radius, we calculate the $\ell_2$ distance between $65,713$ words and their closest top-$k$ embeddings under Glove embedding space and BERT embedding space (see Figure~\ref{fig:top-k_closest_embedding}). The results indicate that under $\rad_I \!=\! 0.2$, the LSTM model (using GloVe embedding space) could withstand $\sim$7\% of random word insertions among all the top-3 closest words. The BERT model performs better than LSTM, which can withstand $\sim$53\% and $\sim$11\% of random word insertions among the top-5 and top-50 closest embeddings, respectively, under $\rad_I \!=\! 2$. 
% The results demonstrate that Text-CRS can provide certified robustness against multiple random word insertions. 
Figure~\ref{fig:noise4_certify_result} shows the certified accuracy under different word deletion probabilities, where a radius $\rad_D \!=\! 2$ indicates that up to two words can be deleted while ensuring the certified robustness of the text. Note that Figure~\ref{fig:noise8_certify_result} also shows the certified radius on word position changes in word insertion and deletion operations.

\begin{figure}[!t]
% \vspace{-0.1in}
    \centering
    \begin{subfigure}[b]{0.49\columnwidth}
        \centering
        \includegraphics[width=\linewidth]{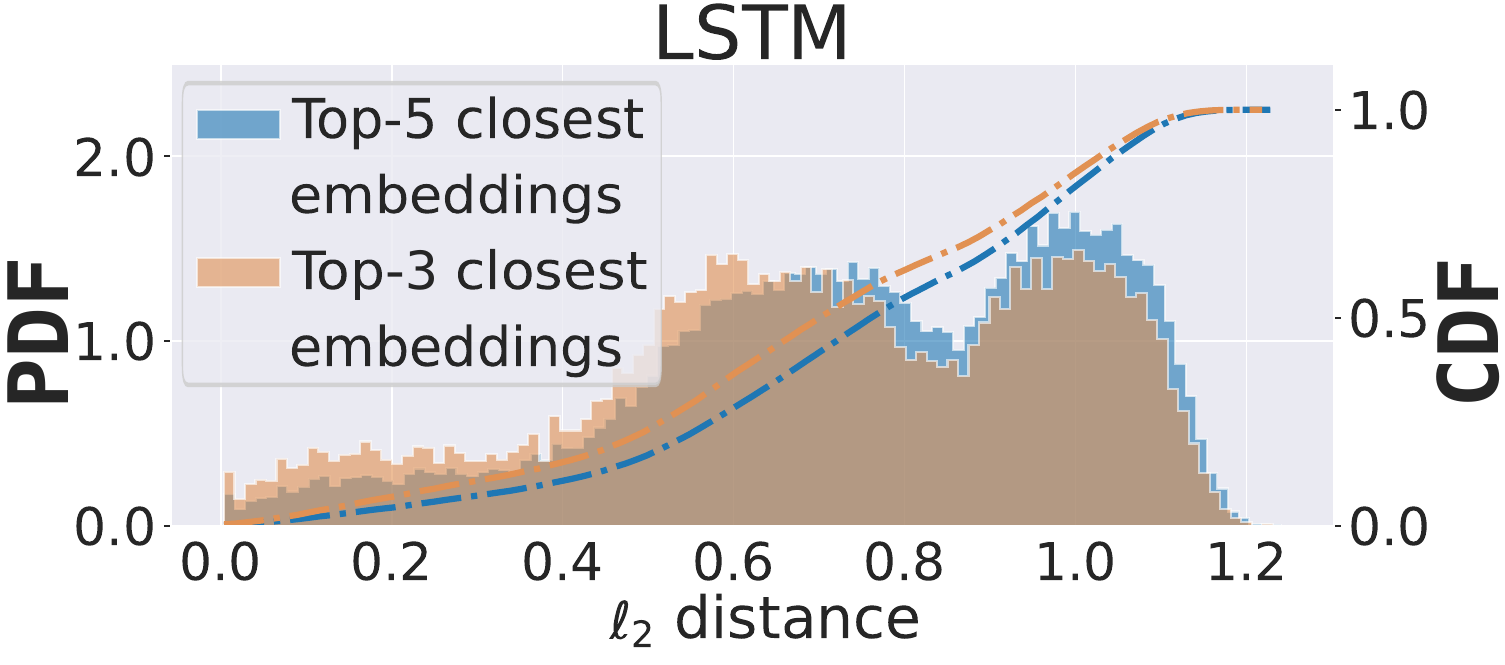}
    \end{subfigure}
    \begin{subfigure}[b]{0.49\columnwidth}
        \centering
         \includegraphics[width=\linewidth]{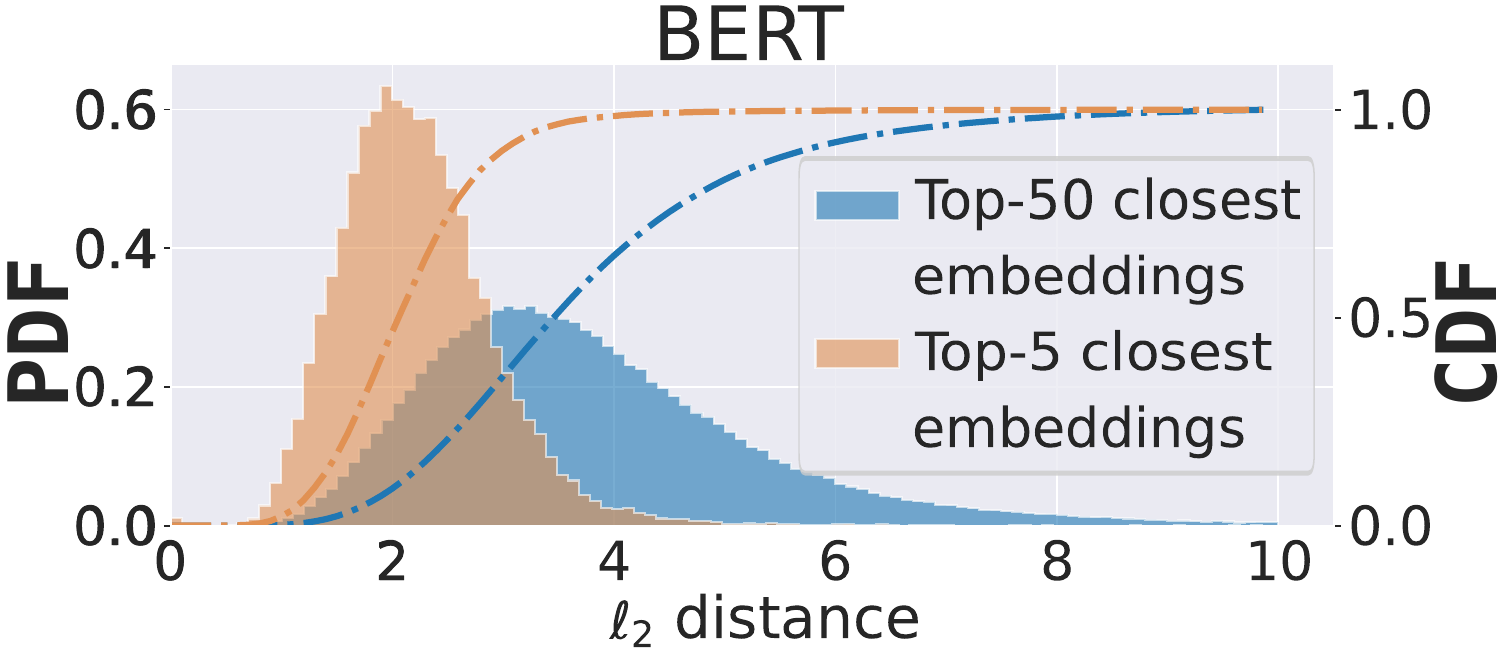} 
    \end{subfigure}
\vspace{-4mm}
    \caption{PDF and CDF of top-$k$ closest embeddings in Glove (LSTM) and BERT embedding space.}
    \label{fig:top-k_closest_embedding}
\vspace{-4mm}
\end{figure}

Furthermore, regarding the word reordering, the noise level has little effect on the certified accuracy, particularly for the BERT, which uses a transformer structure to represent each word as a weighted sum of all input word embeddings~\cite{vaswani2017attention}. Consequently, the prediction of BERT contains information about every word, and the reordering of words has a reduced effect on the resultant output of BERT. This implies that we can choose a large noise level to achieve high certified accuracy while obtaining the largest certified radius simultaneously. It is consistent with our reordering strategies in word insertion and deletion operations.

\begin{table}[]
\setlength\tabcolsep{1pt}
\centering
\scriptsize
\caption{Certified accuracy of word insertion smoothing method against different attacks. } %“-” denotes the BAE-Insert attack cannot be performed on LSTM.
\vspace{-0.05in}
\resizebox{\linewidth}{!}{
\begin{tabular}{lccccc}
\toprule
\begin{tabular}[c]{@{}c@{}}Dataset (Model)\end{tabular} & \multicolumn{1}{c} {\begin{tabular}[c]{@{}c@{}}Text-\\ Fooler\cite{jin2020bert} \end{tabular}} & \multicolumn{1}{c}{\begin{tabular}[c]{@{}c@{}}Word-\\ Reorder \cite{moradi2021evaluating}\end{tabular}} & \multicolumn{1}{c}{\begin{tabular}[c]{@{}c@{}} Synonym-\\ Insert\cite{morris2020textattack}\end{tabular}} & \multicolumn{1}{c}{\begin{tabular}[c]{@{}c@{}}BAE-\\ Insert \cite{garg2020bae}\end{tabular}} & \multicolumn{1}{c}{\begin{tabular}[c]{@{}c@{}}Input-\\ Reduction\cite{feng2018pathologies}\end{tabular}} \\
\midrule
AG (LSTM) & 81.6\% & 85.0\% & 84.2\% & - & 70.0\% \\
AG (BERT) & 83.4\% & 90.2\% & 83.6\% & 79.4\% & 58.4\% \\
Amazon (LSTM) & 75.4\% & 83.6\% & 71.8\% & - & 70.4\% \\
Amazon (BERT) & 82.8\% & 84.4\% & 80.6\% & 71.2\% & 74.4\% \\
IMDB (LSTM) & 64.2\% & 87.2\% & 69.6\% & - & 68.8\% \\
IMDB (BERT) & 81.4\% & 87.0\% & 86.2\% & 80.6\% & 77.4\% \\
\midrule
Average & 78.1\% & 86.2\% & 79.3\% & 77.1\% & 69.9\% \\
\bottomrule
\end{tabular}
}
\vspace{-4mm}
\label{tab:universe_certi_acc}
\end{table}

\vspace{-0.1in}
\subsubsection{Certified Accuracy against Unseen Attacks}

We select the noise parameters w.r.t. the highest certified accuracy for the three methods and evaluate the robustness under these noises. Table~\ref{tab:certi_acc_attacks} displays the certified accuracy of three methods against five types of attacks. The `\emph{Vanilla}' column shows that the accuracy is $0\%$ for all successful adversarial examples on the vanilla models. Text-CRS demonstrates an average certified accuracy that is 64\% and 70\% higher than SAFER and CISS, respectively, across all attacks.
Specifically, our method achieves the highest certified accuracy for the TextFooler attack, surpassing SAFER and CISS in all settings. SAFER and CISS only provide certification against substitution attacks, resulting in $0\%$ certified accuracy for WordReorder to InputReduction attacks. Therefore, we evaluate their empirical accuracy against these attacks. It is important to note that empirical accuracy is generally higher than certified accuracy under the same setting. 
On average, our certified accuracy is still $5.5\%$ and $22.8\%$ higher than the empirical accuracy of SAFER and CISS, respectively.

\vspace{-0.1in}
\subsubsection{Word Insertion Smoothing vs. All Attacks (Universality)}
We evaluate the certified accuracy of the word insertion smoothing against all aforementioned attacks (universality), as shown in Table~\ref{tab:universe_certi_acc}. On average, our word insertion smoothing achieves a certified accuracy of $78.1\%$. It can certify against all attacks, though with a slight performance drop compare to the operation-specific methods in our Text-CRS. However, its certified accuracy is still higher than the empirical accuracy of SAFER and CISS, except the SAFER against TextFooler (substitution operation-specific). %However, the word insertion method shows a slight degradation of 5.9\% on average compared to our operation-specific methods. 
Therefore, the word insertion smoothing in Text-CRS is suitable for providing high certified accuracy when the attack type is unknown (due to its high universality), while higher certified accuracy can be achieved using specific methods in Text-CRS to defend known attacks. 
% 
% Comparing the average accuracy of SAFER and CISS in Table~\ref{tab:certi_acc_attacks}, our method not only achieves an accuracy higher than the empirical accuracy of the comparison methods but also provides certified robustness. On the other hand, the word insertion method is less accurate than the method we designed for specific operations, but the accuracy is reduced to within 5\%. Thus when faced with an attack on an unknown operation, our word insertion smoothing method can provide an acceptable certified accuracy of 79.7\%.

% \begin{figure}
%     \centering 
%     \begin{subfigure}[b]{0.9\columnwidth}
%         \centering
%         \includegraphics[width=\textwidth]{figures/compare_N/compare_data_model.png}
%     \end{subfigure}
    
%     \begin{subfigure}[b]{0.9\columnwidth}
%         \centering
%         \includegraphics[width=\linewidth]{figures/compare_N/compare_noise.png}
%     \end{subfigure}   
    
%     \caption{Impact of the number of samples ($N$) on certified accuracy against adversarial attacks under LSTM and BERT. \textbf{Top}: accuracy against the SynonymInsert attack under three datasets. \textbf{Bottom}: accuracy against four attacks on Amazon.}
%     \label{fig:impact_N}
% \end{figure}

\begin{figure}
    \centering 
    \begin{subfigure}[b]{0.32\columnwidth}
        \centering
        \includegraphics[width=\textwidth]{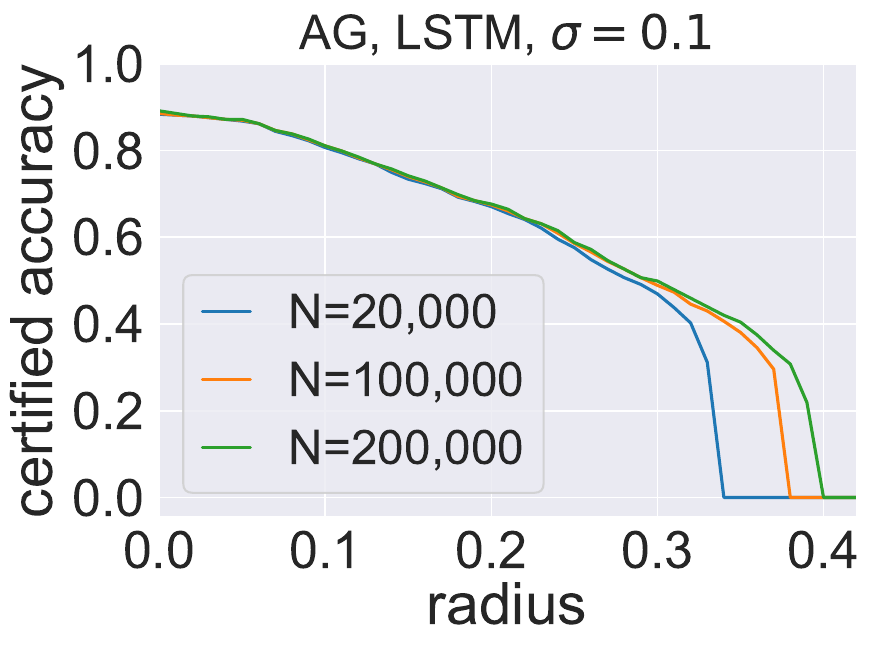}
    \end{subfigure}
    \begin{subfigure}[b]{0.32\columnwidth}
        \centering
        \includegraphics[width=\linewidth]{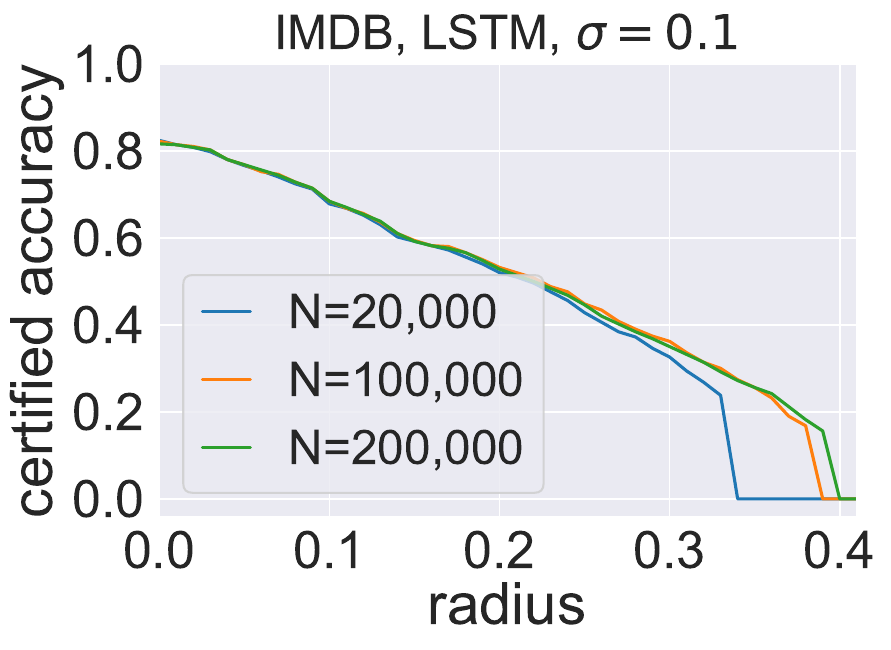}
    \end{subfigure}
    \begin{subfigure}[b]{0.32\columnwidth}
        \centering
        \includegraphics[width=\linewidth]{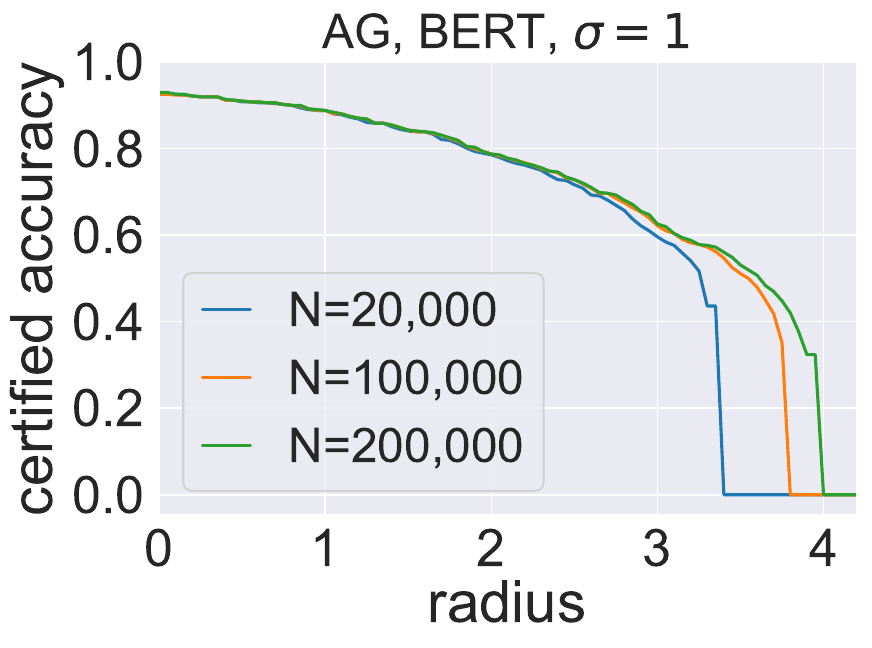}
    \end{subfigure}
\vspace{-1mm} 
    \caption{Impact of the number of samples ($N$) on certified accuracy against word insertion under LSTM and BERT.}\vspace{-0.15in}
\vspace{-2mm}
    \label{fig:impact_N_1}
\end{figure}

% We investigate the impact of the number of samples ($N$) on the certified accuracy against adversarial attacks when the noise level is high. 

% Figure~\ref{fig:impact_N} (top and bottom) illustrates the changes in certified accuracy of Text-CRS against the SynonymInsert attack for different datasets and models, as well as Text-CRS against different attacks. The results show that as $N$ increases, certified accuracy also increases, as the estimate for $\underline{p_A}$ approaches the actual lower bound. The impact of $N$ on the certified accuracy of the LSTM is greater than that of the BERT, as the high accuracy of BERT leads to a more precise estimation of $\underline{p_A}$ even with smaller values of $N$. Additionally, we observe that $N$ has little impact on the certified accuracy of the WordReorder attack due to the limited noise space of word reordering compared to the other three attacks. As a result, $N=10000$ is sufficient to evaluate accurate $\underline{p_A}$.

\vspace{-0.1in}
\subsubsection{Impact of the Number of Samples ($N$)} 
Figure~\ref{fig:impact_N_1} show that as the number of samples for estimation ($N$) increases, the certified accuracy also increases and the certified radius becomes larger, since the estimation for $\underline{p_A}$ and $\overline{p_B}$ becomes tighter. The impact of $N$ on the certified accuracy of the IMDB dataset is greater than that of the AG dataset, since longer inputs have a larger noise space, necessitating more samples to approximate $\underline{p_A}$ and $\overline{p_B}$.

% the high accuracy of BERT leads to a more precise estimation of $\underline{p_A}$ even with smaller values of $N$. 

% \subsubsection{Runtime Analysis}

% We answer the ChatGPT using the following questions:
% \begin{center}
% \fcolorbox{darkblue}{lightblue!30}{\parbox{.95\linewidth}{\textbf{\textcolor{darkblue}{Questions:}}. 
%     \item (1) Please answer if the following two sentences have similar semantical information despite the grammatical errors. Please first choose Yes or No. 
%     \item (2) Then summarize the differences between the meaning of the two sentences by points. 
%     \item (3) Finally, please calculate the cosine similarity between these two sentences.}}
% \end{center}

% \begin{table*}[]
% \caption{Query-answer examples from ChatGPT}
%     \centering
%     \begin{tabular}{c|ccc}
%     \toprule
%     Sentences (Red indicates substitution or insertion of words) & Semantically similar & Differences summarized by ChatGPT & Cosine Similarity \\
%     \midrule
%     Sentences & Yes & Differences summarized by ChatGPT & 0.91 \\
%     Sentences  & No & Differences summarized by ChatGPT & 0.891 \\
%     \bottomrule
%     \end{tabular}
%     \label{tab:chatgpt_q_a}
% \end{table*}

\vspace{-0.09in}
\section{Related Work} \label{sec:related}
\vspace{-0.1in}

\noindent\textbf{Word-level Adversarial Attacks.} 
These attacks aim to mislead the model by modifying the words in four adversarial operations: synonym substitution~\cite{alzantot2018generating,ren2019generating,jin2020bert, zang2020word,tan2020s,dongtowards}, word reordering~\cite{moradi2021evaluating,nie2019analyzing,yan2021consert,lee2020slm}, word insertion~\cite{morris2020textattack,li2021contextualized,garg2020bae, behjati2019universal}, and word deletion~\cite{feng2018pathologies,moradi2021evaluating,xie2022word}. 
For instance, Ren et al.~\cite{ren2019generating} propose a greedy PWWS algorithm to determine the replacement order of words in a sentence and the selection of synonyms. 
% Jin et al.~\cite{jin2020bert} proposed TextFooler, which picks the word crucial to the prediction, i.e., when this word is removed, the prediction undergoes a significant change; then, selects synonyms with high cosine similarity to the original embedding vector for substitution. 
Tan et al.~\cite{tan2020s} proposed Morpheus, which generates plausible and semantically similar adversarial texts by replacing the words with their inflected form. Moradi and Samwald~\cite{moradi2021evaluating} investigated the sensitivity of NLP systems to such operation. % and showed that such operation led to an average 18.3\% decrease in the accuracy of LSTM-based models.
Morris et al.~\cite{morris2020textattack} proposed TextAttack, including a basic strategy of inserting synonyms of words already in the text.
% to maintain semantic similarity between the adversarial and clean text. 
Li et al.~\cite{li2021contextualized} proposed CLARE, which, similar to BAE, also employs masked language models to predict newly inserted \texttt{<mask>} tokens in the text and replace them.
% . The predicted words are then used to replace the \texttt{<mask>} tokens to generate the adversarial text. 
These adversarial texts have improved syntactic and semantic consistency compared to directly inserting synonym words~\cite{garg2020bae}. 
Feng et al.~\cite{feng2018pathologies} proposed input reduction, which involves the iterative removal of the least significant words from the clean text. 
% This method was employed to demonstrate that specific keywords play a critical role in the prediction of the language models.

\vspace{0.03in}

\noindent\textbf{Certified Defenses against Word-level Adversarial Operations.}
These methods rely on interval bound propagation (IBP)~\cite{jia2019certified, huang2019achieving, ko2019popqorn}, zonotope abstraction~\cite{du2021cert, bonaert2021fast} or randomized smoothing~\cite{ye2020safer,zeng2021certified}. IBP~\cite{jia2019certified} and zonotope abstraction~\cite{du2021cert} are both linear relaxation methods, which calculate the lower and upper bound of the model output and then minimize the worst-case loss for certification. Ko et al.~\cite{ko2019popqorn} introduced POPQORN, which uses linear functions to bound the nonlinear activation function in RNN. On more complicated Transformer models, Bonaert et al.~\cite{bonaert2021fast} proposed DeepT to certify against synonym substitution operations based on multi-norm zonotope abstract interpretation. However, IBP and zonotope abstraction methods are not scalable, and few can tightly certify large-scale pre-trained models, such as BERT. 
To address this challenge, Ye et al.~\cite{ye2020safer} proposed SAFER, which leverages randomized smoothing to provide $\ell_0$ certified robustness against synonym substitutions. Zhao et al.~\cite{zhao2022certified} proposed CISS, which combines the IBP encoder and randomized smoothing to guarantee $\ell_2$ robustness against word substitution attacks. However, since CCIS maps the input into a semantic space, only the certified radius in the semantic space is available, not the certified radius in the practical word space. % CCIS improves the certified accuracy compared to SAFER. 
Zeng et al.~\cite{zeng2021certified} proposed RanMASK which provides $\ell_0$ certification against random word substitution. 
However, the exact $\ell_0$ radius of RanMASK is impractical to compute, requiring search traversal. Additionally, such methods cannot certify against universal adversarial operations. 

% Although CISS evaluated the empirical accuracy against other adversarial attacks, however, this method can not provide guarantees against other operations universally. 

% IBP: This technique can be used to verify the robustness of a neural network against adversarial attacks by ensuring that the output remains within a certain range for all possible inputs. 
% IBP only works for certifying neural networks with continuous inputs

% WordDP~\cite{wangw2021certified}
% manifold-based defenses: empirical

\vspace{0.03in}

\noindent\textbf{Randomized Smoothing against Semantic Attacks.}
Semantic attacks manipulate inputs through semantic transformations, such as image rotation and blurring, to mislead the models. To mitigate them, some randomized smoothing methods have been proposed by sampling random noise from diverse distributions. For instance, Li et al. proposed TSS~\cite{li2021tss} to use Gaussian, uniform, and Laplace distributions to certify against general semantic transformations. DeformRS~\cite{alfarra2022deformrs} and GSmooth~\cite{hao2022gsmooth} certify image semantic transformations like translation, scaling, and steering. Liu et al.~\cite{liu2021pointguard} proposed PointGuard to certify against point modification, addition, and deletion via uniform distribution. Perez et al.~\cite{perez20223deformrs} proposed 3DeformRS, for probabilistic certification of point cloud DNNs against point semantic transformations. Finally, Bojchevski et al.~\cite{bojchevski2020efficient} and Wang et al.~\cite{wang2021certified} used the Binomial distribution to certify the graph neural networks against discrete structure perturbations. 
 Notice that, Huang et al.~\cite{huang2023rs} concurrently proposed RS-Del, a method of randomized deletion smoothing that offers certified edit distance robustness for discrete sequence classifiers, such as malware detection. This approach can be resistant to adversarial manipulations (e.g., deletions, insertions) within byte sequences. 
% However, previous methods cannot be directly applied to the NLP domain. First, NLP involves numerous words in a larger, heterogeneous, and discrete space, which differs from image or graph spaces. Second, word reordering and insertion are new semantic transformations not typically encountered in computer vision.

% However, the direct application of continuous smoothing methods to word-level operations is not feasible due to the discrete nature of words. And discrete smoothing methods for graph data~\cite{bojchevski2020efficient,wang2021certified} are also not applicable due to the significantly larger word space compared to the (binary) graph space. 
% Furthermore, previous research on word-level operations solely offers robustness guarantees in the $\ell_0$ norm concerning synonym substitution operations\cite{ye2020safer,wang2021certified}, while disregarding the other three types of operations. Therefore, it is imperative to develop a general randomized smoothing framework for fundamental word-level operations.

% \input{8discussion}

\vspace{-0.09in}
\section{Conclusion}
\vspace{-0.1in}

In this paper, we present a generalized framework Text-CRS for certifying model robustness against word-level adversarial attacks. Specifically, we propose four randomized smoothing methods that utilize appropriate noise to align with four fundamental adversarial operations, including one that is applicable to all operations. In addition, we propose an enhanced training toolkit to further improve the certified accuracy. We conduct extensive evaluations of our methods, considering both adversarial operations and real-world adversarial attacks on diverse datasets and models. The results demonstrate that our method outperforms SOTA methods in their settings (substitution) and establishes new benchmarks of certified accuracy for the other three operations.

\vspace{-0.09in}
\section*{Acknowledgments}
\vspace{-0.1in}

This work is partially supported by the National Key R\&D Program of China (2020AAA0107700), the National Natural Science Foundation of China (62172359), the Hangzhou Leading Innovation and Entrepreneurship Team (TD2020003), the National Science Foundation grants (CNS-2308730, CNS-2302689, CNS-2319277, CMMI-2326341, CNS-2241713, and ECCS-2216926), and the Cisco Research Award. The authors would like to thank the anonymous reviewers for their constructive comments.

% This work is supported by National Key R&D Program of China (Grant No. 2020AAA0107700), National Natural Science Foundation of China (No. 62172359, 62032021, 61972348, 62102354, 62227805, 62072398, 61772236), Funding for Postdoctoral Scientific Research Projects in Zhejiang Province (ZJ2021139), Fundamental Research Funds for the Central Universities (No. 2021FZZX001-27), Research Institute of Cyberspace Governance in Zhejiang University, National Key Laboratory of Science and Technology on Information System Security (6142111210301), and State Key Laboratory of Mathematical Engineering and Advanced Computing.

% \section*{Acknowledgment}

% % The preferred spelling of the word ``acknowledgment'' in America is without 
% % an ``e'' after the ``g''. Avoid the stilted expression ``one of us (R. B. 
% % G.) thanks $\ldots$''. Instead, try ``R. B. G. thanks$\ldots$''. Put sponsor 
% % acknowledgments in the unnumbered footnote on the first page.

% This work is partially supported by the NSFC grants (2030 xxxxxx), and the NSF grants (CNS-2308730 and CNS-2302689), as well as the Cisco Research Award. Finally, the authors would like to thank the anonymous reviewers for their constructive comments.

\bibliographystyle{IEEEtran}
\bibliography{IEEEabrv_full,alan}

% \section*{Appendix}
% https://www.derivative-calculator.net/#
\appendices

% \BW{Do we need to restate the Definition 2 and Proposition 1?} \xy{I have deleted Definition 2. Proposition 1 needs to be revised.}

% \section{Certification for Combination of Input Space} \label{appendix:staircase}
% \noindent\textbf{Definition 2} (Staircase Mechanism). (restated~\cite{geng2015staircase}) Given $\gamma \in[0,1]$ and sensitivity $\epsilon>0$, we define the Staircase distribution with the probability density function $f_\gamma^\epsilon(\cdot)$ as
% \begin{equation}
%     f_\gamma^\epsilon(x\ |\ \mu,\Delta)=\left\{
%     \begin{aligned} 
%     &\exp({-k\epsilon})a(\gamma), \\  
%     &\qquad \|x-\mu\|_1\in[k\Delta,(k+\gamma)\Delta] \\ 
%     &\exp({-(k+1)\epsilon})a(\gamma), \\  
%     &\qquad \|x-\mu\|_1\in[(k+\gamma)\Delta,(k+1)\Delta] 
%     \end{aligned} \right.
% \label{eq:f_gamma_ori}
% \nonumber
% \end{equation}
% %
% where the normalization factor $a(\gamma)$ is utilized to ensure that $\int_{\BR} f_\gamma^\epsilon(x)\rd x=1$. Specifically, we define $b\triangleq \exp({-\epsilon})$, and $c_k \triangleq \sum_{i=0}^{\infty} i^k b^i, \forall k\in \BZ$. 
% Then the closed-form expression for $a(\gamma)$ is
% $$a(\gamma)\triangleq \dfrac{d!}{2^d\gamma^d \sum_{k=1}^d [\tbinom{d}{k}c_{d-k}(b+(1-b)\gamma^k)]}$$

\section{Proofs}

\vspace{-3mm}
\subsection{Proof for Theorem~\ref{thm:wss_1}}  \label{proof:wss_1}
\vspace{-3mm}
% We first state below the Neyman-Pearson Lemma~\cite{neyman1933ix}.
% \begin{lemma}[Neyman-Pearson] %  3 in \cite{cohen2019certified} 
% \label{lem:np}
% Let $W$ and $V$ be random variables in $\BR^{n}$ with densities $\mu_V$ and $\mu_W$. Let $h: \BR^{n} \to \{0,1\}$ be a random or deterministic function. Then:
% \begin{enumerate}[1)]
%     \item If $Q\!=\!\{z\in\BR^{n}\!:\! \frac{\mu_V}{\mu_W} \leq t\}$ for some $t>0$ and $\BP(h(X)=1) \geq \BP(X\in S)$ then $\BP(h(Y)=1) \geq \BP(Y\in S)$.
%     \item  If $Q\!=\!\{z\in\BR^{n}\!:\! \frac{\mu_V}{\mu_W} \geq t\}$ for some $t>0$ and $\BP(h(X)=1) \leq \BP(X\in S)$ then $\BP(h(Y)=1) \leq \BP(Y\in S)$.
% \end{enumerate}
% \end{lemma}

We first introduce the special case of Neyman-Pearson Lemma under isotropic Staircase Mechanisms.

\begin{lemma}[Neyman-Pearson for Staircase Mechanism under Different Means]
\label{lem:np_staircase}
Let $W\sim \CS_\gamma^\epsilon(\tau, \Delta)$ and $V\sim \CS_\gamma^\epsilon(\tau+\delta, \Delta)$ be two random synonym indexes. Let $h: \BR^{n} \to \{0,1\}$ be any deterministic or random function. Then:
\begin{enumerate}[label=\arabic*)]
    \item If $Q=\{z\in\BR^{n}: \|z\|_1-\|z-\delta\|_1 \leq \beta\}$ for some $\beta$ and $\BP(h(W)=1) \geq \BP(W\in Q)$ then $\BP(h(V)=1) \geq \BP(V\in Q)$
    \item If $Q=\{z\in\BR^{n}: \|z\|_1-\|z-\delta\|_1 \geq \beta\}$ for some $\beta$ and $\BP(h(W)=1) \leq \BP(W\in Q)$ then $\BP(h(V)=1) \leq \BP(V\in Q)$
\end{enumerate}
\end{lemma}

\begin{proof}
In order to prove
\begin{align}
\small
\begin{aligned}
    &\{z: \|z\|_1-\|z-\delta\|_1 \leq \beta\} \Longleftrightarrow \{z: \frac{\mu_V}{\mu_W} \leq t\} \textrm{\ and\ } \\
    &\{z: \|z\|_1-\|z-\delta\|_1 \geq \beta\} \Longleftrightarrow \{z: \frac{\mu_V}{\mu_W} \geq t\},
\end{aligned}
\end{align}
we need to show that for any $\beta$, there is some $t>0$, and for any $t>0$, there is also some $\beta$. 

When $W$ and $V$ are under the isotropic Staircase probability distributions, the likelihood ratio turns out to be
\begin{small}
\begin{align}
% \begin{aligned}
    \label{eq:stair_ratio_1} \frac{\mu_V}{\mu_W}=&\ \frac{\exp({-l_\Delta(z_0\ |\  \tau+\delta)\epsilon})a(\gamma)}{\exp({-l_\Delta(z_0\ |\  \tau)\epsilon})a(\gamma)} \\
    \nonumber =&\ \exp([\lfloor\frac{\|z_0-\tau\|_1}{\Delta}+(1-\gamma)\rfloor \\
    \label{eq:stair_ratio_2} &\ \qquad\qquad  -\lfloor\frac{\|z_0-(\tau+\delta)\|_1}{\Delta}+(1-\gamma)\rfloor]\epsilon) \\
    \label{eq:stair_ratio_3} =&\ \exp(\frac{\epsilon}{\Delta}[\|z\|_1-\|z-\delta\|_1]),
% \end{aligned}
\end{align}
\end{small}
where $z=z_0-\tau$. We assume that the perturbation ($\delta$) and the noise ($z$) are discrete, which means $\delta = k_1\cdot\Delta, k_1\in\BZ$ and $z=k_2\cdot\Delta, k_2\in \BZ$. Then we can derive Eq.(\ref{eq:stair_ratio_3}). %\BW{Which one?}
Therefore, given any $\beta$, we can choose $t=\exp(\frac{\epsilon \beta}{\Delta})$, and derive that $\frac{\mu_V}{\mu_W}\leq t$. Similarly, given any $t>0$, we can choose $\beta=\frac{\Delta}{\epsilon} \log t$, and derive $\|z\|_1-\|z-\delta\|_1\leq \beta$. Note that we clip case where $\mu_W= 0$.
\end{proof}

% \noindent\textbf{Theorem~\ref{thm:wss_1}.} (restated) \emph{Let $\phi_S: \CW \times \BR^n \to \CW$ be the embedding substituting transformation based on Staircase noise $\varepsilon \sim \CS_\gamma^\epsilon(w, \Delta)$ and let $g_S$ be the smoothed classifier from any deterministic or random function $h$, as in Eq.(\ref{eq:smoothg1}). Suppose $y_A, y_B\in \CY$ and $\underline{p_A}, \overline{p_B} \in [0,1]$ satisfy:}
% %
% \begin{equation}
% \begin{aligned}
%     \BP(h(u\cdot \phi_S(w,\varepsilon))=y_A) \geq \underline{p_A} &\geq \overline{p_B} \geq\\
%     \max_{y_B\neq y_A}& \BP(h(u\cdot \phi_S(w,\varepsilon))=y_B) 
% \nonumber
% \end{aligned}
% \end{equation}
% \emph{Then $g_S(u\cdot \phi_S(w, \delta_{S}\cdot \Delta))=y_A$ for all $\|\delta_{S}\|_1 \leq \rad_S$, where }
% %
% \begin{equation} \label{eq:rad_s}
% \rad_S = \max \{\frac{1}{2\epsilon}\log({\underline{p_A}}/{\overline{p_B}}), -\frac{1}{\epsilon}\log(1-\underline{p_A}+\overline{p_B})\} %
% \end{equation}

% \BW{We first show three Claims below, where are used to prove Theorem 1. Put claims here?}  
Next, we prove the Theorem \ref{thm:wss_1}.

\begin{proof}
%\BW{Define the random variable $W$ and $V$. Same as those in Lemma 2?} \xy{I have revised Lemma and Theorem in this section. ($w\to \tau$, synonym indexes)}
Denote $\delta_{{S}}=\{a_1, \cdots, a_n\}$. 
Let $\tau$ be the synonym indexes of the input $w$
%, i.e., an all-zero vector of length $n$, and $\delta_{{S}}=\{a_1, \cdots, a_n\}$. Let 
and $W\sim \CS_\gamma^\epsilon(\tau, \Delta)$ and $V\sim \CS_\gamma^\epsilon(\tau+\delta_{{S}}, \Delta)$ be random synonym indexes, as defined by Lemma~\ref{lem:np_staircase}.
The assumption is
\begin{small}
\begin{equation}
    \BP(h(W) = y_A) \geq \underline{p_A} \geq \overline{p_B} \geq \BP(h(W) = y_B)
\nonumber
\end{equation}
\end{small}
 %\BP(X \in A) > \BP(X \in B)

By the definition of $g$, we need to show that
\begin{small}
\begin{equation}
    \BP(h(V) = y_A) \geq \BP(h(V) = y_B)
\nonumber
\end{equation} %\BP(Y \in A) > \BP(Y \in B)
\end{small}

Denote $T(z) =\|z\|_1-\|z-\delta\|_1$ and use Triangle Inequality we can derive a bound for $T(x)$: 
\begin{small}
\begin{equation}
    -\|\delta\|_1 \leq T(z) \leq \|\delta\|_1
\end{equation}
\end{small}

We pick $\beta_1, \beta_2$ such that there exist the following $A, B$
\begin{equation}
\small
\begin{aligned}
    A:&=\{z: T(z)= \|z\|_1-\|z-\delta\|_1 \leq \beta_1\} \\
    B:&=\{z: T(z)= \|z\|_1-\|z-\delta\|_1 \geq \beta_2\}
\nonumber
\end{aligned}
\end{equation}
that satisfy conditions $\BP(W \in A) = \underline{p_A}$ and $\BP(W \in B) = \overline{p_B}$.
According to the assumption, we have
\begin{equation}
\small
\begin{aligned}
\BP(W \in A) = \underline{p_A} \leq p_A = \BP(h(W) = y_A) \\
\BP(W \in B) = \overline{p_B} \geq p_B = \BP(h(W) = y_B)
\nonumber
\end{aligned}
\end{equation}

Thus, by applying Lemma~\ref{lem:np_staircase}, we have
\begin{equation}
\small
\begin{aligned}
\BP(h(V) \!=\! y_A) \geq \BP(V \!\in\! A) \textrm{\ and\ }
\BP(h(V) \!=\! y_B) \leq \BP(V \!\in\! B).
\nonumber
\end{aligned}
\end{equation}

Based on our {\bf Claims} shown later, we have 
% Then we compute the following: \BW{Based on the PDF of staircase mechanism?}
\begin{small}
\begin{align}
\label{eq:YinA1}    \BP(V \in A) &\geq \exp(-\frac{\|\delta\|_1}{\Delta}\epsilon) \underline{p_A} \quad \textrm{and} \\
\label{eq:YinA2}    \BP(V \in A) &\geq 1-\exp(\frac{\|\delta\|_1}{\Delta}\epsilon)(1-\underline{p_A}) \\
\label{eq:YinB} \BP(V \in B) &\leq \exp(\frac{\|\delta\|_1}{\Delta}\epsilon)\overline{p_B}
\end{align}
\end{small}

In order to obtain $\BP(V \in A) > \BP(V \in B)$, from Eq.(\ref{eq:YinA1}) and Eq.(\ref{eq:YinB}), we need 
$\rad_S=\frac{\|\delta\|_1}{\Delta} \leq \frac{1}{2\epsilon} \log(\underline{p_A}/\overline{p_B}).$
Similarly, from Eq.(\ref{eq:YinA2}) and Eq.(\ref{eq:YinB}), we need 
$\rad_S=\frac{\|\delta\|_1}{\Delta} \leq -\frac{1}{\epsilon}\log(1-\underline{p_A}+\overline{p_B}).$
\end{proof}

\begin{claim}
\label{claim1}
\begin{small}
$\BP(V \in A) \geq \exp(-\frac{\|\delta\|_1}{\Delta}\epsilon) \underline{p_A}$
\end{small}
\begin{proof}
Recall that $\int_A \exp(-\frac{\|z\|_1}{\Delta}\epsilon)a(\gamma) \rd z=\underline{p_A}$.
\begin{small}
\begin{equation}
\begin{aligned}
    \BP(V \in A)
    =& %\int_A \exp({-\frac{\|z-\delta\|_1}{\Delta}\epsilon})a(\gamma) \rd z \\
    \int_A [\exp({-\frac{\|z\|_1}{\Delta}\epsilon}) \exp({\frac{T(z)}{\Delta}\epsilon})]a(\gamma) \rd z \\
    \geq& \exp(-\frac{\|\delta\|_1}{\Delta}\epsilon) \int_A \exp(-\frac{\|z\|_1}{\Delta}\epsilon)a(\gamma) \rd z \\
    =& \exp(-\frac{\|\delta\|_1}{\Delta} \epsilon) \underline{p_A}
\nonumber
\end{aligned}
\end{equation}
\end{small}
\vspace{-3mm}
\end{proof}
\end{claim}

\begin{claim}
\label{claim2}
\begin{small}
    $\BP(V \in A) \geq 1-\exp(\frac{\|\delta\|_1}{\Delta}\epsilon)(1-\underline{p_A})$
\end{small}
\begin{proof}
\begin{small}
\begin{equation}
\begin{aligned}
    \BP(V \in A)
    %=& \int_A \exp({-\frac{\|z-\delta\|_1}{\Delta}\epsilon})a(\gamma) \rd z \\
%    =& \int_A [\exp({-\frac{\|z\|_1}{\Delta}\epsilon}) \exp({\frac{T(z)}{\Delta}\epsilon})]a(\gamma) \rd z \\
    =& 1-\int_{\CW\backslash A} [\exp({-\frac{\|z\|_1}{\Delta}\epsilon}) \exp({\frac{T(z)}{\Delta}\epsilon})]a(\gamma) \rd z \\
    \geq& 1-\exp(\frac{\|\delta\|_1}{\Delta}\epsilon) \int_{\CW\backslash A} \exp(-\frac{\|z\|_1}{\Delta}\epsilon)a(\gamma) \rd z \\
    =& 1-\exp(\frac{\|\delta\|_1}{\Delta}\epsilon)(1-\underline{p_A})
\nonumber
\end{aligned}
\end{equation}
\end{small}
\vspace{-3mm}
\end{proof}
\end{claim}

\begin{claim}
\label{claim3}
\begin{small}
    $\BP(V \in B) \leq \exp(\frac{\|\delta\|_1}{\Delta}\epsilon )\overline{p_B}$
\end{small}
\begin{proof}
Recall that $\int_B \exp(-\frac{\|z\|_1}{\Delta}\epsilon)a(\gamma) \rd z = \overline{p_B}$.
\begin{small}
\begin{equation}
\begin{aligned}
    \BP(V \in B)
    %=& \int_B \exp({-\frac{\|z-\delta\|_1}{\Delta}\epsilon})a(\gamma) \rd z \\
    =& \int_B [\exp({-\frac{\|z\|_1}{\Delta}\epsilon}) \exp({\frac{T(z)}{\Delta}\epsilon})]a(\gamma) \rd z \\
    \leq \exp(\frac{\|\delta\|_1}{\Delta}\epsilon) &\int_B \exp(-\frac{\|z\|_1}{\Delta}\epsilon)a(\gamma) \rd z 
    =\exp(\frac{\|\delta\|_1}{\Delta} \epsilon) \overline{p_B}
\nonumber
\end{aligned}
\end{equation}
\end{small}
\vspace{-3mm}
\end{proof}
\end{claim}

% ===================================
% According to the formula for the summation of a geometric progression, we have 
% $$ F(z)=\int_{-d\Delta}^{z} f(u) \rd u=\left\{
% \begin{aligned}
% &\frac{1}{2} \frac{1-\exp(({\frac{z}{\Delta}+d)\epsilon})}{1-\exp({d\epsilon})}, \ & z\leq 0 \\
% &1-\frac{1}{2} \frac{1-\exp(({-\frac{z}{\Delta}+d)\epsilon})}{1-\exp({d\epsilon})}, \ & z>0
% \end{aligned}
% \right.
% $$ 
% where $z=z_0-w=k_2\Delta,\ k_2\in \BZ$.
% ===================================

\vspace{-3mm}
\subsection{Proof for Theorem~\ref{thm:wcr_1}} \label{proof:wcr_1}
\vspace{-3mm}
% \BW{This is the Lipschitz definition. No need to claim a proposition and no need to prove it. Making the definition as concise as possible} \xy{I searched the Lipschitz definition anew and quoted Lemma3.}

We first invoke the lemma 
% invoke Lemma 2.6 in~\cite{shalev2012online} to 
that relates the Lipschitz constant $L$ and the norm of subgradients of $g$. 
% https://homes.cs.washington.edu/~marcotcr/blog/lipschitz/
% Theorem 1 in https://www.tandfonline.com/doi/pdf/10.1080/13928619.2006.9637758

\begin{lemma}[\cite{shalev2012online}] \label{lem:l_lip_norm}
Given a norm $\|\cdot\|$ and consider a differentiable function $g: \BR^n\to \BR$.  
If $\textrm{sup}_x\|\nabla g(x)\|_* \leq L, \forall x\in \BR^n$, where $\|\cdot\|_*$ is the dual norm of $\|\cdot\|$, then $g$ is $L$-Lipschitz over $\BR^n$ with respect to $\|\cdot\|$, that is $|g(w) - g(v)| \leq L\|w - v\|$.
\end{lemma}

% \begin{proposition}
% Consider a differentiable function $g: \BR^n\to \BR$. If $\textrm{sup}_w\|\nabla g(w)\|_* \leq L$ where $\|\cdot\|_*$ has a dual norm $\|z\|=\max_w z^\top x \text{ s.t.\ } \|w\|_*\leq 1$,
% then $g$ is L-Lipschitz under norm $\|\cdot\|_*$, that is $|g(w) - g(v)| \leq L\|w - v\|$.
% \end{proposition}

% \begin{proof}
% Consider some $w,v\in \BR^n$ and a parameterization in $t$ as $\gamma(t)=(1-t)w+tv, \forall t\in[0,1]$. Note that $\gamma(0)=w$ and $\gamma(1)=v$. By the fundamental Theorem of calculus, we have:
% %
% \begin{equation}
% \begin{aligned}
%     |g(v)-g(w)|=&|g(\gamma(1))-g(\gamma(0))| 
%         =\bigg\vert\int_0^1 \frac{\rd g(\gamma(t))}{\rd t} \rd t \bigg\vert \\
%     =&\bigg\vert\int_0^1 \nabla g^\top \nabla \gamma \rd t \bigg\vert
%         \leq \int_0^1 \big\vert\nabla g^\top \nabla \gamma \big\vert \rd t \\
%     \leq& \int_0^1 \|\nabla g(x)\|_* \|\nabla \gamma(t)\|\rd t \leq L\|v-w\|
% \nonumber
% \end{aligned}
% \end{equation}
% %
% \end{proof}

Following Lemma~\ref{lem:l_lip_norm}, we show that the smoothed classifier $g(u\cdot w)$ is $L$-Lipschitz in $u$ as it satisfies $\textrm{sup}_u\|\nabla g(u\cdot w)\|_\infty \leq L$, where each $u_i$ in $u$ is a variable that follows uniform distribution and $w$ is a constant matrix. 

\begin{proposition} \label{porp:g_l_lip}
Given a uniform distribution noise $\rho\sim \FU[-\lambda, \lambda]$, the smoothed classifier $g(u,w)=\BP[h(\theta_R(u,\rho) \cdot w)]$ is $1/2\lambda$-Lipschitz in $u$ under $\|\cdot\|$ norm. %as defined in (\ref{eq:smoothg2})
\end{proposition}

\begin{proof}
It suffices to show that $\|\nabla_u g(u\cdot w)\|_\infty\leq 1/2\lambda$ to complete the proof. Here, $\theta_R(u,\rho)=u+\rho$ denotes applying uniform noise $\rho$ on $u$, i.e., randomly shuffling each row vector $u_i$ in $u$. Without loss of generality, we analyze $\partial g/\partial u_1$. Since $w$ is a fixed embedding matrix and does not affect the proof, we ignore $w$ in the following proof.
Let $u=[u_1, \hat{u}]$, where $\hat{u}=[u_2, \cdots, u_n]$, and $\rho=[\rho_1, \hat{\rho}]$, then: %\in \BR^{n\times (n-1)}
%\BW{$\hat{u} \in \mathbb{R}^{(n-1)\times d}$ or $\in \mathbb{R}^{n\times (n-1)} $} \xy{$\hat{u} \in(n-1)\times n$, because $u_i\in \BR^{1\times n}$ is a row vector}
%
% \BW{
% %I am kinda confused about the proof here. $u$ is a one-hot vector, while $\rho$ is a continuous noise. Then 
% What does $u+\rho$ mean? Or just do not care?} \xy{Thank you, I have explain $\theta(u,\rho)$ above.}
\begin{equation}
\small
\begin{aligned}
    \frac{\partial g}{\partial u_1}
    = & \frac{1}{(2\lambda)^n} \frac{\partial}{\partial u_1} \int_\Lambda \int_{-\lambda}^\lambda h( u_1+\rho_1, \hat{u}+\hat{\rho})  \textrm{d}\rho_1 \textrm{d}^{(n-1)}\hat{\rho} \\
    = & \frac{1}{(2\lambda)^n} \int_\Lambda\frac{\partial}{\partial u_1} \int_{u_1-\lambda}^{u_1+ \lambda} h(q, \hat{u}+\hat{\rho}) \textrm{d}q \textrm{d}^{(n-1)}\hat{\rho} \\
    = & \frac{1}{(2\lambda)^n} \int_\Lambda\! \Big( h(u_1+\lambda, \hat{u}+ \hat{\rho}) -h(u_1-\lambda, \hat{u}+\hat{\rho})\Big) \textrm{d}^{(n-1)}\hat{\rho}
\nonumber
\end{aligned}
\end{equation}
%
% \begin{equation}
% \begin{aligned}
%     & \frac{\partial g}{\partial u_1} \\
%     = & \frac{1}{(2\lambda)^n} \frac{\partial}{\partial u_1} \int_\Lambda \int_{-\lambda}^\lambda h((u_1+\rho_1, \hat{u}+\hat{\rho})\cdot w) \textrm{d}\rho_1 \textrm{d}^{(d-1)}\hat{\rho} \\
%     = & \frac{1}{(2\lambda)^n} \int_\Lambda\frac{\partial}{\partial u_1} \int_{-\lambda u_1}^{\lambda u_1} \frac{1}{\rho_1} h((q,\hat{u}+\hat{\rho})\cdot w) \textrm{d}q \textrm{d}^{(d-1)}\hat{\rho} \\
%     = & \frac{1}{(2\lambda)^n} \int_\Lambda \big\vert h((u_1+\lambda, \hat{u}+\hat{\rho})\cdot w) \\
%     &\qquad \qquad \qquad -h((u_1-\lambda, \hat{\rho}+ \hat{u})\cdot w)\big\vert \textrm{d}^{(d-1)}\hat{\rho}
% \nonumber
% \end{aligned}
% \end{equation}
%
where $\Lambda={[-\lambda,\lambda]^{n-1}}$. The second step follows the change of variable $q= u_1+\rho_1 \in [u_1-\lambda, u_1+\lambda]$. 
The last step follows the Leibniz rule. Let $H=\int h(q) \textrm{d}q$, then we have  $\frac{\partial H}{\partial u_1}=\frac{\partial H}{\partial q} \cdot \frac{\partial q}{\partial u_1}=h(q)$. Thus,
%
% \BW{I know this part tries to use the proof in [30], but please carefully double-check it.} \xy{I have checked it.}
%
\begin{equation}
\small
\begin{aligned}
    \big\vert \frac{\partial g}{\partial u_1} \big\vert 
    \leq & \frac{1}{(2\lambda)^n} \int_\Lambda \Big\vert h( u_1+\lambda, \hat{u}+\hat{\rho}) -h(u_1-\lambda, \hat{u}+\hat{\rho})\Big\vert \textrm{d}^{(n-1)}\hat{\rho} \\
    \leq & \frac{1}{(2\lambda)^n}\cdot (2\lambda)^{n-1}=\frac{1}{2\lambda}
\nonumber
\end{aligned}
\end{equation}
%
% \begin{equation}
% \begin{aligned}
%     \bigg\vert \frac{\partial g}{\partial u_1} \bigg\vert 
%     & \leq \frac{1}{(2\lambda)^n} \int_\Lambda \big\vert h( (u_1+\lambda, \hat{u}+\hat{\rho})\cdot w) \\
%     & \qquad \qquad \quad -h((u_1-\lambda, \hat{u}+\hat{\rho})\cdot w)\big\vert \textrm{d}^{(d-1)}\hat{\rho} \\
%     & \leq \frac{1}{(2\lambda)^n}\cdot (2\lambda)^{n-1}=\frac{1}{2\lambda}
% \nonumber
% \end{aligned}
% \end{equation}

Similarly, $\lvert \partial g/\partial u_i \rvert\leq 1/{2\lambda}$ for $\forall i\in \{2,\cdots, n\}$.
%This results in having 
Hence $\| \nabla_u g(u \cdot w) \|_\infty = \max_i |\partial g/\partial u_i| \leq 1/{2\lambda}$. 
\end{proof}

% \BW{Is Theorem 6 from Salman et al. NIPS 2019?} \xy{It is from ~\cite{alfarra2022deformrs}.}

% \BW{Is $h:\BR^n\to \BR$ or $h:\BR^n\to \BR^C$? I suggest using a different notation other than $h$ since it indicates the classification model.} \xy{Thank you, I have changed $h^i$ to $f^i$.}

Next, we show that the Lipschitz function is certifiable. %invoke Theorem~2 in~\cite{alfarra2022deformrs} to 

\begin{theorem} \label{thm:l-lipschitz}
Given a classifier $h$, let the function $f^i:\BR^n\to \BR$, defined as $f^i(x)=\BP(h(x)=i)$, be L-Lipschitz continuous under the norm $\|\cdot\|, \forall i\in \mathcal{Y}=\{1, \cdots, C\}$. If $y_A=\argmax_i f^i(x)$, then, we have $\argmax_i f^i(x+\delta)=y_A$ for all $\delta$ satisfying:
\begin{small}
\begin{equation}
\|\delta\|\leq \frac{1}{2L}(f^{y_A}(x)-\max_i f^{i\neq y_A}(x)).
\label{eq:l_lip_func}
\end{equation}
\end{small}
\end{theorem}

\begin{proof}
Take $y_B=\argmax_i h^{i\neq y_A}(x)$. Hence:
\begin{equation}
\small
\begin{aligned}
    & \big|f^{y_A}\!(x \!+\! \delta) \!-\! f^{y_A}\!(x)\big| \!\leq\! L \|\delta\| \Longrightarrow f^{y_A}\!(x \!+\! \delta) \!\geq\! f^{y_A}\!(x) \!-\! L \|\delta\| \\
    & \big|f^{y_B}\!(x \!+\! \delta) \!-\! f^{y_B}\!(x)\big| \!\leq\! L \|\delta\| \Longrightarrow f^{y_B}\!(x \!+\! \delta) \!\leq\! f^{y_B}\!(x) \!-\! L \|\delta\|
\nonumber
\end{aligned}
\end{equation}
By subtracting the inequalities and re-arranging terms, we have that as long as $f^{y_A}(x) -L\|\delta\| > f^{y_B}(x) -L\|\delta\|$, i.e., the bound in Eq.(\ref{eq:l_lip_func}), then $f^{y_A}(x+\delta) > f^{y_B}(x+\delta)$.
% , completing the proof.
\end{proof}

We now prove our Theorem~\ref{thm:wcr_1} based on the above Proposition~\ref{porp:g_l_lip} and Theorem~\ref{thm:l-lipschitz}.

% \BW{Proving Theorem 2 separately.} \xy{I have re-prove Theorem 2 as follows, but it's not very informative. Should I restate Theorem 2 here?} \BW{I suggest restating it} \xy{Thank you, Prof. Wang. I have restated and proven it.}

% \noindent\textbf{Theorem~\ref{thm:wcr_1}.} (restated)
% \emph{Let $\theta_R\!: \CU\times \BZ^n \!\to \CU$ be a permutation based on a uniform noise $\rho\!\sim\!\FU[-\lambda, \lambda]$ 
% and $g_R$ be the smoothed classifier from a base classifier $h$. Suppose $g_R$ assigns a class $y_A$ for the input $u\cdot w$, and $\underline{p_A}, \overline{p_B}\in(0,1)$. If }
% \begin{equation}
% \begin{aligned}
% \BP(h(\theta_R(u, \rho)\cdot w)=y_A) \geq \underline{p_A} &\geq \overline{p_B} \geq \\
% \max_{y_B\neq y_A}& \BP(h(\theta_R(u, \rho)\cdot w)=y_B)),
% \nonumber
% \end{aligned}
% \end{equation}
% %
% \emph{then $g_R(\theta_R(u, \delta_{R})\cdot w)=y_A$ for all permutation perturbations satisfying $\|\delta_{R}\|_1 \leq \rad_R$, where}
% \begin{equation} 
%     \rad_R= \lambda(\underline{p_A}-\overline{p_B})
% \label{eq:rad_r}
% \end{equation}

\begin{proof}
According to Proposition~\ref{porp:g_l_lip}, the uniform-based smoothed classifier $g_R$ is $1/2\lambda$-Lipschitz in $u$ under $\|\cdot\|$ norm. Combining Theorem~\ref{thm:l-lipschitz} and substituting $L=1/2\lambda$ in Eq.(\ref{eq:l_lip_func}), we have $\argmax_i \BP(g_R(\theta_R(u,\delta_R)\cdot w)=i)=y_A$ for all $\delta_R$ satisfying $\|\delta_R\|_1 \leq \& \lambda( \BP(g_R(\theta_R(u,\rho)\cdot w)={y_A}) \\
    \& \qquad -\max_{i,i\neq y_A} \BP(g_R(\theta_R(u,\rho)\cdot w)={i}))$, which holds in the case of $\|\delta_R\|_1 \leq \lambda(\underline{p_A}-\overline{p_B})$.
\end{proof}

\vspace{-3mm}
\subsection{Proof for Theorem~\ref{thm:combination}}  \label{proof:combination}
\vspace{-3mm}
% \noindent\textbf{Theorem~\ref{thm:combination}.} (restated)
% \emph{If a smoothed classifier $g$ is certified robust to permutation perturbation $\delta_u$ as defined in Eq.(\ref{eq:u_delta_1}), and to embedding permutation $\delta_w$ as defined in Eq.(\ref{eq:w_delta_1}), then $g$ can provide certified robustness to the combination of perturbations $\delta_u$ and $\delta_w$ as defined in Eq.(\ref{eq:uw_delta_1}).} assuming $\theta(u, \rho)$ is uniformly distributed in the permutation space.
% % when $u'=\theta(u, \rho)$ is uniformly distributed across the whole permutation space.
% %
% \begin{align}
%     \label{eq:u_delta_1} 
%     & \forall \delta_u, \delta_w, 
%     g(\theta(u\!+\!\delta_u,\rho)\cdot\! \phi(w, \varepsilon))\! =\! g(\theta(u,\rho)\!\cdot\! \phi(w, \varepsilon)) \\
%     \label{eq:w_delta_1} 
%     & \qquad \, \, \, \& \, 
%     g(\theta(u,\rho)\!\cdot\! \phi(w\!+\!\delta_w, \varepsilon))\! =\! g(\theta(u,\rho)\!\cdot\! \phi(w, \varepsilon)) \\
%     \label{eq:uw_delta_1} 
%     & \implies 
%     g(\theta(u\!+\!\delta_u,\rho)\!\cdot\! \phi(w\!+\!\delta_w, \varepsilon))\! =\! g(\theta(u,\rho)\!\cdot\! \phi(w, \varepsilon))
% \end{align}

\begin{proof}

% Let $\xi=u+\delta_u$, from Eq.(\ref{eq:u_delta}), we have
% \begin{equation}
%     \label{eq:covert_1} g(\theta(\xi,\rho)\cdot \phi(w,\varepsilon)) = g(\theta(u,\rho)\cdot \phi(w, \varepsilon))
% \end{equation}

% Based on $\rho\sim\FU[-\lambda, \lambda]$ and $\rho\sim \FU[-\lambda, \lambda]$, we have
% \begin{equation}
%     \label{eq:covert_2} g(\theta(\xi,\rho)\cdot \phi(w,\varepsilon)) = g(\theta(\xi,\rho)\cdot \phi(w,\varepsilon))
% \end{equation}

Because $u$ is uniformly distributed over the whole permutation space, $\theta(u,\rho)$ also is constant at arbitrary $u$, i.e., $\theta(u,\rho)=\theta(\xi,\rho)$, where $\xi$ is an arbitrary permutation. Therefore, considering Eq. \ref{eq:w_delta}, we have:
%
% \begin{equation}
%     \label{eq:covert_3} g(\theta(\xi,\rho)\cdot \phi(w+\delta_w, \varepsilon)) = g(\theta(u,\rho)\cdot \phi(w, \varepsilon))
% \end{equation}
%
% Let $\xi=u+\delta_u$, then we have
\begin{small}
\begin{align*}
    \nonumber g(\theta(u,\rho)\cdot \phi(w, \varepsilon)) =&\ g(\theta(u,\rho)\cdot \phi(w+\delta_w, \varepsilon)) \\
    \label{eq:convert_4} =&\ g(\theta(\xi,\rho)\cdot \phi(w+\delta_w, \varepsilon)) \\
    =&\ g(\theta(u+\delta_u,\rho)\cdot \phi(w+\delta_w, \varepsilon)) 
\end{align*}
\end{small}
which is 
% Eq.(\ref{eq:convert_4}) is 
exactly the Eq.(\ref{eq:uw_delta}). 
\end{proof}

\vspace{-3mm}
\subsection{Proof for Theorem~\ref{thm:wd_1}} \label{proof:wd_1}
\vspace{-3mm}

% \BW{Overall, I am not clear on how to link the two perturbations on the permutation and embedding together and derive the certified robustness.} 

% \BW{Also, Theorem 3 is said to be a universal theorem and be proved by directly applying the RS in Cohen et al. I am  unclear about that.}
% \xy{I have provided Theorem~\ref{thm:combination} for the combination of two perturbations.}

% \BW{Can shorten the proof somehow?} \xy{I have shortened the proof. Which step do you think can be further shorten?}

% We use ``1'' at different positions to indicate different words in the text. We use ``0'' to represent the deleted word, i.e., the word is converted to a \texttt{<pad>} token. 
We associate a binary variable to indicate the state (i.e., existed or deleted) of each word embedding. Initially, we have an all-ones state vector $x_0=\{1,1, \cdots, 1\}$ for all word emebeddings, indicating the existence of all words.     
%The embedding state of a text with $n$ embeddings is $x_0=\{a_1, \cdots, a_n\}=\{1, \cdots, 1\}$. 
The new embedding state of randomly deleting $\delta$ embeddings is $x_\delta=\{1, \cdots, 1, 0, \cdots, 0\}$, where the last $\delta$ states are set to be 0. % as follows
% \[
%     x_0 = \overbrace{1, \cdots, 1, 1, \cdots, 1}^{n}, \quad
%     x_\delta = \overbrace{1, \cdots, 1, \underbrace{0, \cdots, 0}_{\delta}}^{n}. 
% \]
% Our Bernoulli-based embedding deletion mechanism assumes that the probability of an embedding being deleted follows the Bernoulli distribution, i.e., $\BP(a_i=1 \to a_i=0)=p$ and $\BP(a_i=1 \to a_i=1)=1-p$. The deleted embeddings cannot be restored, i.e., $\BP(a_i=0\to a_i=1)=0$. 
By applying Bernoulli-based embedding deletion mechanism on $x_0$ and $x_\delta$, we obtain $x_{z_0}$ and $x_{z_\delta}$, respectively, as follows
% \[
%     x_{z_0} \!=\! \overbrace{1, \cdots\!, 1, \underbrace{0, \cdots, 0}_{z_0\in[0, n]}}^{n}, \
%     x_{z_\delta} \!=\! \overbrace{1, \cdots\!, 1, \underbrace{0, \cdots, 0}_{z_\delta\in[0, n-\delta]}, \underbrace{0, \cdots, 0}_{\delta}}^{n}.
% \]
\begin{small}
 \begin{equation}
\label{eq:x_z_delta}
    x_{z_0} \!=\! \underbrace{1,\cdots\!,1}_{n-z_0}, \underbrace{0,\cdots\!, 0}_{z_0\in[0, n]}; \
    x_{z_\delta} \!=\! \underbrace{1, \cdots\!, 1}_{n-\delta-z_\delta}, \! \underbrace{0, \cdots\!, 0}_{z_\delta\in[0, n-\delta]}, \! \underbrace{0,\cdots\!, 0}_{\delta}.
\end{equation}   
\end{small}

where the maximum of $z_0$ and $z_\delta$ is $n$ and $n-\delta$. \emph{For the sake of brevity, we place all ``0'' at the end of the text. Actually, ``0'' can occur anywhere in texts $x_\delta, x_{z_0}$, and $x_{z_\delta}$.} 

% \emph{Then, we assume $z\in[0,n]$ for $x_{0}$ and $z\in[\delta, n]$ for $x_{\delta}$.} Here $z$ represents the noise level that can be added to $x_0$ (or $x_\delta$), i.e., $x_0$ (or $x_\delta$) can transform $z\in[0, n]$ (or $z\in[\delta,n]$) words into \texttt{<pad>}s and generate $x_z$. The $x_z$ can be represented as

% Let noise be added on $x_0$ satisfies $\FB(n,p)$ and added on $x_\delta$ satisfies $\FB(n-\delta, p)$, where $n$ and $n-\delta$ limits the minimum value of noise. 
Next, we state the special case of Neyman-Pearson Lemma under isotropic Bernoulli Distributions.

\begin{lemma}[Neyman-Pearson for Bernoulli under Different Number of $1$]
\label{lem:np_bernoulli}
Let $W\sim \FB(n,p)$ and $V\sim \FB(n-\delta,p)$ be two random variables in $\{0,1\}^{n}$, denoting that at most $n$ and $n-\delta$ embedding states $b_i=1$ can be transformed into $b_i=0$ with probability $p$, respectively. Let $h: \{0,1\}^{n} \to \{0,1\}$ be any deterministic or random function. Then:
\begin{enumerate}[label=\arabic*)]
    \item If $Q=\{z\in\BZ: \tbinom{z}{\delta} \leq \beta\}$ for some $\beta$ and $\BP(h(W)=1) \geq \BP(W\in Q)$ then $\BP(h(V)=1) \geq \BP(V\in Q)$
    \item If $Q=\{z\in\BZ: \tbinom{z}{\delta} \geq \beta\}$ for some $\beta$ and $\BP(h(W)=1) \leq \BP(W\in Q)$ then $\BP(h(V)=1) \leq \BP(V\in Q)$
\end{enumerate}
\end{lemma}

\begin{proof} 
In order to prove:
\begin{small}
\begin{align}
\nonumber &\{z:\tbinom{z}{\delta} \leq \beta\} \Longleftrightarrow \{z:\frac{\mu_V}{\mu_W}\leq t\} \textrm{\ and\ } \\
\nonumber &\{z:\tbinom{z}{\delta} \geq \beta\} \Longleftrightarrow \{z:\frac{\mu_V}{\mu_W}\geq t\},
\end{align}    
\end{small}
we need to show that for any $\beta$, there is some $t$, and for any $t$ there is also some $\beta$. 
When $W$ and $V$ are under isotropic Bernoulli distributions, the likelihood ratio turns out to be
\begin{small}
\begin{equation}
\begin{aligned}
    \frac{\mu_V(x_z)}{\mu_W(x_z)} %= \frac{\BP(V=x_z)}{\BP(W=x_z)}
    =&\ \frac{\frac{\tbinom{n-\delta}{z-\delta}\cdot p^{z-\delta}\cdot (1-p)^{n-z}}{1-\sum_{z=0}^{\delta-1}\tbinom{n-\delta}{z-\delta}\cdot p^{z-\delta}\cdot (1-p)^{n-z}}}{\tbinom{n}{z}\cdot p^z\cdot (1-p)^{n-z}} \\
    =&\ \frac{\tbinom{n-\delta}{z-\delta}}{\tbinom{n}{z}}\cdot \frac{1}{p^\delta\cdot \Gamma(\delta, p)} = \frac{\tbinom{z}{\delta}}{\tbinom{n}{\delta} \cdot p^\delta\cdot \Gamma(\delta, p)}
    % = \frac{n(n-1)\cdots(n-\delta+1)}{z(z-1)\cdots(z-\delta+1)}\cdot p^\delta
\nonumber
\end{aligned}
\end{equation}
\end{small}%
where $\delta$ is the perturbation added on $x_0$, $z$ denotes the number of ``1'' transformed into ``0'', and $\Gamma(\delta, p)=1-\sum_{z=0}^{\delta-1}\tbinom{n-\delta}{z-\delta}\cdot p^{z-\delta}\cdot (1-p)^{n-z}\in(0,1)$ is a $z$-independent function. Here we divide the numerator by $\Gamma(\delta, p)$ because $z$ in $\mu_V(x_{z})$ requires $z \in [\delta, n]$, which means the number of ``0'' in $x_{z_\delta}$ in Eq.(\ref{eq:x_z_delta}) is $\geq \delta$.
Then for any $t$, we have $\beta=t/[\tbinom{n}{\delta} \cdot p^\delta \cdot \Gamma(\delta, p)]$. And for any $\beta$, we have $t=\tbinom{n}{\delta} \cdot p^\delta\cdot\Gamma(\delta, p) \cdot \beta$.
% As assumed above, $z$ of $x_0$ satisfies $z\in[0, n]$ and $z$ of $x_\delta$ satisfies $z\in[\delta, n]$. 
\end{proof}

Next, we prove Theorem \ref{thm:wd_1}.

\begin{proof} 
The assumption is 
\begin{small}
 $$
\BP(h(W) = y_A) \geq \underline{p_A} \geq \overline{p_B} \geq \BP(h(W) = y_B)
$$ %\BP(X \in A) > \BP(X \in B)   
\end{small}

By the definition of $g$, we need to show that
\begin{small}
\begin{equation}
    \BP(h(V) = y_A) \geq \BP(h(V) = y_B)
\nonumber
\end{equation} %\BP(Y \in A) > \BP(Y \in B)    
\end{small}

Denote $C(z)\!=\!\tbinom{z}{\delta}$ and according to \cite{UpperBou50:online}, we can derive: %a bound for $C(z)$:
\begin{small}
\begin{equation}
    1 \leq (\frac{z}{\delta})^\delta \leq C(z)=\tbinom{z}{\delta} \leq (\frac{ez}{\delta})^\delta %,\textrm{\ and\ } 1 \leq C(z)\leq \tbinom{n}{\delta}
\end{equation}    
\end{small}

We pick $\beta_1, \beta_2$ such that there exist the following $A,B$
\begin{equation}
\small
\begin{aligned}
    A:&=\{z: C(z)=\tbinom{z}{\delta} \leq \beta_1\} \\
    B:&=\{z: C(z)=\tbinom{z}{\delta} \geq \beta_2\}
\nonumber
\end{aligned}
\end{equation}
that satisfy conditions $\BP(h(W) = y_A)=\underline{p_A}$ and $\BP(h(W) = y_B)=\overline{p_B}$.
According to the assumption, we have 
\begin{equation}
\small
\begin{aligned}
\BP(W \in A) = \underline{p_A} \leq p_A = \BP(h(W) = y_A) \\
\BP(W \in B) = \overline{p_B} \geq p_B = \BP(h(W) = y_B) 
\nonumber
\end{aligned}
\end{equation}

Thus, by applying Lemma~\ref{lem:np_staircase}, we have
\begin{equation}
\small
\begin{aligned}
\BP(h(V) \!=\! y_A) \geq \BP(V \!\in\! A) \textrm{\ and\ }
\BP(h(V) \!=\! y_B) \leq \BP(V \!\in\! B)
\nonumber
\end{aligned}
\end{equation}

Based on our {\bf Claims} shown later, we have
\begin{small}
\begin{align}
\label{eq:VinA} \BP(V \in A) &\geq \frac{1}{\tbinom{n}{\delta}\cdot p^\delta \cdot \Gamma(\delta, p)} \cdot \underline{p_A} \\
\label{eq:VinB} \BP(V \in B) &\leq \frac{1}{\tbinom{n}{\delta}\cdot p^\delta \cdot \Gamma(\delta, p) } \cdot \tbinom{z_{\max}}{\delta}\cdot \overline{p_B}
\end{align}    
\end{small}

where $z_{\max}=\argmax z$ s.t. $\tbinom{n}{z} p^z (1-p)^{(n-z)} \leq \overline{p_B}$.

In order to obtain $\BP(V \in A) > \BP(V \in B)$, from Eq.(\ref{eq:VinA}) and Eq.(\ref{eq:VinB}), we need $\rad_D\!=\!\argmax\delta$
\begin{small}
\begin{equation}
    \textrm{\ s.t.\ } \underline{p_A} \geq \tbinom{z_{\max}}{\delta} \cdot \overline{p_B} \Longleftrightarrow \tbinom{z_{\max}}{\delta} \leq \underline{p_A}/\overline{p_B}
\nonumber
\end{equation}    
\end{small}
where $z_{\max}=\argmax z$ s.t. $\tbinom{n}{z} p^z (1-p)^{(n-z)} \leq \overline{p_B}$.
%
% \tbinom{n}{\delta}^\delta \leq \tbinom{n}{\delta} \leq \tbinom{en}{\delta}^\delta
% $\int_A \tbinom{z}{n}\cdot p^{z} \rd z = \int_A \frac{n!}{z!(n-z)!}\cdot p^{z} \rd z =\underline{p_A}$
\end{proof}

\begin{claim}
\begin{small}
$\BP(V \in A) \geq \frac{1}{\tbinom{n}{\delta}\cdot p^\delta \cdot \Gamma(\delta, p)} \cdot \underline{p_A}$    
\end{small}
\begin{proof}
Recall that $\sum_A \tbinom{n}{z} p^z (1-p)^{(n-z)}=\underline{p_A}$.
\begin{equation}
\small
\begin{aligned}
    \BP(V \in A)
    = & \sum_A \frac{1}{\Gamma(\delta, p)}\tbinom{n-\delta}{z-\delta}\cdot p^{z-\delta} (1-p)^{n-z} \\
%    = & \frac{1}{\tbinom{n}{\delta} \cdot p^\delta \cdot \Gamma(\delta, p)} \sum_A C(z) \cdot \tbinom{n}{z} p^z (1-p)^{(n-z)} \\
    \geq & \frac{1}{\tbinom{n}{\delta}\cdot p^\delta \cdot \Gamma(\delta, p) } \sum_A \tbinom{n}{z} p^z (1-p)^{(n-z)} \\
    = & \frac{1}{\tbinom{n}{\delta}\cdot p^\delta \cdot \Gamma(\delta, p)} \cdot \underline{p_A} 
\nonumber
\end{aligned}
\end{equation}
\vspace{-3mm}
\end{proof}
\end{claim}
%\BW{How does the second-to-last equation connects with $p_A$ in the last inequality?} \xy{I have rewritten the proof.} 
% \vspace{-1mm}

\begin{claim}
\begin{small}
 $\BP(V \in B) \leq \frac{1}{\tbinom{n}{\delta}\cdot p^\delta \cdot \Gamma(\delta, p) } \cdot \tbinom{z_{\max}}{\delta}\cdot \overline{p_B}$,\\ where $z_{\max}=\argmax z$ s.t. $\tbinom{n}{z} p^z (1-p)^{(n-z)} \leq \overline{p_B}$.   
\end{small}
\begin{proof}
Recall that $\sum_B \tbinom{n}{z} p^z (1-p)^{(n-z)}=\overline{p_B}$.
\begin{equation}
\small
\begin{aligned}
    \BP(V \in B)
    = & \sum_B \frac{1}{\Gamma(\delta, p)}\tbinom{n-\delta}{z-\delta} p^{z-\delta} (1-p)^{n-z} \\
  %  = & \frac{1}{\tbinom{n}{\delta} \cdot p^\delta \cdot \Gamma(\delta, p)} \sum_B C(z)\cdot \tbinom{n}{z} p^z (1-p)^{(n-z)} \\
    % \leq &\ \frac{1}{\tbinom{n}{\delta}\cdot p^\delta \cdot \Gamma(\delta, p)\cdot } \sum_B \tbinom{z}{z/2} \tbinom{n}{z} p^z (1-p)^{(n-z)} \\
    = & \frac{1}{\tbinom{n}{\delta}\cdot p^\delta \cdot \Gamma(\delta, p) } \sum_B \tbinom{z}{\delta}\cdot \tbinom{n}{z} p^z (1-p)^{(n-z)} \\
    \leq & \frac{1}{\tbinom{n}{\delta}\cdot p^\delta \cdot \Gamma(\delta, p) } \cdot \tbinom{z_{\max}}{\delta}\cdot \overline{p_B}
\nonumber
\end{aligned}
\end{equation}
where $z_{\max}=\argmax z$ s.t. $\tbinom{n}{z} p^z (1-p)^{(n-z)} \leq \overline{p_B}$.
\end{proof}
\end{claim}
\vspace{-1mm}
%\BW{Similarly, how does the second-to-last equation connects with $p_B$ in the last inequality?} \xy{I have rewritten the proof.}

% \begin{claim}
% Given $z$, $(\frac{ez}{\delta})^\delta \leq e^z$
% \begin{proof}
% \begin{equation}
% \begin{aligned}
% \textrm{Given\ } z, \textrm{\ when\ } \delta\leq z, \frac{\partial(\frac{ez}{\delta})^\delta}{\partial \delta} = \ln\left(\dfrac{z}{\delta}\right)\cdot\left(\dfrac{z}{\delta}\right)^\delta\,\textrm{e}^\delta \geq 0
% \nonumber
% \end{aligned}
% \end{equation}
% Then $(\frac{ez}{\delta})^\delta_{\max}=e^z$

% \begin{equation}
% \begin{aligned}
% \textrm{Given\ } \delta, \frac{\partial(\frac{e z}{\delta})^\delta}{\partial z} = \dfrac{d e^d\cdot\left(\frac{z}{d}\right)^d}{z} \geq 0, \textrm{\ then\ } (\frac{ez}{\delta})^\delta_{\max}=e^z.
% \nonumber
% \end{aligned}
% \end{equation}

% \end{proof}
% \end{claim}

%===============================================================================

% \vspace{-3mm}
\section{Certified Inference Algorithm} \label{appendix:alg}
% Prediction and Certification Algorithm
\vspace{-6mm}

\begin{figure}[H]
    \centering
    \begin{subfigure}{0.49\columnwidth}
        \includegraphics[width=\textwidth]{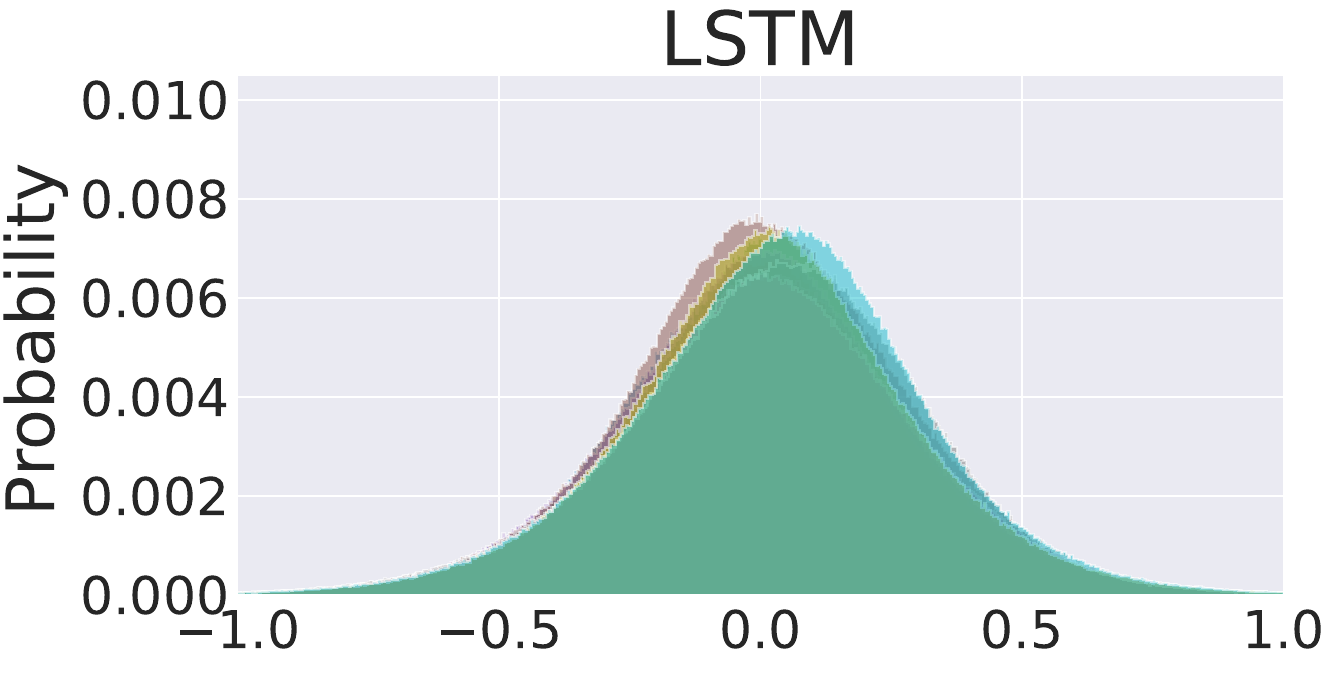}
    \end{subfigure}
    \begin{subfigure}{0.49\columnwidth}
        \includegraphics[width=\textwidth]{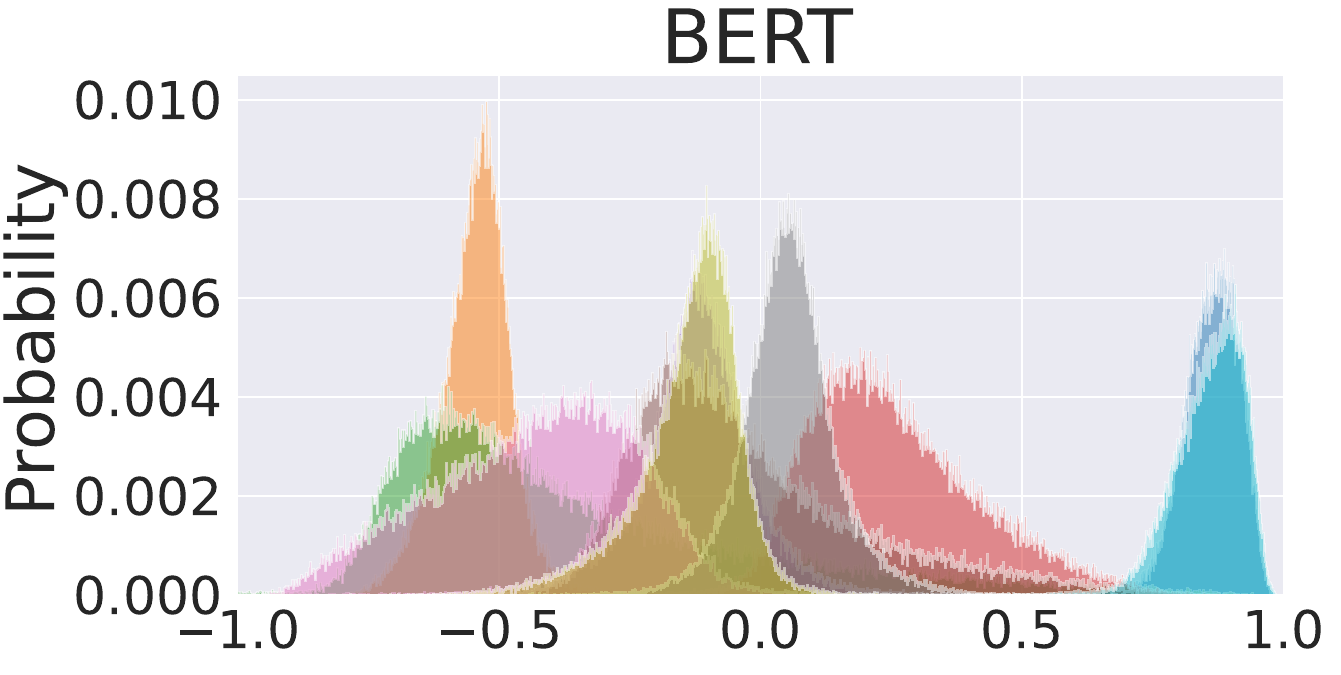}
    \end{subfigure}
    \vspace{-5mm}
    \caption{For LSTM and BERT, the distributions of the embedding elements for all words in ten randomly chosen dimensions. These distributions approximate Gaussian distributions with different means.}
    \label{fig:emb_values}
    \vspace{-4mm}
\end{figure}

% \begin{algorithm}[h]
% \small
% \caption{Prediction algorithm} \label{alg:predict}
% \begin{algorithmic}[1]
% \Require Test sample $x$, $h$, $T$, $L_{emb}$, $\theta_T(\cdot, \rho)$, $\phi_T(\cdot, \varepsilon)$, $N$
% % \Ensure $y = x^N$
% \State $u\cdot w\gets L_{emb}(x)$
% \State $\texttt{counts} \gets \textsc{SampleUnderNoise}(h, \theta(u, \rho), \phi(w, \varepsilon), N)$
% \State $y_A, y_B \gets$ top two indices in \texttt{counts}
% \State $n_A, n_B \gets \texttt{counts}[y_A], \texttt{counts}[y_B]$
% % \If{$\textsc{BinomPValue}(n_A, n_A+n_B, 0.5)\leq \alpha$}
% %     \State \Return prediction $y_A$
% % \Else
% %     \State \Return ABSTAIN
% % \EndIf
% \State \textbf{if} $\textsc{BinomPValue}(n_A, n_A+n_B, 0.5)\leq \alpha$ \textbf{return} $y_A$
% \State \textbf{else} \textbf{return} ABSTAIN
% \end{algorithmic}
% \end{algorithm}

\vspace{-1mm}
\begin{algorithm}[H]
\small
\caption{Certified inference algorithm} \label{alg:certify}
\begin{algorithmic}[1]
\Require Test sample $x$, $h$, $T$, $L_{emb}$, $\theta_T(\cdot, \rho)$, $\phi_T(\cdot, \varepsilon)$, $N$, $N_0$
\State $u\cdot w\gets L_{emb}(x)$
\State $\texttt{counts0} \gets \textsc{SampleUnderNoise}(h, \theta(u, \rho), \phi(w, \varepsilon), N_0)$
\State $y_A \gets$ top index in \texttt{counts0}
\State $\texttt{counts} \gets \textsc{SampleUnderNoise}(h, \theta(u, \rho), \phi(w, \varepsilon), N)$
\State $\underline{p_A} \gets \textsc{LowerConfBound}(\texttt{counts}[p_A], N, 1-\alpha)$
% \If{$\underline{p_A} > \frac{1}{2}$}
%     \State \Return prediction $y_A$ and radius $\rad_T$
% \Else
%     \State \Return ABSTAIN
% \EndIf
\State \textbf{if} $\underline{p_A} > \frac{1}{2}$ \textbf{return} prediction $y_A$ and radius $\rad_T$
\State \textbf{else} \textbf{return} ABSTAIN
\end{algorithmic}
\end{algorithm}

% \vspace{-5mm}

% \vspace{-3mm}
\section{Additional Experimental Results}
\vspace{-3mm}

Table~\ref{tab:clean_acc} shows that Text-CRS can train robust models with little sacrifice in performance. Table~\ref{tab:attack_acc} shows that all attacks result in a significant reduction in model accuracy.

\vspace{0.02in}
\noindent\textbf{Ablation Study for Enhanced Training Toolkit.}
We compare the accuracy of Text-CRS against word insertions, with and without the training toolkit. For LSTM, both OGN and ESR are added, improving the average clean accuracy from 79.2\% to 84.1\%, see Figure~\ref{fig:noise3_enhance_lstm_acc} (left). Figure~\ref{fig:noise3_enhance_lstm_acc} (right) shows a significant increase in certified accuracy against the SynonymInsert, particularly in Amazon. %, which shows a 14.2\% improvement on $\sigma=0.1$. 
We also evaluate the certified accuracy under different radii, see Figure~\ref{fig:noise3_enhance_lstm_radius}.
% (refer to Figure~\ref{fig:noise3_enhance_lstm_other_radius} for $\sigma=0.2$ and $0.3$). 
The results show that the training toolkit lead to higher certified accuracy at the same radius. For BERT, acceptable accuracy is achieved for low and medium noise levels. %due to its pre-training. 
However, the accuracy drops to 50\% for high noise levels, indicating a failure to identify the gradient update direction. Thus, we employ PLM to guide the training, increasing the accuracy to 84\%, see the last row, column 8 of Table~\ref{tab:certi_acc_benchmark}. 
% \hl{briefly discussing Figure 13 and 14 here (within 18 pages)}

% \vspace{-3mm}
% \vspace{-2mm}
\begin{figure}[!h] %
    \centering 
    \begin{subfigure}[b]{0.49\columnwidth}
        \centering
        \includegraphics[width=\textwidth]{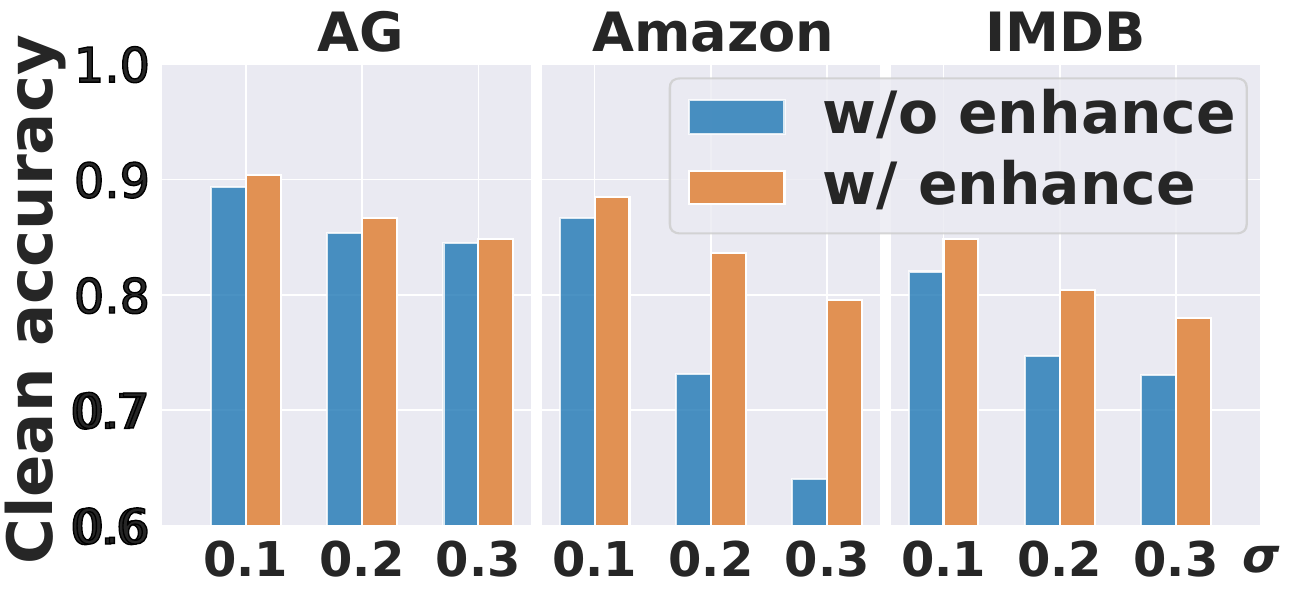}
    \end{subfigure}
    \begin{subfigure}[b]{0.49\columnwidth}
        \centering
        \includegraphics[width=\linewidth]{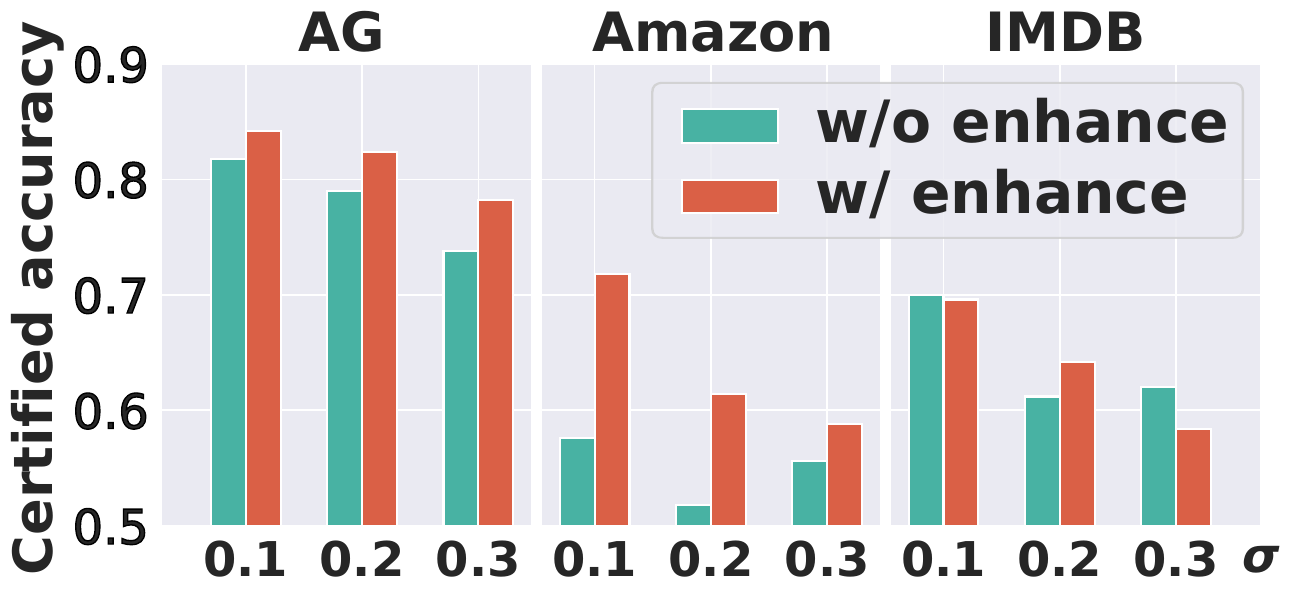}
    \end{subfigure}   
    \vspace{-5mm}
    \caption{Accuracy w/ and w/o enhanced training on different $\sigma$ under LSTM. Left: clean accuracy. Right: certified accuracy against the SynonymInsert attack.}\vspace{-0.2in}
    \label{fig:noise3_enhance_lstm_acc}
\end{figure}

% \vspace{-3mm}
\vspace{-1mm}
\begin{figure}[!h]
    \centering
    \begin{subfigure}[b]{0.32\columnwidth}
        \centering
        \includegraphics[width=\textwidth]{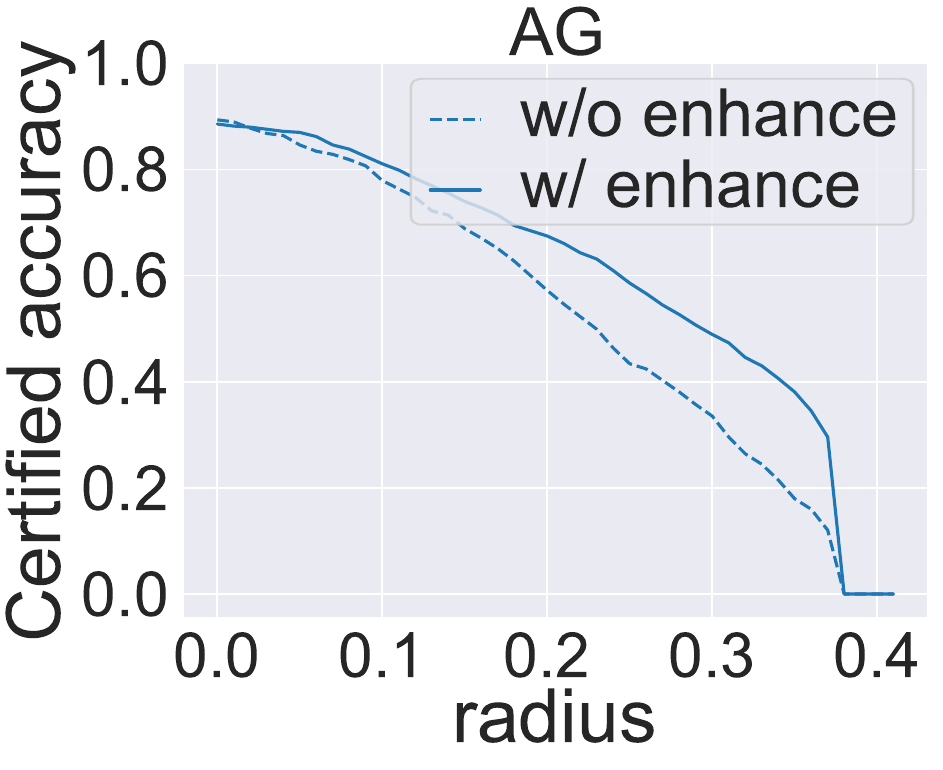}
    \end{subfigure}
    \begin{subfigure}[b]{0.32\columnwidth}
        \centering
        \includegraphics[width=\linewidth]{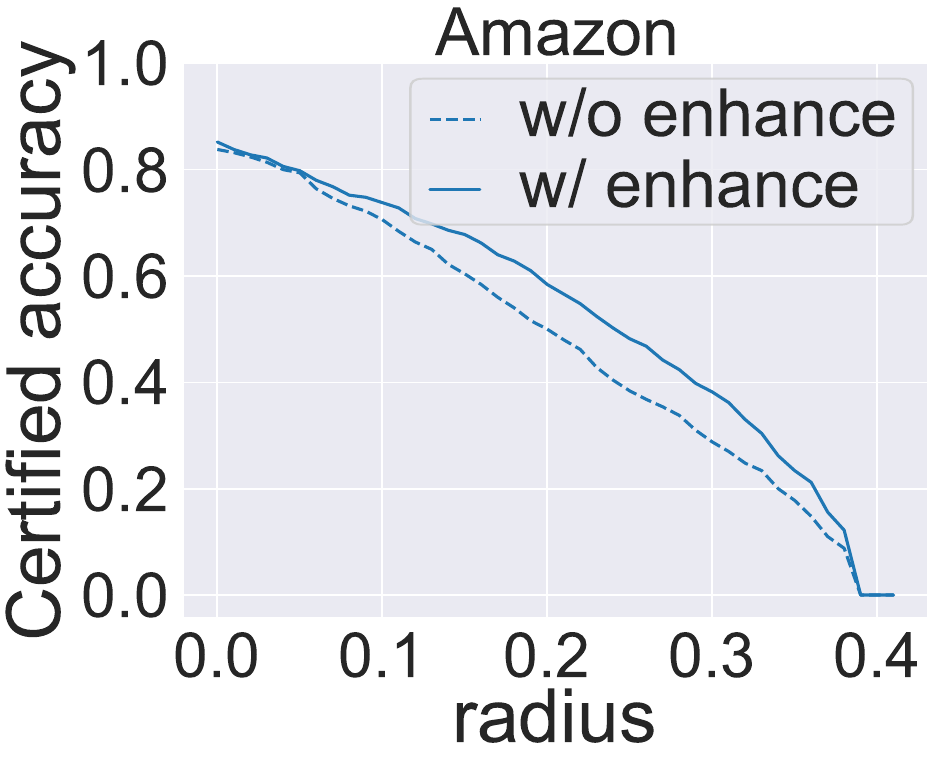}
    \end{subfigure}   
    \begin{subfigure}[b]{0.32\columnwidth}
        \centering
        \includegraphics[width=\linewidth]{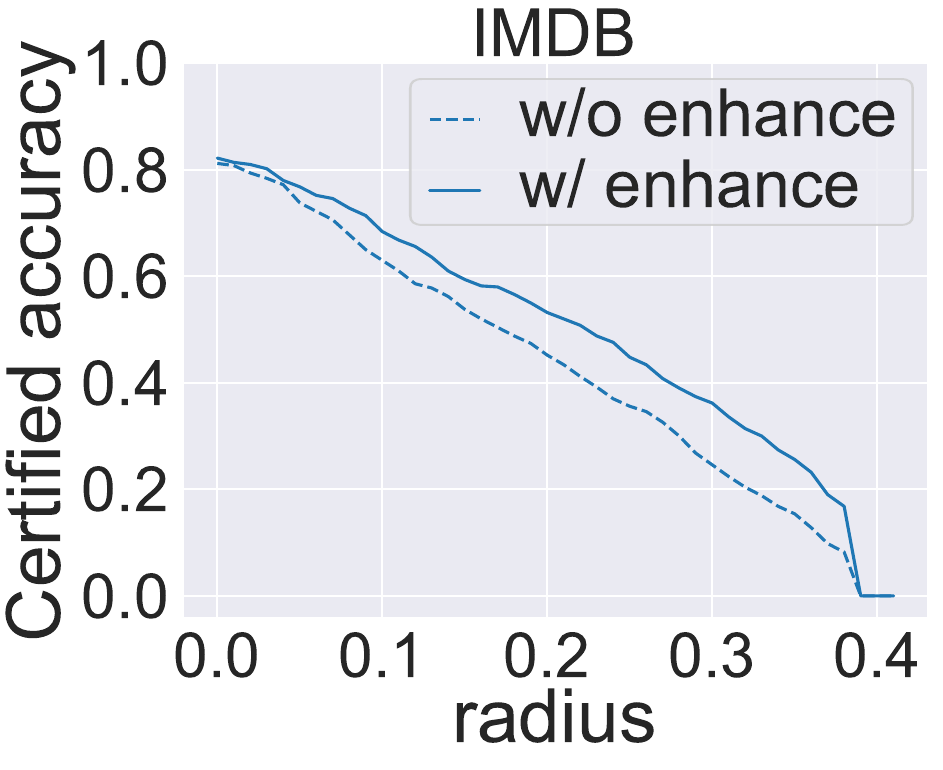}
    \end{subfigure}

    % \begin{subfigure}[b]{0.32\columnwidth}
    %     \centering
    %     \includegraphics[width=\textwidth]{figures/enhance_result/lstm_agnews_0.2.pdf}
    % \end{subfigure}
    % \begin{subfigure}[b]{0.32\columnwidth}
    %     \centering
    %     \includegraphics[width=\linewidth]{figures/enhance_result/lstm_amazon_0.2.pdf}
    % \end{subfigure}   
    % \begin{subfigure}[b]{0.32\columnwidth}
    %     \centering
    %     \includegraphics[width=\linewidth]{figures/enhance_result/lstm_imdb_0.2.pdf}
    % \end{subfigure}

    % \begin{subfigure}[b]{0.32\columnwidth}
    %     \centering
    %     \includegraphics[width=\textwidth]{figures/enhance_result/lstm_agnews_0.3.pdf}
    % \end{subfigure}
    % \begin{subfigure}[b]{0.32\columnwidth}
    %     \centering
    %     \includegraphics[width=\linewidth]{figures/enhance_result/lstm_amazon_0.3.pdf}
    % \end{subfigure}   
    % \begin{subfigure}[b]{0.32\columnwidth}
    %     \centering
    %     \includegraphics[width=\linewidth]{figures/enhance_result/lstm_imdb_0.3.pdf}
    % \end{subfigure}
    \vspace{-2mm}
    \caption{Certified accuracy at different radii w/ and w/o enhanced training on the noise level $\sigma=0.1$ under LSTM.}\vspace{-3mm} %  Dataset from top to bottom: AG, Amazon, and IMDB. The noise level from top to bottom: $\sigma=0.1$, $\sigma=0.2$, and $\sigma=0.3$.
  %  \vspace{-1mm}
    \label{fig:noise3_enhance_lstm_radius}
\end{figure}

% \vspace{-2mm}
\begin{table}[!h]
\centering
\setlength\tabcolsep{4pt}
\scriptsize
\caption{Clean model accuracy under vanilla training (\emph{Clean vanilla}) and under robust training (\emph{Clean Acc.}).}
\vspace{-2mm}
\begin{tabular}{ccccccc}
\toprule
\multirow{2}{*}{\begin{tabular}[c]{@{}c@{}}Dataset\\ (Model) \end{tabular}} & \multirow{2}{*}{\begin{tabular}[c]{@{}c@{}} \emph{Clean}\\ \emph{vanilla} \end{tabular}} & \multirow{2}{*}{Noise} & \multicolumn{4}{c}{\emph{Clean Acc.}} \\
\cmidrule{4-7}
 &  &  & Substitution & Reordering & Insertion & Deletion \\
 \midrule
\multirow{3}{*}{\begin{tabular}[c]{@{}c@{}}AG\\      (LSTM)\end{tabular}} &  & Low & \textbf{90.12\%} & \textbf{91.67\%} & 90.38\% & \textbf{91.53\%} \\
 & 91.79\% & Med. & 89.59\% & 91.21\% & \textbf{86.66\%} & 91.01\% \\
 &  & High & 88.82\% & 91.40\% & 84.83\% & 90.38\% \\
  \midrule
\multirow{3}{*}{\begin{tabular}[c]{@{}c@{}}AG\\      (BERT)\end{tabular}} &  & Low & \textbf{93.24\%} & \textbf{93.62\%} & \textbf{93.43\%} & 93.70\% \\
 & 93.68\% & Med. & 92.75\% & 93.55\% & 93.05\% & \textbf{93.71\%} \\
 &  & High & 92.63\% & 93.43\% & 91.78\% & 93.51\% \\
  \midrule
\multirow{3}{*}{\begin{tabular}[c]{@{}c@{}}Amazon\\      (LSTM)\end{tabular}} &  & Low & \textbf{87.86\%} & \textbf{88.59\%} & \textbf{88.51\%} & \textbf{88.71\%} \\
 & 89.82\% & Med. & 87.29\% & 88.44\% & 83.61\% & 88.62\% \\
 &  & High & 86.23\% & 88.36\% & 79.53\% & 88.10\% \\
  \midrule
\multirow{3}{*}{\begin{tabular}[c]{@{}c@{}}Amazon\\      (BERT)\end{tabular}} &  & Low & \textbf{93.91\%} & 94.11\% & \textbf{94.64\%} & \textbf{94.73\%} \\
 & 94.35\% & Med. & 93.07\% & 94.27\% & 94.43\% & 94.49\% \\
 &  & High & 91.62\% & \textbf{94.30\%} & 92.89\% & 93.84\% \\
  \midrule
\multirow{3}{*}{\begin{tabular}[c]{@{}c@{}}IMDB\\      (LSTM)\end{tabular}} &  & Low & \textbf{83.39\%} & \textbf{86.85\%} & \textbf{84.86\%} & \textbf{86.33\%} \\
 & 86.17\% & Med. & 82.58\% & 86.08\% & 80.45\% & \textbf{86.32\%} \\
 &  & High & 81.07\% & 86.11\% & 77.96\% & 84.76\% \\
  \midrule
\multirow{3}{*}{\begin{tabular}[c]{@{}c@{}}IMDB\\      (BERT)\end{tabular}} &  & Low & \textbf{91.52\%} & \textbf{92.08\%} & \textbf{91.88\%} & \textbf{92.46\%} \\
 & 91.52\% & Med. & 90.42\% & 91.92\% & 91.68\% & 92.17\% \\
 &  & High & 88.68\% & 91.99\% & 87.49\% & 90.55\% \\
 \bottomrule
\end{tabular}
%\vspace{-2mm}
\label{tab:clean_acc}
\end{table}

% \vspace{-4mm}
\begin{table}[!h]
\centering
\setlength\tabcolsep{2pt}
\scriptsize
\caption{\emph{Attack accuracy} of different real-world attacks.}
\vspace{-2mm}
\begin{tabular}{cccccc}
\toprule
\begin{tabular}[c]{@{}c@{}}Dataset\\ (Model)\end{tabular} & \multicolumn{1}{c}{\begin{tabular}[c]{@{}c@{}}Text-\\ Fooler\cite{jin2020bert} \end{tabular}} & \multicolumn{1}{c}{\begin{tabular}[c]{@{}c@{}}Word\\      Reorder\cite{moradi2021evaluating}\end{tabular}} & \multicolumn{1}{c}{\begin{tabular}[c]{@{}c@{}} Synonym\\  Insert\cite{morris2020textattack}\end{tabular}} & \multicolumn{1}{c}{\begin{tabular}[c]{@{}c@{}}BAE-\\      Insert\cite{garg2020bae}\end{tabular}} & \multicolumn{1}{c}{\begin{tabular}[c]{@{}c@{}}Input\\      Reduction\cite{feng2018pathologies}\end{tabular}} \\
\midrule
AG (LSTM) & 2.46\% & 76.38\% & 70.55\% & - & 40.85\% \\
AG (BERT) & 6.67\% & 49.36\% & 75.07\% & 28.53\% & 56.59\% \\
Amazon(LSTM) & 0.09\% & 54.16\% & 61.25\% & - & 30.85\% \\
Amazon(BERT) & 15.38\% & 5.29\% & 66.39\% & 9.30\% & 41.68\% \\
IMDB (LSTM) & 0.00\% & 40.20\% & 58.12\% & - & 30.65\% \\
IMDB (BERT) & 21.81\% & 7.54\% & 57.03\% & 18.66\% & 36.03\% \\
\bottomrule
\end{tabular}
\vspace{-2mm}
\label{tab:attack_acc}
\end{table}

\newpage
\section{Meta-Review}

\vspace{-2mm}
\subsection{Summary}
\vspace{-2mm}
This paper presents the first generalized framework for certifying the robustness of text classification models against textual adversarial attacks. The authors specifically tackle the vulnerability of language models to such attacks and demonstrate the effectiveness of their framework on a range of models and attack scenarios.

\vspace{-2mm}
\subsection{Scientific Contributions}
\vspace{-2mm}
\begin{itemize}
\item Addresses a Long-Known Issue
\item Provides a Valuable Step Forward in an Established Field
\end{itemize}

\vspace{-2mm}
\subsection{Reasons for Acceptance}
\vspace{-2mm}
\begin{enumerate}
\item The paper addresses a long-known issue: The authors address that previous approaches to certified robustness against word-level attacks have been constrained to synonym substitution, while widely utilized attacks rely on word reordering, insertion, or deletion.

\item The paper provides a valuable step forward in an established field. The paper highlights the limitations of currently certified robustness and emphasizes the importance of provable robustness guarantees. To address these concerns, the authors propose a generalized framework that offers guarantees against all four classes of word-level textual adversarial attacks.

\item The paper creates a new tool to enable future science. The authors propose a training toolkit designed to enhance the robustness of language models, which has the potential to inspire and facilitate future research.
\end{enumerate}

% \vspace{-2mm}
% \subsection{Noteworthy Concerns} % Exclude if your meta-review does not have noteworthy concerns
% \vspace{-2mm}
% It would be better if the authors give a clearer description of the threat model and a detailed setup of the training toolkit used in the evaluation.

% \begin{enumerate} % Enumerate environment is not necessary if there is only one
% \item 
% \end{enumerate}

% \section{Response to the Meta-Review} % Optional

% Less than 500 words response to the meta-review. The response to the meta-review is optional. Provide a response if you disagree with the meta-review. Shepherds will only deny responses to meta-reviews if they are too long or are abusive/inappropriate.

\end{document}